%% file: main.tex
\documentclass[12pt]{article}

\input{preamble.tex}

\begin{document}

\input{titlepage.tex}


\section{Introduction}

\input{content/1-introduction.tex}

\section{Literature Review}

\input{content/2-literature-review.tex}

\section{Theory}

\input{content/3-theory.tex}

\section{Estimation for the United States}

\input{content/4-estimations.tex}

\section{The Drivers of Wealth Inequality}

\input{content/5-wealth-inequality.tex}

\section{The Taxation of Wealth}

\input{content/6-taxation.tex}

\section{Conclusion}

\input{content/7-conclusion}

{
\sloppy
\printbibliography
}

\newpage

\newrefsection

\setcounter{page}{1}
\renewcommand\thefigure{\thesection.\arabic{figure}}
\setcounter{figure}{0}  

\appendix

\input{appendix/appendix-titlepage.tex}

\section{Proofs Omitted from Main Text}

\input{appendix/1-omitted-proofs.tex}

\section{Detailed Data Construction}\label{sec:detailed-data}

\input{appendix/2-data.tex}

\section{Detailed Estimations and Results}\label{sec:detailed-estimations}

\input{appendix/3-estimations.tex}

\printbibliography[title={Appendix References}]

\end{document}

%% file: preamble.tex
\usepackage[utf8]{inputenc}

\usepackage[bitstream-charter]{mathdesign}
\usepackage[T1]{fontenc}

\usepackage{lipsum}
\usepackage[margin=1in]{geometry}
\usepackage[multiple]{footmisc}
\usepackage{setspace}
\usepackage{amsmath}
\usepackage{microtype}
\usepackage{dsfont}
\usepackage{commath}
\usepackage[hyphens]{url}
\usepackage{hyperref}
\hypersetup{hidelinks}
\usepackage[page]{appendix}
\usepackage{titling}
\usepackage{mathtools}
\usepackage{amsthm}
\usepackage{bm}
\usepackage{graphicx}
\usepackage{subcaption}
\usepackage[dvipsnames]{xcolor}
\usepackage[acronym,nomain]{glossaries}
\usepackage{tcolorbox}
\usepackage{ragged2e}
\usepackage{accents}
\usepackage{caption}
\usepackage{enumitem}
\usepackage{tikz}
\usetikzlibrary{arrows.meta,shapes.geometric,decorations.pathmorphing,patterns,patterns.meta,decorations.pathreplacing,shapes.misc}
\usepackage{etoc}
\usepackage{booktabs}
\usepackage[para,flushleft]{threeparttable}
\usepackage{multirow}

\usepackage{sectsty}
\sectionfont{\centering}

\makeatletter
\renewcommand*{\@fnsymbol}[1]{\ifcase#1\or*\else\@arabic{\numexpr#1-1\relax}\fi}
\makeatother

\newacronym{scf}{SCF}{Survey of Consumer Finances}
\newacronym{dfa}{DFA}{Distributional Financial Accounts}
\newacronym{dina}{DINA}{Distributional National Accounts}
\newacronym{sde}{SDE}{stochastic differential equation}
\newacronym{pdf}{PDF}{probability density function}
\newacronym{pde}{PDE}{partial differential equation}
\newacronym{cdf}{CDF}{cumulative distribution function}
\newacronym{ccdf}{CCDF}{complementary cumulative distribution function}
\newacronym{hjb}{HJB}{Hamilton--Jacobi--Bellman}
\newacronym{aic}{AIC}{Akaike information criterion}
\newacronym{ols}{OLS}{ordinary least squares}
\newacronym{psid}{PSID}{Panel Study of Income Dynamics}
\newacronym{irs}{IRS}{Internal Revenue Service}
\newacronym{sipp}{SIPP}{Survey of Income and Program Participation}
\newacronym{nvss}{NVSS}{National Vital Statistics System}
\newacronym{gic}{GIC}{growth incidence curve}

\definecolor{Charcoal}{RGB}{47, 72, 88}
\definecolor{QueenBlue}{RGB}{51, 101, 138}
\definecolor{SlateGray}{RGB}{91, 130, 152}
\definecolor{DarkSkyBlue}{RGB}{134, 187, 216}
\definecolor{HoneyYellow}{RGB}{246, 174, 45}
\definecolor{SafetyBlazeOrange}{RGB}{242, 100, 25}
\definecolor{DarkOliveGreen}{RGB}{71, 104, 44}
\definecolor{Mahogany}{RGB}{192, 64, 0}
\definecolor{Byzantium}{RGB}{130, 2, 99}
\definecolor{MaximumGreen}{RGB}{103, 148, 54}
\definecolor{RubyRed}{RGB}{163, 22, 33}
\definecolor{Alabaster}{RGB}{241, 242, 235}
\definecolor{SpanishGray}{RGB}{150, 154, 151}

\usepackage[backref=true, backend=biber, hyperref=true, style=authoryear, isbn=false, natbib, uniquename=false]{biblatex}
\renewbibmacro{in:}{}
\AtEveryBibitem{%
    \clearlist{language}%
    \clearfield{urlyear}%
    \clearfield{urlmonth}%
}
\AtEveryCite{\color{QueenBlue}}

\DeclareMathOperator{\ash}{asinh}
\DeclareMathOperator{\var}{Var}

\DeclareMathOperator{\expc}{\mathds{E}}
\DeclareMathOperator{\prob}{\mathds{P}}
\renewcommand{\vec}[1]{\bm{#1}}

\newtheorem{prop}{Proposition}
\newtheorem{assump}{Assumption}
\newtheorem*{thm}{Theorem}

%% file: titlepage.tex
\title{\vspace*{-1.5em}Uncovering the Dynamics \\of the Wealth Distribution}
\author{Thomas Blanchet\thanks{Paris School of Economics. An online simulator for the United States that uses the formulas developed in this paper is available at \url{https://thomasblanchet.github.io/wealth-tax/}. All the data and replication codes are available at \url{https://github.com/thomasblanchet/uncovering-wealth-dynamics}. I thank Jess Benhabib, François Bourguignon, Laurent Bach, Frank Cowell, Xavier d'Haultfœuille, Thomas Piketty, Muriel Roger, Emmanuel Saez, Gabriel Zucman, as well as numerous seminar and conference participants for helpful discussions and comments.
}}

\date{This version: \today \\ \vspace{1em}%
%
}

\maketitle
\vspace*{-1em}

\begin{abstract}
\noindent I introduce a new way of decomposing the evolution of the wealth distribution using a simple continuous time stochastic model, which separates the effects of mobility, savings, labor income, rates of return, demography, inheritance, and assortative mating. Based on two results from stochastic calculus, I show that this decomposition is nonparametrically identified and can be estimated based solely on repeated cross-sections of the data. I estimate it in the United States since 1962 using historical data on income, wealth, and demography. I find that the main drivers of the rise of the top 1\% wealth share since the 1980s have been, in decreasing level of importance, higher savings at the top, higher rates of return on wealth (essentially in the form of capital gains), and higher labor income inequality. I then use the model to study the effects of wealth taxation. I derive simple formulas for how the tax base reacts to the net-of-tax rate in the long run, which nest insights from several existing models, and can be calibrated using estimable elasticities. In the benchmark calibration, the revenue-maximizing wealth tax rate at the top is high (around 12\%), but the revenue collected from the tax is much lower than in the static case.
\end{abstract}

\newpage

%% file: content/1-introduction.tex

Wealth inequality has sharply increased in the United States. By combining income tax returns with macroeconomic balance sheets, \citet{saez_rise_2020} find that the share of wealth owned by the top 1\% has increased by more than 10~pp. since the late 1970s.\footnote{\citet{saez_rise_2020} is a revision of \citet{saez_wealth_2016}, which accounts for heterogeneous returns, as well as the more important role of private business wealth at the top found by \citet{smith_top_2022}. See also \citet{saez_trends_2020}.} \citet{kuhn_income_2020} find a similar trend using survey data (Figure~\ref{fig:wealth-shares-top1}).\footnote{Other estimates, notably \citet{smith_top_2022}, also find a similar increase of the top 1\% share, although they also find a somewhat more muted increase than \citet{saez_rise_2020} for the top $0.01\%$ and narrower top groups.}

\begin{figure}[ht]
    \centering
    \begin{minipage}{0.7\linewidth}
        \includegraphics[width=\linewidth]{graphs/wealth-shares-top1.pdf}
        \centering\footnotesize\textit{Sources:} \citet{saez_rise_2020}, \citet{kuhn_income_2020}.
    \end{minipage}
    \caption{Top 1\% Wealth Share in the United States}
    \label{fig:wealth-shares-top1}
\end{figure}

But despite this growing amount of data documenting the historical evolution of the distribution of wealth, our understanding of the drivers of this trend remains elusive. Is it solely a consequence of the rise of labor income inequality? Or is it also the result of higher rates of return? What about capital gains? Does it have anything to do with the decline of the estate tax? Does it reflect changes to the distribution of saving rates? These are basic questions, and if we had long-run longitudinal data on income and wealth, they would be fairly easy to answer. After all, they simply ask how observable phenomena affect the process of wealth accumulation in direct, mechanical ways. But because longitudinal wealth data is scarce and limited, in practice, answering them has remained a challenge. That we don't have a straightforward understanding of such proximate causes of wealth inequality is a problem in and of itself, but it also impedes our ability to answer several related questions. It makes it more difficult to adequately calibrate economic models of the wealth distribution, which are used to investigate the deeper, underlying causes of rising inequality. It limits our understanding of the effect of widely discussed policies, such as wealth and inheritance taxation.

\paragraph{Contribution}

This paper addresses this issue by introducing a new way to decompose the evolution of wealth inequality in terms of individual-level factors, and in spite of the lack of panel data. The decomposition accounts for demography, inheritance, assortative mating, labor income, rates of return, consumption, and, importantly, mobility. It does so in a way that only requires repeated cross-sections, and therefore can be applied to historical data. The approach is tractable and allows for a transparent, visual identification of the key parameters.

I estimate this decomposition in the United States since the 1960s, and in doing so, I establish the main direct drivers of the rise of wealth inequality over that period. To that end, I use historical microdata on the distribution of income and wealth \citep{saez_rise_2020}, which I combine with a large set of data from censuses, surveys, demographic tables, and macroeconomic accounts. For example, I construct measures of assortative mating as well as age-specific marriage and divorce rates since the 1960s to estimate their effect on the wealth distribution. I also microsimulate the entire demographic history of the United States since the mid-19th century to statistically reconstruct intergenerational linkages, and combine them with distributional parameters for intergenerational wealth transmission as a way to measure the effect of inheritance and of the estate tax.

Finally, I develop a way to incorporate wealth taxation within the framework of this paper, derive a simple formula for how the wealth distribution would react to any wealth tax in the long run, and use my empirical estimates to calibrate it for the United States. The framework of the paper presents several advantages here as well. A typical difficulty when analyzing wealth or capital taxation is that the standard models lead to a menagerie of ``corner solutions,'' and the associated policy recommendations are often extreme, fragile, and hard to interpret. Small changes in parameters (e.g., an intertemporal elasticity of substitution just above or below one) lead to diametrically opposite results (e.g., an optimal tax of 0\% or 100\%) \citep{straub_positive_2020}.\footnote{This is for example the case in the reanalysis of the model of \citet{judd_redistributive_1985} by \citet{straub_positive_2020} in the case with no government spending.} In light of this situation, \citet{saez_simpler_2018} have argued for a simpler framework --- one that centers the same equity-efficiency trade-off that governs the theory of optimal labor income taxation \citep{piketty_optimal_2013} and, in practice, dominates policy considerations. In this framework, the elasticity of capital supply with respect to the net-of-tax rate directly dictates the optimal level of capital taxation.\footnote{As long as this elasticity is finite, positive capital taxes are desirable. We can interpret earlier arguments that capital should never be taxed \citep{chamley_optimal_1986,judd_redistributive_1985} as the consequence of assuming an infinite elasticity \citep{piketty_theory_2013,saez_simpler_2018}.} But finding a way to model this elasticity in a realistic yet tractable way, without facing the same degeneracy issues as earlier models, remains difficult. Here, because the model is stochastic and features mobility, it organically leads to well-behaved steady-states and finite elasticities of wealth with respect to the net-of-tax rate, under a wide range of economic behaviors, and without the need to resort to \textit{ad hoc} modeling devices such as wealth in the utility function \citep[e.g.,][]{saez_simpler_2018}. Because the continuous-time formalism is highly tractable, I can obtain simple formulas that combine and extend insights from earlier studies of progressive wealth taxation \citep[e.g.][]{jakobsen_wealth_2020,saez_progressive_2019}.

\paragraph{Overview of the Methodology}

This paper uses a stochastic, continuous-time framework, and its methodology rests on two results from stochastic calculus. First, I use {\color{QueenBlue}\citeauthor{gyongy_mimicking_1986}'s~(\citeyear{gyongy_mimicking_1986})} theorem to prove that two factors drive the dynamics of the wealth distribution: the mean and the variance of savings, conditional on wealth. This remains true even in models that feature heterogeneous agents and complex shocks: {\color{QueenBlue}\citeauthor{gyongy_mimicking_1986}'s~(\citeyear{gyongy_mimicking_1986})} theorem shows that this complexity can be marginalized out, and therefore that the evolution of the wealth distribution can, under very general conditions, be summarized by a single \gls{sde}. This \gls{sde} captures the two forces that shape the distribution of wealth: the average savings of the different parts of the distribution (a.k.a. the \textit{drift}), and the mobility across the distribution (a.k.a. the \textit{diffusion}).

Having reduced the dynamic of wealth to this simple model, I apply the \citet{kolmogorov_uber_1931} forward equation, a standard result of stochastic calculus, which is known to be extremely useful in understanding how wealth distributions get their power-law shape \citep{gabaix_power_2009}. This equation relates how wealth evolves stochastically at the individual level with how the distribution of wealth evolves deterministically in the aggregate. Traditionally, this equation is used deductively: starting from a set of parameters for how individual wealth evolves, we apply it to determine how the distribution changes. This paper argues that it is possible to use that equation \textit{inductively}. That is, looking at how the distribution changes, we can infer the parameters (the mean and the variance of savings) that govern the underlying individual wealth accumulation process.

That we can infer the individual wealth accumulation process --- including a mobility parameter --- from cross-sectional data may \textit{prima facie} be counterintuitive. This paper provides a simple explanation for that fact. Basically, mobility is a force that spreads out observations and, therefore,  flattens the wealth density locally. Consequently, the changes it exerts on the distribution depend on how flat it is to begin with. In contrast, the local effect exerted by the drift (i.e., how much wealth grows at a given point \textit{on average}) is akin to a local translation and is independent of the shape of the density. This distinction is not a mere technicality: it is the central piece of the mechanism that allows steady-state distributions to emerge in the first place, and it has empirically measurable consequences. By looking at how the local evolution of the wealth density relates to its current shape, I can therefore estimate both parameters.

I estimate the full model, while also accounting for income and the role played by demography, assortative mating, and inheritance. I show that the model correctly reproduces the past evolution of wealth inequality. When I compare my findings to the (albeit scarce and limited) longitudinal survey data on wealth that exists in the United States, I find that my results are consistent. I use the model to decompose the evolution of wealth inequality and to generate simple counterfactuals showing how wealth inequality would have evolved if certain factors (e.g., labor income inequality, capital gains) had stayed at their 1980 levels.

Finally, I include wealth taxation into the model, and derive a simple formula for how the wealth distribution would react to an arbitrary wealth tax in the long run, which connects the field of capital taxation to the theories of the wealth distribution. This formula depends on three factors: mobility across the wealth distribution, tax avoidance, and savings responses. Mobility is built into the framework and calibrated using this paper's estimate. For tax avoidance and savings responses, I rely on reduced-form expressions that depend on simple, estimable behavioral elasticities, which I calibrate using quasi-experimental evidence from the literature \citep{brulhart_taxing_2016,seim_behavioral_2017,jakobsen_wealth_2020,zoutman_elasticity_2018,ring_wealth_2020,londono-velez_enforcing_2021}. A comprehensive tool to simulate the effect of arbitrary wealth taxes in the United States, using this paper's formulas, is available online at {\small \url{https://thomasblanchet.github.io/wealth-tax/}}.

\paragraph{Findings}

After applying the decomposition to the United States since 1962, I find that the main drivers of the rise of the top 1\% wealth share since the 1980s have been, in decreasing level of importance, higher savings at the top, higher rates of return on wealth (essentially in the form of capital gains), and higher labor income inequality. Other factors have played a minor role.

Notably, the wealth accumulation process that I estimate features a rather large heterogeneity of savings for a given level of wealth and, therefore, a large degree of mobility across the distribution. That amount of mobility is comparable to what I independently find in longitudinal survey data from the \gls{scf} and the \gls{psid}, and in a panel of the 400 richest Americans from Forbes. This finding cuts against the ``dynastic view'' of wealth accumulation, and suggests that existing models of the wealth distribution underestimate the extent of wealth mobility. This has several implications.

First, the large heterogeneity of savings implies that high levels of wealth inequality can be sustained even if the wealthy consume, on average, a sizeable fraction of their wealth every year. Second, high mobility offers a straightforward answer to a puzzle pointed out by \citet{gabaix_dynamics_2016}: that the usual stochastic models that can explain steady-state levels of inequality are unable to account for the speed at which inequality has increased in the United States.\footnote{\citet{gabaix_dynamics_2016} primarily focus on income inequality, but their formal findings apply to wealth as well.} But high mobility leads to faster dynamics and, therefore, can account for the dynamic of wealth we observe in practice.\footnote{This solution differs from the solutions proposed by \citet{gabaix_dynamics_2016}, which involve introducing temporary changes to the drift to accelerate dynamics at the beginning of a transition.}

The study of wealth taxation also yields new insights. First, the degree of mobility across the wealth distribution --- which explicitly appears in the formula I derive --- generates a mechanical effect that fixes many degeneracy issues from deterministic models. And it is a crucial determinant of the response of the wealth stock to a tax. Indeed, without mobility, an annual wealth tax would affect the same people repeatedly, and absent behavioral responses, the tax base would eventually go to zero. With mobility, new, previously untaxed wealth continuously enters the tax base, which avoids that phenomenon. \citet{saez_progressive_2019} study this effect, but within a very restrictive setting --- my formula extends and operationalizes their idea to more complex and realistic cases, and this affects some of the conclusions.\footnote{\citet{saez_progressive_2019} consider a wealth tax at a constant \textit{average} rate above a threshold. My formulas apply to arbitrary wealth taxes, including the more common case of a constant \textit{marginal} rate within a bracket.} In light of this result, the significant mobility that I find suggests that a wealth tax could raise more revenue, all other things being equal. Second, even in the simplest model, the reduced-form elasticity of wealth with respect to the net-of-tax rate is nonconstant --- much higher for low rates than higher ones. (This directly results from the fact that the complete formula depends on both the average and the marginal tax rate.) As a consequence, the revenue-maximizing rate for a linear wealth tax above \$50m is high (12\% in my benchmark calibration), yet the revenue raised from said tax in the long run is quite low, only about a quarter of what the tax would raise in the absence of response.\footnote{If the elasticity were constant, a 12\% revenue-maximizing rate would be associated with a revenue reduction of 60\% compared to the inelastic case, as opposed to 75\% here.} Third, my simulations suggest that the effects of an annual wealth tax differ in fundamental ways from those of an inheritance tax of seemingly comparable magnitude. I show in a simple model that this distinction is driven by the lifetime of generations: the longer people live, the less potent the estate tax compared to a wealth tax.

%% file: content/2-literature-review.tex
\subsection{Models of the Wealth Distribution}

The conventional way of studying the determinants of wealth inequality is to construct a model of an economy where people face various individual shocks. The prototypical model for this line of work comes from \citet{aiyagari_uninsured_1994} and \citet{huggett_risk-free_1993}, who study the distribution of wealth in \citet{bewley_permanent_1977} type models in which people face idiosyncratic uninsurable labor income shocks. In these models, people accumulate wealth for precautionary or consumption smoothing purposes. But these motives quickly vanish as wealth increases, so they cannot rationalize large wealth holdings. As a result, these models notoriously underestimate the real extent of wealth inequality.

The literature, then, has searched for ways to extend these models and match the data. One possibility involves additional saving motives for the rich, such as a taste for wealth \citep{carroll_why_1998} or bequests \citep{de_nardi_wealth_2004}. Another solution is to introduce idiosyncratic stochastic rates of returns in the form of entrepreneurial risk \citep{quadrini_entrepreneurship_2000,cagetti_entrepreneurship_2006,benhabib_distribution_2011}. Yet another option introduces heterogeneous shocks to the discount rate \citep{krusell_income_1998}.

Recent contributions have zoomed in on specific features and issues. Following the finding from \citet{bach_rich_2020} (in Sweden) and \citet{fagereng_heterogeneity_2020} (in Norway) that the wealthy enjoy higher returns, several papers have examined this mechanism in the United States. \citet{xavier_wealth_2021} finds that heterogeneous returns play a critical role in explaining the current levels of wealth inequality. \citet{cioffi_heterogeneous_2021} embeds a portfolio choice model in a wealth accumulation model. In this paper, because of heterogeneous portfolio composition, households have different exposures to aggregate risks. The wealthy invest more in high-risk, high-return assets like equity. As a result, \citet{cioffi_heterogeneous_2021} finds that abnormally high stock market returns have been a critical driver of the rise in wealth inequality in the United States. \citet{hubmer_comprehensive_2018}, on the other hand, use a different model and conclude that changes in tax progressivity, rather than changes in rates of return, explain the rise in wealth inequality.

In general, models of the wealth distribution seek to explain two facts: that wealth inequality is high (that is, higher than income inequality) and that its top tail is shaped like a power law. Existing models differ across many dimensions, but when it comes to explaining these facts, they usually share a common mechanism: namely, that wealth accumulates through a series of individual multiplicative random shocks, with frictions at the bottom. Assuming that the shocks have adequate properties, such models can generate both high wealth inequality and power-law tails. In discrete time, we can formulate this idea using the theory of \citet{kesten_random_1973} processes. Assume that the wealth $w_{it}$ of person $i$ evolves according to a linear recurrence equation with random coefficients: $w_{i,t+1}=a_{it} w_{it}+b_{it}$, where $a_{it}$ is a multiplicative shock (typically reflecting stochastic returns or preference shocks) and $b_{it}$ is a friction term (typically reflecting labor income). Then, under very general conditions, wealth admits a steady-state distribution with a power-law tail, whose fatness is determined by the properties of the multiplicative shock $a_{it}$ \citep{kesten_random_1973,grincevicius_one_1975,vervaat_stochastic_1979,goldie_implicit_1991,grey_regular_1994}. We can formulate a continuous-time version of this mechanism using stochastic calculus \citep{gabaix_power_2009}.\footnote{\citet{de_saporta_tail_2004} rigorously study the the continuous time version of \citet{kesten_random_1973} processes. In this paper I consider simpler --- and less literal --- continuous time analogs to \citet{kesten_random_1973} processes.} This formalism presents major advantages in terms of flexibility and tractability, notably because of the \citet{kolmogorov_uber_1931} equation, which establishes a straightforward relationship between the distribution of the individual random shocks and the evolution of the wealth distribution in the aggregate.

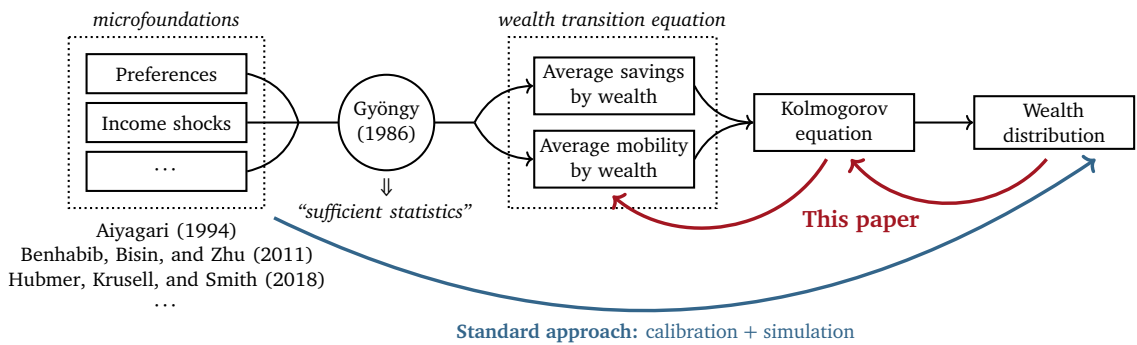
\begin{figure}[ht]
    \input{figures/models-chart.tex}
    \caption{The Different Approaches for Studying the Wealth Distribution}
    \label{fig:models-chart}
\end{figure}

We can understand the relation of the current paper to this literature with the help of Figure~\ref{fig:models-chart}. The ``standard approach,'' represented by the blue arrow at the bottom, takes a \textit{deductive} path. Starting from a set of microfoundations, it solves and simulates the model to arrive at a wealth distribution. Assuming the model matches the data, we can use it to study policies and generate counterfactuals. Implicitly or explicitly, that approach can be broken down into several steps. First, microfoundations lead to a wealth transition equation that effectively determines the dynamics of the wealth distribution. For every level of wealth, two parameters characterize this equation: the average, and the variance of wealth growth (i.e., mobility). That these two parameters are sufficient to characterize the dynamics of the distribution can be rigorously proven using {\color{QueenBlue}\citeauthor{gyongy_mimicking_1986}'s~(\citeyear{gyongy_mimicking_1986})} theorem. Then, I can feed these parameters into the \citet{kolmogorov_uber_1931} equation, which describes how the wealth density evolves.

This paper takes the reverse path, as characterized by the two red arrows in Figure~\ref{fig:models-chart}. I start from the data on the evolution of the wealth distribution and then use the \citet{kolmogorov_uber_1931} equation to \textit{infer} the parameters of the wealth transition equation. Many different microfoundations can lead to the same wealth transition equation; therefore, in my approach, the complete model of the economy remains unknown. Yet knowing just the wealth transition equation already opens many possibilities. It is sufficient to study many mechanisms, counterfactuals, and policies. And when more complete models remain necessary, the approach makes it easier to discriminate between them.

We can, in particular, compare this paper to two recent studies that also use a primarily empirical approach within the framework of stochastic wealth accumulation models. \citet{benhabib_wealth_2019} fit a model of stochastic wealth accumulation to data for the United States. They start from an explicitly microfounded model of consumption and use the method of simulated moments to identify the parameters that best replicate observed data. These parameters determine the shape of the utility function and the process for the rate of return. There are four main differences between their approach and mine. First, they estimate consumption by estimating a two-parameter utility function, while I estimate a nonparametric profile of saving rates by wealth. Their approach is more structural and tightly parameterized; my approach is more flexible and reduced form. Second, they use direct information on wealth mobility in a given year, from the 2007--2009 panel of the \gls{scf}, to estimate their model; I validate my estimate of mobility against the \gls{scf}, but I do not use it directly. Third, their model is nonlinear and estimated using the method of moments; my approach works by fitting a univariate linear equation for each level of wealth, making the source of identification highly explicit and easy to check graphically. Fourth, their primary goal is to replicate the levels rather than the trends in wealth inequality, so their main model assumes that the economy is at its steady state. In a separate exercise, they additionally replicate trends. My approach directly reproduces both levels and trends by construction --- and in fact, I use the trends as an essential source of identification.

\citet{gomez_decomposing_2022} decomposes the evolution of top wealth shares in the United States while accounting for the role of demography and mobility. Like this paper, \citet{gomez_decomposing_2022} decomposes changes in the wealth distribution that are caused by the first moment of wealth growth rates (average savings) and by the second moment (mobility).\footnote{The decomposition in \citet{gomez_decomposing_2022} also accounts for higher moments of the distribution of wealth growth rates (i.e., skewness, kurtosis, etc.), but he finds that these have a negligible impact on the wealth distribution.} Like this paper, \citet{gomez_decomposing_2022} finds an important role for mobility. His methodology, however, requires longitudinal data, for which he uses the Forbes 400 ranking of the wealthiest Americans.\footnote{\citet{gomez_decomposing_2022} also applies his methodology without direct panel data but with a separate estimate of the variance of wealth growth. However, estimating this variance still requires some form of longitudinal data.} In contrast, I use historical data on income and wealth to directly infer this parameter, allowing me to analyze the rise of wealth inequality since 1962. In contrast, the Forbes 400 rankings started in 1982. My analysis uses the entire wealth distribution, not only the top of the tail. Finally, I also account for labor income, inheritance taxation, and assortative mating. \citet{gomez_decomposing_2022} and this paper provide complementary evidence of the dynamics of wealth accumulation and find a similar result: that mobility across the wealth distribution is sizeable and plays a crucial role in shaping the wealth distribution.

\subsection{Synthetic Saving Rates}

Another way of using wealth distribution data to study the drivers of wealth inequality is to construct \textit{synthetic saving rates} for the different parts of the wealth distribution \citep[e.g.,][]{saez_wealth_2016,kuhn_income_2020,garbinti_accounting_2021}. Synthetic savings are constructed as if each bracket of the wealth distribution was a single infinitely lived individual: for example, if the average wealth of the top 1\% is \$10m at the start of the year and \$11m at the end, then the top 1\% had \$1m in synthetic savings. As such, synthetic savings are a composite measure that combines the effects of average wealth growth with the effect of mobility. They do not distinguish between them.

Since my approach explicitly accounts for mobility, it lets me disentangle the two effects. In the spirit of \citet{gomez_decomposing_2022}, I show that synthetic savings are the sum of several terms: one that depends on average wealth growth, one that depends on mobility, and additional terms that account for other factors, such as demography. This decomposition makes it possible to model synthetic savings more realistically.

Notably, I find that the way synthetic saving rates combine the effect of mobility with other factors is endogenous to the wealth distribution itself. That is, synthetic savings will differ depending on whether inequality is high or low, even if people's actual behavior is the same. To eliminate this endogeneity, we would have to assume zero mobility in the wealth distribution.

\subsection{Wealth and Capital Taxation}

A well-known, early result in the theory of optimal taxation is that, in a standard neoclassical model, the optimal linear tax rate on capital is zero, even from the perspective of workers who own no capital \citep{chamley_optimal_1986,judd_redistributive_1985}.\footnote{Another famous rationale for not taxing capital comes from \citet{atkinson_design_1976}. In their model, there is no heterogeneity of wealth conditional on income. Therefore, assuming income can be taxed, taxing wealth provides no additional equity gains while also distorting intertemporal choices. This justification for not taxing capital would not apply to my model since wealth is heterogeneous conditional on income.} This result has been overturned in more sophisticated models which introduce, for example, uncertainty \citep{aiyagari_uninsured_1994}, incomplete markets \citep{farhi_capital_2010}, heterogeneous altruism \citep{farhi_estate_2013}, tax progressivity \citep{saez_optimal_2013} or capital accumulation by workers \citep{bassetto_redistribution_2006}. More recently, \citet{straub_positive_2020} have reanalyzed the models of \citet{chamley_optimal_1986} and \citet{judd_redistributive_1985}, and have overturned many of their conclusion: they show that significant capital taxes are in fact often desirable within these two models.

Overall, the main issue with the theory of capital taxation is that its models tend to behave erratically. Many of its results focus on edge cases and corner solutions, which are highly dependent on the specific primitives of the economy, can be hard to interpret, and imply unrealistic behavior. As a result, the theory has remained of limited use for guiding policy. For example, a common interpretation of the zero-tax result of \citet{chamley_optimal_1986,judd_redistributive_1985} is that their model implicitly assumes the capital stock to be infinitely elastic.\footnote{See \citet{auerbach_taxation_2001,piketty_theory_2013,saez_simpler_2018}, although this interpretation is not universally accepted, see \citet{straub_positive_2020}.} This assumption makes capital taxation infinitely costly, so the equity gain from taxation is never worth the efficiency cost. Such a setting arguably constitutes an extreme edge case. In more realistic models featuring a finite elasticity, the usual equity-efficiency trade-off would determine the desirability of capital taxation \citep{piketty_theory_2013,saez_simpler_2018}. In line with this view, \citet{saez_simpler_2018} ultimately argued that the long-run elasticity of the capital stock to the net-of-tax rate is a sufficient statistic for the optimal design of capital taxation and developed formulas for optimal tax rates, similar to those that exist for labor income \citep{saez_using_2001,piketty_optimal_2013}.

However, the value of that elasticity remains elusive. In the \textit{short run}, several empirical papers have used quasi-experimental settings to estimate it: \citet{seim_behavioral_2017} in Sweden, \citet{londono-velez_enforcing_2021} in Colombia, \citet{brulhart_taxing_2016} in Switzerland, \citet{zoutman_elasticity_2018} in the Netherlands and \citet{jakobsen_wealth_2020} in Denmark. With some exceptions, these elasticities tend to be small. This is consistent with the view that a government trying to raise revenue with a one-off, unexpected wealth tax could choose a very high marginal rate. But the policy-relevant parameter is the \textit{long-run} elasticity, which is far more uncertain and likely larger. Indeed, the short-run elasticity only captures tax avoidance or short-run saving responses. But over time, wealth taxes also keep people from accumulating wealth, either through mechanical (lower post-tax rates of return) or behavioral effects (lower savings). This process slowly erodes the tax base. But because it takes a long time to materialize, it is hard to cleanly identify it in the data.

This paper provides a useful framework for addressing this issue. We can start from the empirically estimated wealth transition equation, which reproduces the true evolution of the wealth distribution. Then we can modify that equation to incorporate a wealth tax, as well as an arbitrary set of behavioral responses. For this paper, I focus on two effects (tax avoidance and a decrease in savings), but the framework is flexible enough to accommodate many other extensions. Using the modified equation for the evolution of wealth, we can simulate counterfactual evolutions of the wealth distribution and therefore estimate how the tax base would react to any wealth tax. We can fully simulate the model to observe transitory dynamics. However, if we focus on the long run, we can obtain simple analytical formulas that characterize the alternative steady-state. In my model, the long-run elasticity of capital supply remains finite in all cases because of mobility, but setting the mobility parameter to zero recovers the infinite elasticity as in \citet{chamley_optimal_1986} and \citet{judd_redistributive_1985}.

We can connect the approach in this paper to two recent papers that try to assess the long-run effects of wealth taxes. \citet{jakobsen_wealth_2020} use their short-run elasticity estimates to calibrate a structural model of savings at the top. They indeed find a higher elasticity in the long run. \citet{saez_progressive_2019} consider the problem of taxing the very top of the wealth distribution (billionaires) using data from the Forbes rankings. These two papers provide models that shed different lights on the problem. \citet{jakobsen_wealth_2020} model wealth accumulation using a deterministic model of intertemporal choice. This model features standard preferences over a consumption path and a taste for end-of-life wealth (i.e., bequests). They use it to derive analytical expressions linking the long-run elasticity of wealth to the short-run elasticity and preference parameters. This model emphasizes the role of behavioral savings responses, but as it is deterministic, it does not account for the role of mobility. In contrast, \citet{saez_progressive_2019} only focus on mobility. Because they look at billionaire wealth, they assume that the role of consumption is negligible. In that model, mobility is the sole determinant of wealth elasticity in the long run. If it is low, then a wealth tax ends up taxing the same people again and again: as their wealth mechanically goes down, so does the tax base. Therefore, the elasticity of taxable wealth is high, and the ability to tax wealth is limited in the long run. However, if mobility is high, the tax base often gets renewed. Individuals are subjected to the tax for shorter periods, with new, previously untaxed wealth entering the tax base regularly: as a result, the elasticity is lower.

My model combines the mechanisms from both papers and accounts for tax avoidance. Combining the mechanisms as such substantively alters some of the conclusions. In \citet{jakobsen_wealth_2020}, there is a mechanical effect of wealth taxes, and because the model is deterministic, this mechanical effect becomes infinite in the long run. Their model can generate a finite elasticity of capital supply overall, because an infinite behavioral effect compensates for the infinite mechanical effect in the opposite direction. However, looking at these effects separately in the long run still runs into degeneracy problems.\footnote{\citet{jakobsen_wealth_2020} circumvent the issue by looking at a long but finite horizon (30 years).} In comparison, because I introduce mobility, my model remains well-behaved at the steady state even if I shut down behavioral responses. In \citet{saez_progressive_2019}, the wealth tax can only consist of a constant \textit{average} tax rate over a threshold. In comparison, my model allows for arbitrary tax schedules, including the more standard case of a constant \textit{marginal} rate above a threshold. Far above the tax threshold, the marginal and the average tax rate are similar, so their model behaves similarly to the mechanical component of my mine. Close to the threshold, however, the distinction between them matters. In particular, I find that under a purely mechanical model, even confiscatory wealth tax rates would raise a non-negligible revenue from people just above the tax threshold who recently entered the tax base. And so the Laffer curve never goes to zero. As a result, unlike \citet{saez_progressive_2019}, I conclude that the mechanical model remains insufficient to characterize revenue-maximizing tax rates: doing so requires a behavioral response.

%% file: figures/models-chart.tex
\centering

\scalebox{0.65}{
\begin{tikzpicture}

\node[very thick,rectangle,draw=black,text width=3cm,minimum height=0.8cm,align=center] at (0, 0) {Preferences};
\node[very thick,rectangle,draw=black,text width=3cm,minimum height=0.8cm,align=center] at (0, -1) {Income shocks};
\node[very thick,rectangle,draw=black,text width=3cm,minimum height=0.8cm,align=center] at (0, -2) {$\cdots$};

\draw[very thick,dotted] (-2, 0.75) rectangle (2, -2.75);
\node[above,text width=3cm,align=center] at (0, 0.75) {\textit{microfoundations}};
\node[below,text width=8cm,align=center] at (0, -2.9) {\AtNextCite{\color{black}}\citet{aiyagari_uninsured_1994} \\ \AtNextCite{\color{black}}\citet{benhabib_distribution_2011} \\ \AtNextCite{\color{black}}\citet{hubmer_comprehensive_2018} \\ $\cdots$};

\node[very thick,right,rectangle,draw=black,text width=3cm,minimum height=0.8cm,align=center] at (16.5, -1) {Wealth \\ distribution};

\draw[line width=0.7mm, ->, color=QueenBlue] (2.2, -2.95) to[bend right=30] (19, -1.75); \node[color=QueenBlue] at (10, -5.3) {\textbf{Standard approach:} calibration + simulation};

\draw[very thick] (1.65, 0) to[bend left] (2.7, -1);
\draw[very thick] (1.65, -1) -- (6.3, -1);
\draw[very thick] (1.65, -2) to[bend right] (2.7, -1);

\draw[very thick,->] (6.3, -1) to[bend left] (7.5, -0.25);
\draw[very thick,->] (6.3, -1) to[bend right] (7.5, -1.75);

\node[very thick,right,rectangle,draw=black,text width=3cm,minimum height=0.8cm,align=center] at (7.5, -0.25) {Average savings by wealth};

\node[very thick,right,rectangle,draw=black,text width=3cm,minimum height=0.8cm,align=center] at (7.5, -1.75) {Average mobility by wealth};

\draw[very thick,dotted] (7, -2.75) rectangle (11.25, 0.75);
\node[above,text width=5cm,align=center] at (9.125, 0.75) {\textit{wealth transition equation}};

\node[very thick,right,circle,draw=black,align=center,fill=white] at (3.5, -1) {\AtNextCite{\color{black}}\citeauthor{gyongy_mimicking_1986} \\ (\AtNextCite{\color{black}}\citeyear{gyongy_mimicking_1986})};
\node[below,text width=5cm,align=center] at (4.5, -2) {$\Downarrow$ \\ \textit{``sufficient statistics''}};

\draw[very thick,->] (10.8, -0.25) to[bend right] (12, -1);
\draw[very thick] (10.8, -1.75) to[bend left] (12, -1);

\node[very thick,right,rectangle,draw=black,text width=3cm,minimum height=0.8cm,align=center] at (12, -1) {Kolmogorov \\ equation};

\draw[very thick,->] (15.3, -1) -- (16.5, -1);

\draw[line width=0.7mm, ->,color=RubyRed] (18, -1.75) to[bend left=50] (14, -1.75);
\draw[line width=0.7mm, ->,color=RubyRed] (13.5, -1.75) to[bend left=50] (9.125, -2.5);

\node[color=RubyRed, scale=1.2] at (14.2, -3) {\textbf{This paper}};




\end{tikzpicture}
}

%% file: content/3-theory.tex
Time is continuous, indexed by $t$. Individuals are indexed by $i$. Individual wealth $w_{it}$ evolves stochastically. I model this evolution using an Itô process, which can flexibly model most continuous-time stochastic processes. Such a process is locally characterized by two parameters: the \textit{drift} $\mu_{it}$, and the \textit{diffusion} $\sigma_{it}^2$. Over a small period of time $\dif t$, the change $\dif w_{it}$ in the value of the wealth $w_{it}$ of individual $i$ at time $t$ has mean $\mu_{it}\dif t$ and variance $\sigma_{it}^2 \dif t$. This process is represented in the form of a \gls{sde} and commonly written as follows:
\begin{equation*}
\dif w_{it} = \mu_{it} \dif t + \sigma_{it} \dif B_{it}
\end{equation*}
This section will explain how the parameters for drift and diffusion at the individual level can be connected to the aggregate evolution of the wealth distribution (while accounting for heterogeneity, demography, etc.), and therefore how they can be inferred from changes in the shape of the wealth distribution.

\subsection{Income and Consumption}

\subsubsection{Evolution of Individual Wealth}

Individual $i$ has idiosyncratic consumption, labor income, and rate of return on their wealth. Similarly to {\color{QueenBlue}\citeauthor{carroll_buffer-stock_1997}'s~(\citeyear{carroll_buffer-stock_1997})} interpretation of {\color{QueenBlue}\citeauthor{friedman_theory_1957}'s~(\citeyear{friedman_theory_1957})} permanent income hypothesis, labor income is the sum of a permanent (i.e., slow-moving) component $y_{it}$, and of transitory shocks with variance $\upsilon^2_{it}$ --- with $y_{it}$ and $\upsilon^2_{it}$ both being arbitrary bounded stochastic processes. Consumption and rates of return follow a similar model, with slow-moving components $c_{it}$ and $r_{it}$, and transitory shocks with variances $\gamma^2_{it}$ and $\phi^2_{it}$. I express all monetary quantities as a fraction of average national income, which grows at rate $g_t$. As a result, wealth $w_{it}$ evolves according to the \gls{sde}:\footnote{This formulation assumes that the transitory shocks are uncorrelated. We could account for correlated shocks by including additional covariance terms, as in Bienaymé's identity. To simplify the exposition, I focus on the uncorrelated case.}
\begin{equation}\label{eq:income-consumption}
\dif w_{it} = \underbrace{(y_{it} + (r_{it} - g_t)w_{it} - c_{it})}_{\mu_{it} \;\text{(drift)}} \dif t + \underbrace{\left(\upsilon^2_{it} + \phi^2_{it}w^2_{it} + \gamma^2_{it}\right)^{1/2}}_{\sigma_{it} \;\text{(diffusion)}} \dif B_{it}
\end{equation}
That is, individual wealth is a stochastic process with a drift $\mu_{it}$ equal to $y_{it} + (r_{it} - g_t)w_{it} - c_{it}$ and a diffusion $\sigma_{it}^2$ equal to $\upsilon^2_{it} + \phi^2_{it}w^2_{it} + \gamma^2_{it}$. Note that the economy's growth rate $g_t$ appears in the drift term due to the normalization of all the quantities by the economy's average income.

\subsubsection{Evolution of the Distribution of Wealth}

We can directly relate the evolution of individual wealth described by equation~(\ref{eq:income-consumption}) to the overall distribution of wealth. A standard result of stochastic calculus states that, if a large number of stochastic processes follow the same \gls{sde}, then the \gls{pdf} that describes the distribution of their value at a given time follows a \gls{pde} known as the \citet{kolmogorov_uber_1931} forward equation~\citep{gabaix_power_2009}. This result provides a direct way to connect the evolution of individual wealth (as in equation~(\ref{eq:income-consumption})) with the aggregate distribution of wealth (as is observed in historical data). However, I need to account for the possibility that the drift $\mu_{it}$ and the diffusion $\sigma_{it}^2$ vary across individuals $i$. To that end, I will apply a result from \citet{gyongy_mimicking_1986}, which states that it is possible to average out the heterogeneity of individual wealth processes and still retrieve the same wealth distribution in the aggregate. After applying that result, I use the \citet{kolmogorov_uber_1931} forward equation to connect equation~(\ref{eq:income-consumption}) to the evolution of the wealth distribution.

\paragraph{Reduction to a Single Equation using {\color{QueenBlue}\citeauthor{gyongy_mimicking_1986}'s~(\citeyear{gyongy_mimicking_1986})} Theorem}

In simple terms, {\color{QueenBlue}\citeauthor{gyongy_mimicking_1986}'s (\citeyear{gyongy_mimicking_1986})} theorem states the following. Consider a large number of stochastic processes $w_{it}$, each following a \gls{sde} with their own drift $\mu_{it}$ and diffusion $\sigma_{it}^2$. Then the \gls{pdf} describing the distribution of the value of these processes will behave exactly as if their drift $\mu_{it}$ and their diffusion $\sigma_{it}^2$ were replaced by the conditional expectations $\mu_t(w)=\mathds{E}[\mu_{it}|w]$ and $\sigma^2_t(w)=\mathds{E}[\sigma^2_{it}|w]$. This result, plus some basic rules of stochastic calculus, makes it possible to reduce the arbitrarily complex nature of individual wealth accumulation into a single \gls{sde} that characterizes the entire wealth distribution. I can state the following result.

\begin{prop}\label{prop:gyongy}
Let $y_t(w)$, $r_t(w)$ and $c_t(w)$ be the average labor income, rate of return and consumption, conditional on wealth $w$ at time $t$. Similarly, let $\upsilon^2_t(w)$, $\phi^2_t(w)$ and $\gamma^2_t(w)$ be the variance of labor income, rates of return and consumption, conditional on wealth $w$ at time $t$. Then, the stochastic process governed by the \gls{sde} with deterministic coefficients:
\begin{equation*}
\dif w_{it} = \underbrace{(y_t(w_{it}) + (r_t(w_{it}) - g_t)w_{it} - c_t(w_{it}))}_{\mu_t(w) \;\text{(drift)}} \dif t + \underbrace{\left(\upsilon^2_t(w_{it}) + \phi^2_t(w_{it})w^2_{it} + \gamma^2_t(w_{it})\right)^{1/2}}_{\sigma_t(w) \;\text{(diffusion, or mobility)}} \dif B_{it}
\end{equation*}
has the same marginal distribution as the process described by equation~(\ref{eq:income-consumption}).
\end{prop}%
\begin{proof}%
See Appendix~\ref{sec:proof-gyongy}.
\end{proof}%
\noindent Proposition~\ref{prop:gyongy} reduces the dynamics of the wealth distribution to a single \gls{sde}, in which everyone with the same wealth $w$ faces the same drift $\mu_t(w)$ and the same diffusion $\sigma^2_t(w)$. In that equation, the diffusion $\sigma^2_t(w)$ becomes easily interpretable as a mobility parameter. If it were equal to zero, everyone with the same wealth would face the same wealth growth, so there would be no movement across the distribution. But when that parameter is not zero, the approach can account for the mobility across the wealth distribution, which is a sizeable phenomenon \citep{gomez_decomposing_2022}. From now on, I will therefore refer to $\sigma^2_t(w)$ as a mobility parameter.

\paragraph{The Kolmogorov Forward Equation}

Using the reduced \gls{sde} of Proposition~\ref{prop:gyongy}, I can now apply the Kolmogorov forward equation. The density $f_t$ which describes the distribution of wealth at time $t$ obeys the \gls{pde}:
\begin{equation}\label{eq:kf}
\partial_t f_t(w) = -\partial_w[\mu_t(w)f_t(w)] + \frac{1}{2}\partial^2_w[\sigma_t^2(w)f_t(w)]
\end{equation}
This equation describes the evolution of a quantity that is observable in the historical data (the density of wealth $f_t$) while connecting it to parameters that characterize wealth accumulation at the individual level: the drift $\mu_t(w)$, and the diffusion $\sigma_t^2(w)$. Thus, it directly connects individual economic behavior with the distribution of wealth.

\subsubsection{Interpretation of the Different Effects}

The previous section derived equation~(\ref{eq:kf}) using results from stochastic calculus. This section provides a more direct and intuitive understanding of the equation. The purpose is twofold. First, it provides an understanding of the central equation that does not require familiarity with stochastic calculus. Second, it gives a more transparent basis for how and why the decomposition introduced in this paper works.

\paragraph{Integration of the Kolmogorov Forward Equation}

For empirical purposes, it is useful to rewrite equation~(\ref{eq:kf}) in its integrated version, which involves the \gls{cdf} of wealth, $F_t(w)$. After re-arranging terms, we get:
\begin{equation}\label{eq:kf-int}
\underbrace{-\frac{\partial_t F_t(w)}{f_t(w)}}_{\substack{\text{local change in the} \\ \text{wealth distribution}}} \;\; = \underbrace{\mu_t(w)}_{\substack{\text{local effect of average}\\ \text{change in wealth}}} - \ \ \; \underbrace{{\frac{1}{2}}\partial_w \sigma^2_t(w)}_{\substack{\text{local effect of the} \\ \text{mobility gradient}}} \ \; - \quad \underbrace{{\frac{1}{2}}\sigma^2_t(w) \frac{\partial_w f_t(w)}{f_t(w)}}_{\substack{\text{local effect} \\ \text{of mobility}}}
\end{equation}
Each of these terms captures a different mechanism. I discuss them in turn below. To fix ideas, let us consider $w=\$1\text{bn}$ and that the population is normalized to one, so that $1 - F_t(w)$ is the number of billionaires, and $-{\partial_t F_t(w)}$ is the change in the number of billionaires.

\paragraph{Local Effect of Average Change in Wealth}

Assume that wealth growth is a deterministic function of wealth: everyone with the same wealth experiences the same wealth change, so there is no mobility ($\sigma_t \equiv 0$). Consider the case where wealth growth is positive at the top, so that the number of billionaires increases.

Over a short period, the number of people crossing the \$1bn threshold will be proportional to (i) $f_t(w)$, the number of people that were initially at the threshold, and (ii) $\mu_t(w)$, the pace at which their wealth increases. Therefore, we get $-{\partial_t F_t(w)} = \mu_t(w)f_t(w)$, which corresponds to equation~(\ref{eq:kf-int}) when $\sigma_t \equiv 0$.

This formula is known as a transport equation. If wealth growth is uniform ($\mu_t \equiv \mu$), then it translates the entire wealth distribution by a factor $\mu\dif t$ over a period $\dif t$. The general formulation makes it possible to consider non-uniform wealth growth.\footnote{Note that with the change of variable $p=F_t(w_t)$, we get $\partial_t w_t(p) = \mu_t(p)$. Hence, $\mu_t(p)$ is the growth of the $p$th quantile.}

\paragraph{Local Effect of Mobility}

Now, consider the opposite thought experiment. Everyone, even if they have the same wealth, experiences a different wealth growth. But they are as likely to go up or down, so, on average, wealth growth is zero ($\mu_t \equiv 0$). Furthermore, assume that the amplitude of wealth variations is uniform ($\sigma_t \equiv \sigma$). Under these conditions, the number of billionaires will still change.

Indeed, some people just below \$1bn will see their wealth increase and become billionaires. This flow is proportional to (i) $f_t(w^-)$, the number of people right below \$1bn, and (ii) the amplitude $\sigma^2$ by which their wealth varies. Conversely, some people just above \$1bn will see their wealth decrease and drop out of the list of billionaires. This flow is proportional to (i) $f_t(w^+)$, the number of people right above \$1bn and (ii) the amplitude $\sigma^2$ by which their wealth varies.

In general, these two flows will not cancel out because there are not as many people right below and right above \$1bn. This difference between the population on each side of the threshold is effectively captured by the \textit{derivative} of the density $f_t$. Mathematically, this derivative appears from writing the difference between the two flows, and then taking the limit as $w^+ \rightarrow w$ and $w^- \rightarrow w$. After applying the correct proportionality factor, we get $-{\partial_t F_t(w)} = \frac{1}{2}\sigma^2_t(w) {\partial_w f_t(w)}$, which corresponds to equation~(\ref{eq:kf-int}) when $\mu_t \equiv 0$ and $\sigma_t \equiv \sigma$. Unlike the equation modeling the effect of average wealth growth, in this equation, the effect on the number of billionaires depends not on the \textit{value} of the density but on its \textit{gradient}. This formula is known as a diffusion equation and is best understood as a transformation that flattens wealth density.

To make sense of this effect, consider two extremes. First, if the wealth density is flat at the top, then the number of people who cross the \$1bn threshold from both sides will cancel out. Hence, even though wealth changes at the individual level, the overall effect on the distribution is nil. Now, assume that, on the contrary, the wealth density is infinitely steep, so that there are no billionaires but many people just below \$1bn. Some people will see their wealth increase and become billionaires. But since there is initially no one above \$1bn, there can be no countervailing flow of people leaving the group. And therefore, the number of billionaires will increase very fast.

\paragraph{Local Effect of the Mobility Gradient}

If the amplitude of wealth mobility is not uniform, there is a third effect on the wealth distribution. Indeed, assume that there is more variation in wealth growth above \$1bn than below: then, downward mobility will exceed upward mobility. Thus, even with no average growth and a flat density, people who drop out of the billionaire group will outnumber those who enter it. This phenomenon creates an additional effect on $-{\partial_t F_t(w)}$, which depends on how mobility varies throughout the distribution. Hence, it is a function of the mobility gradient $\partial_w \sigma^2_t(w)$, and is equal to $-{\frac{1}{2}}f_t(w)\partial_w \sigma^2_t(w)$.

\subsubsection{Phase Portrait and The Dynamics of Wealth Inequality}

Consider the case where the drift and mobility parameters are constant over time: $\mu_t(w)\equiv\mu(w)$ and $\sigma^2_t(w)\equiv\sigma^2(w)$. Equation~(\ref{eq:kf-int}) is best represented as a curve that relates the current level to the current evolution of inequality, as in Figure~\ref{fig:phase-portrait}. This curve --- a phase portrait --- lets us picture the dynamics of wealth inequality in a simple way.

\begin{figure}[ht]
    \centering
    \resizebox{0.8\linewidth}{!}{\input{figures/phase-portrait.tex}}
    
    \vspace{1em}
    \begin{minipage}{0.7\linewidth}
    \footnotesize \textit{Note:} This diagram is a phase portrait of the dynamics of inequality for a given wealth level $w$ located towards the top of the wealth distribution. These dynamics can be pictured in a two-dimensional space where each axis represents a quantity associated to the wealth distribution, whose \gls{pdf} is $f_t(w)$ and whose \gls{cdf} is $F_t(w)$. The $x$-axis is equal to $\partial_w f_t(w)/f_t(w)$ and is a proxy for inequality levels (high values mean high inequality). The $y$-axis is equal to $-\partial_t F_t(w)/f_t(w)$ and is a proxy for inequality changes (high values mean increasing inequality). The diagram represents the transition from a low inequality level ($x_0$) to a high inequality level ($x_\infty$), with constant $\mu(w)$ and $\sigma^2(w)$. During this transition, the system moves alongside the orange line with slope $-\frac{1}{2}\sigma^2(w)$ and intercept $\mu(w) - \frac{1}{2}\partial_w\sigma^2(w)$. When the system reaches the $x$-axis, where $-\partial_t F_t(w)/f_t(w) = 0$, it is a its steady-state.
    \end{minipage}
    \caption{Phase Portrait: Convergence to a Higher-Inequality Steady State}
    \label{fig:phase-portrait}
\end{figure}

The term $\partial_w f_t(w)/f_t(w)$ on the right-hand side of equation~(\ref{eq:kf-int}) measures the relative slope of the density. In the upper tail of the wealth distribution, the density is sloping downward, so $\partial_w f_t(w)/f_t(w) < 0$. This value is a good proxy for top wealth inequality. If it is high, then the density decreases at a slow rate: i.e., the tail is fat, and inequality is high. Conversely, if its value is low, then inequality is low as well.\footnote{If wealth follows a Pareto distribution with coefficient $\alpha$, then $\partial_w f_t(w)/f_t(w) = -(1 + \alpha)/w$. So, as $\alpha \rightarrow 1$, inequality goes to infinity and $\partial_w f_t(w)/f_t(w)$ increases to zero.} We can analogously interpret the term $-{\partial_t F_t(w)}/{f_t(w)}$ on the left-hand side of equation~(\ref{eq:kf-int}) as a change in the fatness of the top tail of the distribution. When it exceeds zero, the tail becomes fatter, and inequality increases.\footnote{If wealth follows a Pareto distribution with coefficient $\alpha_t$, then $-{\partial_t F_t(w)}/{f_t(w)} = -(w \log w) (\partial_t \alpha_t/\alpha_t)$. Therefore, $-{\partial_t F_t(w)}/{f_t(w)} > 0$ is associated with a decrease in $\alpha_t$, which corresponds to an increase in inequality.}

In Figure~\ref{fig:phase-portrait}, the first quantity $\partial_w f_t(w)/f_t(w)$ (inequality level, in blue) is on the $x$-axis and the second quantity $-{\partial_t F_t(w)}/{f_t(w)}$ (inequality change, in green) is on the $y$-axis. Equation~(\ref{eq:kf-int}) states that, with stable parameters, the data points representing the current state of inequality must lie alongside a straight line (in orange), with intercept $\mu_t(w) - {\frac{1}{2}}\partial_w \sigma^2_t(w)$ and slope $-\frac{1}{2}\sigma_t^2(w)$. Where this line crosses the $x$-axis, inequality is neither increasing nor decreasing: this is the distribution's steady state.

\paragraph{Dynamics of Wealth Inequality and Convergence to a Steady State}

Consider the transition to a high inequality steady state. On Figure~\ref{fig:phase-portrait}, start from the inequality level $x_0$. This level is below the steady-state level $x_\infty$, so inequality will change. We can decompose this change into (i) the effect of average wealth growth and of the mobility gradient, which the intercept captures, and (ii) the effect of mobility, which the slope captures. As we can see on the graph, (i) decreases inequality, whereas (ii) increases it. At first, the effect of (ii) is stronger, so inequality increases overall by $y_0$. In the next period, inequality is higher, equal to $x_1$. A higher inequality means a fatter tail, hence a flatter density. Having a flatter density weakens the effect of (ii). The effect of (i), on the other hand, remains unchanged. Overall, inequality still increases, but by a lower amount, $y_1$. The process repeats. Inequality keeps increasing, which flattens the density, weakens the effect of (ii) but leaves the effect of (i) unchanged. Asymptotically, we reach $x_\infty$. At this point, the effect of (ii) has been weakened to the point that it perfectly counterbalances (i). We have reached the steady state.

\paragraph{Rationale for the Steady State}

This framework provides a robust justification for the emergence of a steady state --- one that doesn't preclude but also doesn't require behavioral responses. Without mobility, the only way to get a steady-state distribution of wealth is for behavioral responses (and general equilibrium effects) to lead to a point where saving rates as a proportion of wealth become identical throughout the distribution. Any deviation from this situation leads to a degenerate steady state in the long run. In contrast, here, the distribution eventually stabilizes because of the mechanical interaction between mobility and drift, so the model remains well-behaved for a wider range of economic behaviors.

\paragraph{Distinction Between Drift and Mobility}

This framework shows how drift and mobility have different impacts on the distribution and why that distinction matters. In a model without mobility, the red line in Figure~\ref{fig:phase-portrait} would be flat. The only way to account for inequalities of wealth that are neither increasing nor decreasing at a constant rate is to assume that the intercept $\mu_t(w)$ changes with each period (which usually implies some form of behavioral change). In contrast, once we introduce mobility, it becomes possible to account for any linear downward-slopping evolution of inequality in the phase portrait of Figure~\ref{fig:phase-portrait} with a much more parsimonious model that only involves two time-invariant parameters: one for drift and one for mobility.

\subsection{Other Processes Affecting the Wealth Distribution}

I account for other phenomena impacting wealth distribution besides drift and diffusion (which capture individual income and consumption). The three phenomena that I consider in this paper are birth and death, inheritance, and assortative mating. I first introduce these effects in equation~(\ref{eq:kf}) and then gives a version of equation~(\ref{eq:kf-int}) that include them as well.

\paragraph{Births and Deaths}

At any time $t$, a fraction $\delta_t$ of people die. Let $g_t$ be the density of their wealth. Simultaneously, a fraction $\beta_t$ of people appears with a random initial endowment drawn from a distribution with density $h$. The total population grows at a rate $n_t = \beta_t - \delta_t$. This process impacts $\partial_t f_t(w)$, i.e., the change in the wealth density at wealth $w$ and time $t$, by:
\begin{equation*}
\underbrace{\zeta_t(w)}_{\text{demography}} \  \equiv \quad \underbrace{\beta_t h(w)}_{\text{injection}} \quad - \quad \underbrace{\delta_t g_t(w)}_{\text{deaths}} \quad - \ \; \underbrace{n_t f_t(w)}_{\text{normalization}}
\end{equation*}
and the equation~(\ref{eq:kf}) becomes:
\begin{equation*}
\partial_t f_t(w) \;\;\; = \;\;\; \underbrace{-\partial_w[\mu_t(w)f_t(w)] + \frac{1}{2}\partial^2_w[\sigma_t^2(w)f_t(w)]}_{\text{income and consumption}} \quad + \underbrace{\zeta_t(w)}_{\text{demography}}
\end{equation*}

\paragraph{Inheritance}

I model inheritance as a jump process. With a probability $\pi_t(w)$, people see their wealth jump from $w$ to $w + \lambda$ where $\lambda$ is the amount of inheritance received, net of taxes. Let $s_t(\lambda|w)$ be the density of the value of the inheritance, conditional on the value of wealth, and conditional on receiving an inheritance. We can model the jump process as a death with rate $\pi_t(w)$ and as an injection with rate $\int \pi_t(w - \lambda)f_t(w - \lambda)s_t(\lambda|w-\lambda)\,\dif\lambda$. So, the effect of inheritance on $\partial_t f_t(w)$, i.e., the change in the wealth density at wealth $w$ and time $t$, is:
\begin{equation*}
\xi(w) \equiv \int \pi_t(w - \lambda)f_t(w - \lambda)s_t(\lambda|w-\lambda)\,\dif\lambda - \pi_t(w)f_t(w)
\end{equation*}
and the equation~(\ref{eq:kf}) with both demography and inheritance is:
\begin{equation*}
\partial_t f_t(w) \;\;\; = \;\;\; \underbrace{-\partial_w[\mu_t(w)f_t(w)] + \frac{1}{2}\partial^2_w[\sigma_t^2(w)f_t(w)]}_{\text{income and consumption}} \quad + \underbrace{\zeta_t(w)}_{\text{demography}} + \underbrace{\xi_t(w)}_{\text{inheritance}}
\end{equation*}

\paragraph{Marriages, Divorces and Assortative Mating}

Finally, I account for the effect of marriages and divorces. I adopt the convention that wealth is split equally among spouses. At any time $t$, a fraction $\theta_t$ of people get married. Let $p_t(w)$ be the density of their wealth as individuals, and $q_t(w)$ be the density of their wealth as couples. The strength of assortative mating is captured by the relationship between $p_t$ and $q_t$: at the limit, if people always choose a spouse with identical wealth, then $p_t(w)=2q_t(2w)$. The effect of marriages on the wealth distribution is therefore $2\theta_t q_t(2w)-\theta_t p_t(w)$. We can construct an opposite process for divorces. Let $\chi_t(w)$ be combined effect of marriages and divorces on $\partial_t f_t(w)$. Then the equation~(\ref{eq:kf}) with the effect of demography, inheritance, marriages, and divorces becomes:
\begin{equation}\label{eq:kf-full}
\partial_t f_t(w) \;\;\; = \;\;\; \underbrace{-\partial_w[\mu_t(w)f_t(w)] + \frac{1}{2}\partial^2_w[\sigma_t^2(w)f_t(w)]}_{\text{income and consumption}} \quad + \underbrace{\zeta_t(w)}_{\text{demography}} + \underbrace{\xi_t(w)}_{\text{inheritance}} + \underbrace{\chi_t(w)}_{\substack{\text{marriages} \\ \text{and divorces}}}
\end{equation}

\paragraph{Integrated Version}

Rewrite equation~(\ref{eq:kf-full}) in an integrated form, similar to~(\ref{eq:kf-int}). Define $F_t(w) = \int_{-\infty}^w f_t(s)\,\dif s$, $Z_t(w) = \int_{-\infty}^w \zeta_t(s)\,\dif s$, $\Xi_t(w) = \int_{-\infty}^w \xi_t(s)\,\dif s$ and $X_t(w) = \int_{-\infty}^w \chi_t(s)\,\dif s$. After integrating equation~(\ref{eq:kf-full}) and re-arranging terms, we get:
\begin{equation}\label{eq:kf-full-int}
-\frac{\partial_t F_t(w)}{f_t(w)} + \frac{Z_t(w)}{f_t(w)} + \frac{\Xi_t(w)}{f_t(w)} + \frac{X_t(w)}{f_t(w)} = \mu_t(w) - {\frac{1}{2}}\partial_w \sigma^2_t(w) - {\frac{1}{2}}\sigma^2_t(w) \frac{\partial_w f_t(w)}{f_t(w)}
\end{equation}
This equation is similar to~(\ref{eq:kf-int}), with additional correction terms on the left-hand side to account for demography, inheritance, and assortative mating. The fundamental dynamics of wealth inequality remain similar to those pictured in the phase portrait in Figure~\ref{fig:phase-portrait}, except that the $y$-axis needs to be adapted to include the correction terms. This equation will serve as the basis for the decomposition introduced in this paper.

\subsection{Decomposition of the Different Effects}

I will use equation~(\ref{eq:kf-full-int}) to decompose the various factors affecting the distribution of wealth. All the terms in the equation can be directly observed or separately estimated, except those related to consumption. Therefore, these unobserved parameters have to be estimated as a type of ``residual.'' This section explains how.

\paragraph{Estimating Equation for the Complete Model}

Let us go back to the equation~(\ref{eq:kf-full-int}). Separate the drift $\mu_t(w)$ and the diffusion $\sigma^2_t(w)$ into their observed component (income) and their unobserved component (consumption):
\begin{align*}
\underbrace{\mu_t(w)}_{\parbox{4.5em}{\setstretch{1}\scriptsize\centering total drift (average wealth change)}} &= \underbrace{\tilde{\mu}_t(w)}_{\parbox{4.5em}{\setstretch{1}\scriptsize\centering mean income (observed)}} - \underbrace{c_t(w)}_{\parbox{4.5em}{\setstretch{1}\scriptsize\centering  mean consumption (unobserved)}} & & &
\underbrace{\sigma^2_t(w)}_{\parbox{4.5em}{\setstretch{1}\scriptsize\centering total mobility (variance of wealth change)}} &= \underbrace{\tilde{\sigma}^2_t(w)}_{\parbox{4.5em}{\setstretch{1}\scriptsize\centering variance of income (observed)}} + \underbrace{\gamma^2_t(w)}_{\parbox{4.5em}{\setstretch{1}\scriptsize\centering variance of consumption (unobserved)}}
\end{align*}
Move the observed components of drift and mobility to the left-hand side of equation~(\ref{eq:kf-full-int}). The complete equation for the left-hand side becomes:
\begin{equation*}
\text{LHS}_t(w) \  \equiv\   \underbrace{-\frac{\partial_t F_t(w)}{f_t(w)}}_{\parbox{4em}{\setstretch{1}\scriptsize\centering inequality change}} \  + \  \underbrace{\frac{Z_t(w)}{f_t(w)} + \frac{\Xi_t(w)}{f_t(w)} + \frac{X_t(w)}{f_t(w)}}_{\parbox{10em}{\setstretch{1}\scriptsize\centering demography, inheritance, marriages and divorces}} \  - \  \underbrace{\tilde{\mu}_t(w) + {\frac{1}{2}}\partial_w \tilde{\sigma}^2_t(w) + {\frac{1}{2}}\tilde{\sigma}^2_t(w) \frac{\partial_w f_t(w)}{f_t(w)}}_{\parbox{15em}{\setstretch{1}\scriptsize\centering drift and mobility induced by income}}
\end{equation*}
All the components of $\text{LHS}_t(w)$ are either directly observable or separately estimable. We can therefore rewrite~(\ref{eq:kf-full-int}) in its final form:
\begin{equation}\label{eq:estimation-complete}
\text{LHS}_t(w) \  = \  \underbrace{c_t(w) - {\frac{1}{2}}\partial_w \gamma^2_t(w)}_{\text{intercept}} \  - \  \underbrace{{\frac{1}{2}}\gamma^2_t(w)}_{\text{slope}} \times \frac{\partial_w f_t(w)}{f_t(w)}
\end{equation}

\paragraph{Identification}\label{sec:identification-final}

Equation~(\ref{eq:estimation-complete}) provides the basis for estimating the parameters. It shows a relationship similar to that shown in Figure~\ref{fig:phase-portrait}. We require three assumptions to get point estimates of the mean $\{c_t(w)\}_{w \in \mathds{R}}$ and the variance $\{\gamma_t^2(w)\}_{w \in \mathds{R}}$ of consumption conditional on wealth over a period $[t_0, t_1]$.
\begin{assump}\label{assump:observe}
For all $w$, we can observe (or separately estimate) $\mathrm{LHS}_t(w)$ and ${\partial_w f_t(w)}/{f_t(w)}$ at $k$ distinct times $t_k \in [t_0, t_1]$ with $k \geq 2$.
\end{assump}
\begin{assump}\label{assump:stability}
For all $w$, the parameters $c_t(w)$ and $\gamma_t^2(w)$ are constant over time, i.e., $c_t(w) \equiv c(w)$ and $\gamma_t^2(w) \equiv \gamma^2(w)$ for $t \in [t_0, t_1]$.
\end{assump}
\begin{assump}\label{assump:change}
The wealth distribution is changing over time, in the sense that, for all $w$, we observe the distribution for at least to periods $r, s \in [t_0, t_1]$ such that ${\partial_w f_{r}(w)}/{f_{r}(w)} \neq {\partial_w f_{s}(w)}/{f_{s}(w)}$.
\end{assump}
These three assumptions ensure we can meaningfully fit a line through the different points as in Figure~\ref{fig:phase-portrait}. Assumption~\ref{assump:observe} states that we need to observe the wealth distribution and its evolution during two different periods, a consequence of the fact that we need at least two points to be able to fit a line. In practice, and in the presence of statistical noise, more than two points are preferable to get robust estimates. Assumption~\ref{assump:stability} ensures that the parameters that govern the relationship~(\ref{eq:estimation-complete}) remain stable over the period of time under consideration. Finally,  Assumption~\ref{assump:change} states that the distribution of wealth must change over the period where we observe it. This assumption comes from the fact that we are using local variations in the flatness of the density to disentangle the drift from mobility. This strategy only works if the flatness does, in fact, vary. If these three assumptions are satisfied, then we can proceed with the estimation in two steps:
\begin{description}
\item[Step 1] Estimate $c^*_t(w) = c_t(w) - {\frac{1}{2}}\partial_w \gamma^2_t(w)$ and $\gamma^2_t(w)$ for every $w$ using a line-fitting method.
\item[Step 2] Using the estimate for $\gamma^2_t$, estimate the mobility gradient $\partial_w \gamma^2_t$, and use it to get the estimate of $c_t(w) = c^*_t(w) + {\frac{1}{2}}\partial_w \gamma^2_t(w)$.
\end{description}

\paragraph{Ex-ante Estimate of the Effect's Magnitude}

I use variations in the flatness of the density over time to disentangle the effect of drift from that of mobility. But are these variations large enough to provide reliable empirical estimates? Some rough but simple calculations can give a general idea of the magnitude of the effects at play. Between its high point in the 1970s and today, the Pareto coefficient of wealth in the upper tail went from $\alpha \approx 2$ to $\alpha \approx 1.5$. Assume a mobility parameter $\sigma^2 \approx 0.16w^2$ at the top, which matches the estimates in this paper (and separate estimates from the \gls{scf} and the \gls{psid}). The effect of mobility under these conditions is equal to $\frac{1}{2}\sigma^2(1 + \alpha)w$. A useful way to interpret this number is to say that mobility has the same effect on the distribution as an average wealth growth of $\frac{1}{2}\sigma^2(1 + \alpha)$, which went from $\frac{1}{2} \times 0.16 \times (1 + 2) = 24\%$ to $\frac{1}{2} \times 0.16 \times (1 + 1.5) = 20\%$ of wealth --- a 4\% change. The change in the effect of mobility that is attributable to the flattening of the wealth density at the top is, therefore, sizable --- equivalent to the mechanical effect that a permanent 4\% wealth tax would have.

\paragraph{Potential Limitations}

The main limitation of the method is that it requires the parameters for the drift $c(w)$ and the mobility $\gamma^2(w)$ induced by consumption to be constant (or at least not to have a downward or upward trend) over sufficiently long periods. Estimates may be biased if this is not the case. For example, assume that the drift decreases over time. On the phase portrait in Figure~\ref{fig:phase-portrait}, the linear relationship (in orange) moves down over time. As a result, the observed data points on the phase portrait will lie on a curve that could take many different shapes, but, in general, will appear more downward-sloping than the true relationship. If we were to estimate the decomposition in that context, we would therefore overestimate the diffusion parameter.

In practice, there are two ways to address this issue. The first one is to check that the empirical phase portrait to which we fit the line remains relatively close to linearity. While this does not guarantee that the assumption of a constant $c(w)$ and $\gamma^2(w)$ is verified, it can identify some problematic cases. The second one is to verify that the parameters estimated from the decomposition are consistent with external estimates of savings and mobility. I perform both checks in this paper's empirical section, suggesting that the simple model with a constant $c(w)$ and $\gamma^2(w)$ since the 1980s works well.

%% file: figures/phase-portrait.tex
\centering
\begin{tikzpicture}[scale=1.25]

\draw[->,very thick,color=Charcoal] (-6,0) -- (1,0);
\draw[<->,very thick,color=QueenBlue] (0,-2) -- (0,3);

\draw (1,0) node[right,color=Charcoal,scale=0.9]{$\textstyle\partial_w f_t(w)/f_t(w)$};

\draw (3.2, 0) node[right,align=center,scale=0.9,color=Charcoal]{high \\ inequality};
\draw (-6.35, 0) node[left,align=center,scale=0.9,color=Charcoal]{low \\ inequality};

\draw (0,3) node[above,color=QueenBlue,scale=0.9]{$\textstyle-\partial_t F_t(w)/f_t(w)$};

\draw (0,4.25) node[align=center,scale=0.9,color=QueenBlue]{increasing \\ inequality};
\draw (0, -2.75) node[align=center,scale=0.9,color=QueenBlue]{decreasing \\ inequality};

\draw[very thick,color=SafetyBlazeOrange] plot[domain=-5.5:0.5] (\x, {-1.3 - 0.75*\x});
\draw[very thick,dashed,color=SafetyBlazeOrange] (-3,-1.3) -- (1,-1.3) node[right,scale=0.9]{$\textstyle\mu_t(w) - {\frac{1}{2}}\partial_w \sigma^2_t(w)$};
\draw[very thick,<-,color=SafetyBlazeOrange] (-1.75,-1.3) node[above left,scale=0.9]{$\textstyle-\frac{1}{2}\sigma_t^2(w)$} arc (180:143.130:1.75);

\foreach \i in {0,...,4} {
    \pgfmathsetmacro{\x}{-5 + 3.26667*(1 - exp(-0.4*\i))}
    
    \draw (\x,0) node[below,scale=0.9]{$x_\i$};
    
    \draw[dotted] (\x,0) -- (\x,{-1.3 - 0.75*\x}) node[scale=0.9]{$\bullet$} -- (0, {-1.3 - 0.75*\x}) node[right,scale=0.9]{$y_\i$};
}
\foreach \i in {0,...,3} {
    \pgfmathsetmacro{\x}{-5 + 3.26667*(1 - exp(-0.4*\i))}
    \pgfmathsetmacro{\y}{-5 + 3.26667*(1 - exp(-0.4*(\i + 1)))}
    
    \draw[->,thick] (\x + 0.1,{-1.3 - 0.75*\x}) to[bend left] (\y,{-1.3 - 0.75*\y + 0.1});
}

\draw (-1.73333,0) node[above right,scale=0.9]{$x_\infty$};
\draw (-1.73333,0) node[scale=0.9]{$\bullet$};

\end{tikzpicture}

%% file: content/4-estimations.tex
\subsection{Data}

I estimate the parameters in decomposition~(\ref{eq:estimation-complete}) for the United States. I start in 1962, when sufficiently detailed microdata on the distribution of income and wealth in the United States becomes available. I primarily rely on the tax-based microdata from \citet{saez_rise_2020}, which I complement in a number of ways. Using the data collected, I perform microsimulations of mortality, birth, inheritance, marriage, and divorce, which I use to estimate the impact of the corresponding phenomenons on the wealth distribution. I briefly review the data and methodology below and provide more details in Appendix~\ref{sec:detailed-data}.

\paragraph{Income and Wealth}

For income and wealth, I rely on the \gls{dina} tax-based microdata from \citet{saez_rise_2020}. These data distribute the entirety of national income and wealth, as measured by the national accounts, every year since 1962, to adult individuals (20 and older). They infer the distribution of wealth from capital income flows, following the capitalization method of \citet{saez_wealth_2016}.\footnote{Their latest revision \citep{saez_rise_2020} accounts for heterogeneous returns.} This data does not distribute capital gains since they are not part of national income. I incorporate them in the data by assuming a constant capital gains rate by asset class, as in \citet{robbins_capital_2018}. The information on age in the \gls{dina} microdata is also limited, so I replace it with information from the \glsfirst{scf}. I match \gls{dina} and \gls{scf} observations one-to-one based on their rank in the wealth distribution. I use this to estimate a rank in the age distribution by sex for each observation and attribute to them the age that matches this rank according to official demographic data.

\paragraph{Demography and Intergenerational Linkages}

I estimate the entire demographic history of the United States since 1850 by year, age, and sex, including population counts, mortality rates, female and male fertility rates by birth order. I start the estimation in 1850, long before the income and wealth data starts (in 1962), because I use the demographic data to simulate intergenerational linkages between parents and children. So if a centenarian dies in 1962, I must be able to retrace that person's entire fertility history since their birth and retrace the mortality history of that person's children. 

I construct this data by collecting and harmonizing data from official sources {\color{QueenBlue}(United States Census Bureau)}, historical databases {\color{QueenBlue}(Human Mortality Database\nocite{HumanMortalityDatabase}, Human Life Table Database\nocite{HumanLifeTable}, Human Fertility Database\nocite{HumanFertility}, Human Fertility Collection\nocite{HumanFertilityCollection})} and academic publications \citep{haines_estimated_1998}. To make projections for the future, I rely on the forecast (medium variant) of the World Population Prospects \citep{united_nations_world_2019}. For male fertility rates, which are not a standard demographic parameter, I combine the female fertility rates with the joint age distribution of mixed-sex couples in the IPUMS census microdata \citep{ruggles_steven_ipums_2022}.

First, I use this data to simulate deaths, assuming that people die at random according to their age and sex-specific mortality rate. I also account for births by assuming that people enter the population at age 20 with a constant exogenous wealth distribution estimated from the data.\footnote{People aged 20 have very low levels of wealth, so in practice, this is close to assuming that people start with zero wealth.}$^,$\footnote{The birth rate is estimated here as a residual between population growth and the crude death rate, so in effect, it also incorporates the effect of immigration.} Second, I simulate intergenerational linkages. For every person in the data, I assume this person had children according to their sex, age, year, and birth-order-specific fertility rates. Then I assume these children experience mortality in line with their sex, year, and age-specific mortality rates. This methodology generates a distribution for the age and sex of the direct descendants of every person in the sample, which I use to distribute inheritances.

\paragraph{Inheritance and Estate Taxation}

When someone dies, I assume that their wealth is transmitted to their spouse (if any), without estate tax, or to their children, after payment of the estate tax. To calculate the estate tax, I collect complete statutory schedules of the federal estate tax over the second half of the 20th century. I assume that wealth is split equally among descendants, as is the norm in the United States \citep{menchik_primogeniture_1980}. I account for the possibility that wealthier people are more likely to inherit and receive larger inheritances. I use the \gls{scf} to estimate the relative probability of receiving an inheritance, as well as the rank in the inheritance distribution, as a function of the rank in the wealth distribution, conditional on age. Within a given age group, I simulate international wealth transmission according to these parameters.

\paragraph{Marriages, Divorces and Assortative Mating}

I collect data on the aggregate rate of divorce and marriage from the \gls{nvss} \citep{NVSS}. In addition, I reconstruct age and sex-specific rates using population data disaggregated by marital status from IPUMS census microdata \citep{ruggles_steven_ipums_2022}.

I determine the extent of assortative mating by estimating the joint distribution of the ranks in the wealth distribution at the time of marriage using the \gls{sipp} panel (2013--2016). I also consider the impact of assortative mating on divorce by estimating the distribution of the share of wealth owned by each couple member before they get divorced using the same data.

Then, I simulate the process of marriage and divorce by randomly selecting people to get wedded in a given year according to age and sex-specific marriage rates, and then marry them to one another to reproduce the joint distribution of wealth ranks at marriage in the \gls{sipp}. Similarly, I simulate the effect of divorce by randomly selecting people to get separated according to their age and sex-specific divorce rates, and split the couple's wealth among both spouses according to the \gls{sipp} data.

\subsection{Estimation}\label{sec:estimation}

\paragraph{Wealth Distribution}

First, I normalize the distribution of wealth by the average national income per adult. Then, I transform it using the inverse hyperbolic sine function ($\ash$).\footnote{I.e., $\ash: x \mapsto \log(x + \sqrt{x^2+1})$.} This transformation makes the distribution of wealth easier to manipulate empirically. And because we operate in a continuous time framework, it creates no practical difficulty: Itô's lemma establishes a direct correspondence between the parameters of the process for $w_{it}$ and for $\ash(w_{it})$.\footnote{Itô's lemma states that, if a process $x_t$ follows the \gls{sde} $\dif x_t = \mu_t \dif t + \sigma_t \dif B_t$, then $\phi_t(x_t)$ follows the \gls{sde} $\dif \phi_t(x_t) = \left[\partial_t \phi_t(x_t) + \mu_t \partial_x \phi_t(x_t) + \frac{1}{2}\sigma^2_t \partial^2 \phi_t(x,t)\right]\dif t + \sigma_t \partial_x \phi_t(x) \dif B_t$. See Appendix Section~\ref{sec:transform-param-asinh} for details.}

\begin{figure}[ht]
\begin{center}
    \begin{subfigure}[t]{0.499\textwidth}
        \includegraphics[width=\textwidth]{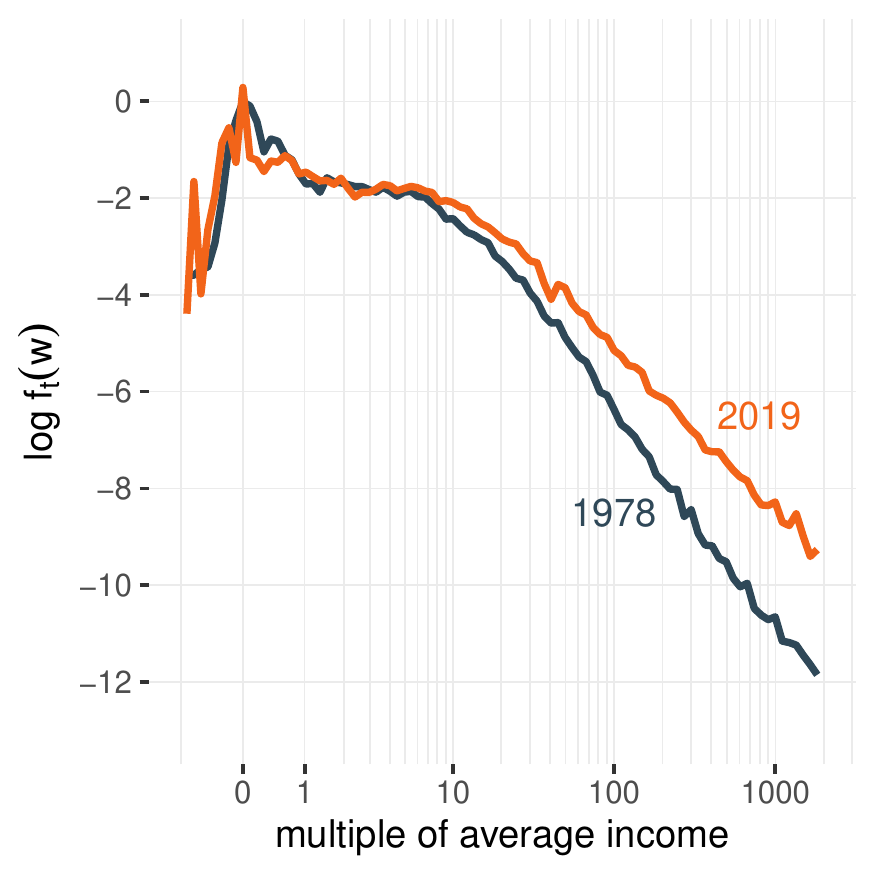}
        \caption{Log Density in 1978 and 2019}
        \label{fig:density-wealth-1978-2019}
    \end{subfigure}%
    \begin{subfigure}[t]{0.499\textwidth}
        \includegraphics[width=\textwidth]{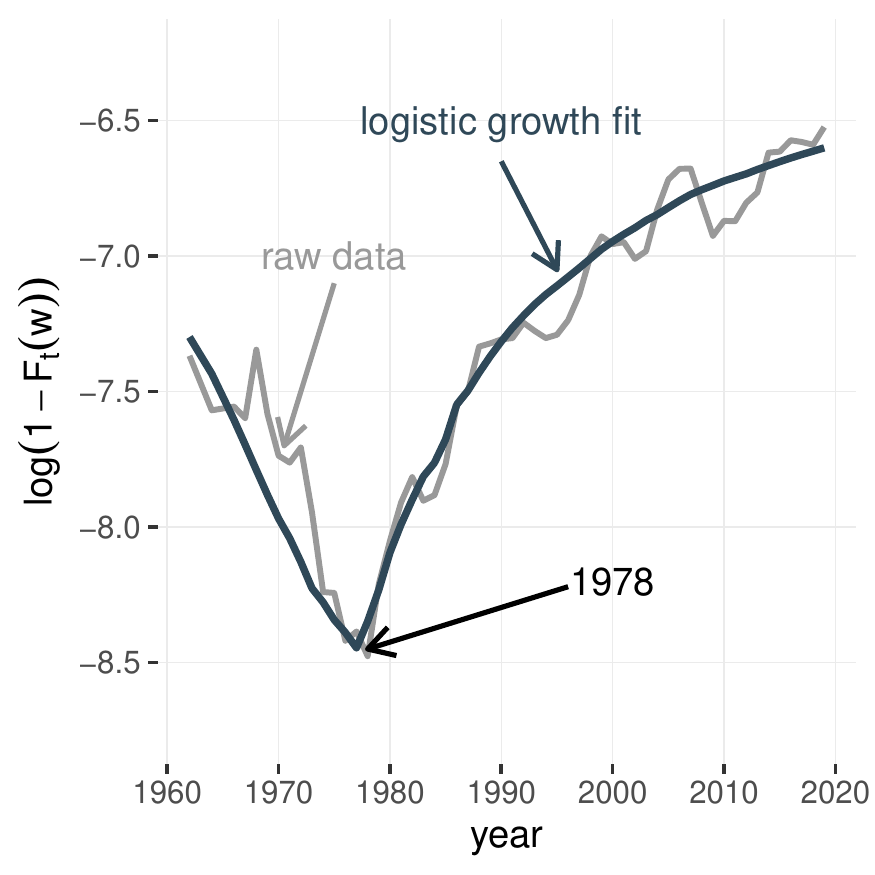}
        \caption{Log CCDF at 500 Times Average Income}
        \label{fig:ccdf-wealth}
    \end{subfigure}
    
    \vspace{1em}
    \begin{minipage}{0.9\linewidth}
    \footnotesize \textit{Source:} Own computation using the \glsfirst{dina} microdata from \citet{saez_rise_2020}. \textit{Note:} Wealth is always expressed as a multiple of national income. Densities in Figure~\ref{fig:density-wealth-1978-2019} are estimated as histograms with 91 bins of size $0.1$ on the $\ash(\text{wealth})$ scale, ranging from $-1$ to $2000$ times average income ($-0.9$ to $8.3$ on the $\ash$ scale).
    \end{minipage}
\end{center}
\vspace{-1em}
\caption{Distribution of Wealth and Its Evolution}
\label{fig:wealth-distr-evolution}
\end{figure}

I select a range of values (from $-1$ to $2000$ times average income) which the wealth microdata consistently covers over the entire period.\footnote{The range goes from $-0.9$ to $8.3$ on the $\ash$ scale.} I divide this range into bins of equal size, each representing $0.1$ units of $\ash(\text{wealth})$.\footnote{This represents 91 bins.} Figure~\ref{fig:density-wealth-1978-2019} plots the density of wealth, as estimated by the frequency of these bins. I display two years: 1978, which has the lowest inequality, and 2019, which has the highest inequality and is also the most recent data available. I use the logarithm of the density, so changes in the top tail of the wealth distribution are more clearly visible. It appears linear in the top tail, which follows from the fact large wealth holdings follow a power law.\footnote{At the top, the wealth distribution is approximately Pareto, and the inverse hyperbolic sine is approximately logarithmic, so the distribution of transformed wealth in approximately exponential.} Note that $\partial_w \log f_t(w) = \partial_w f_t(w)/w$, i.e., the derivative of the logarithm of the density is equal to the quantity of interest on the right-hand side of equation~(\ref{eq:estimation-complete}). Hence, our interest lies in the slopes of the lines, which I estimate by running locally weighted linear regressions through the values of Figure~\ref{fig:density-wealth-1978-2019}.\footnote{In the benchmark specification, I use a rectangular kernel and a bandwidth of $1.5$. See Appendix~\ref{sec:robustness-checks} for robustness checks.} The top tail has driven most of the changes in the wealth distribution, and indeed this is where we observe most of the variation. In 1978, when inequality was at its lowest, the density at the top was quite steep, with a slope around $-2$. By 2019, inequality had increased dramatically, leading to a fatter tail and a flatter density, with a slope around $-1.5$.

\paragraph{Changes in the Distribution of Wealth Over Time}

For each wealth bin, I estimate $(1 - F_t(w))/f_t(w)$, which is the relevant measure for the evolution of the wealth distribution as it appears on the left-hand side of equation~(\ref{eq:estimation-complete}). Let us begin with $\log(1-F_t(w))$, the logarithm of the numerator (see Figure~\ref{fig:ccdf-wealth} for the bin corresponding to 500 times the average income). The raw data shows two clear trends, one on each side of the year 1978, which correspond to the decreasing part and the increasing part of the U-shaped evolution of wealth inequality. I use a parametric approximation to filter out the short-run variations around these trends (which are not my focus here). Using nonlinear least squares, I fit a logistic growth model for each bin, separately on each side of the year 1978.\footnote{The logistic curve can be written as $L: t \mapsto L(t) = \frac{x_{\infty}}{1 + (x_{\infty}/x_0-1)\exp(-\rho t)}$ where $(x_0, x_{\infty}, \rho)$ are the three parameters which capture, respectively, the initial value, the asymptotic value, and the rate of convergence.}$^,$\footnote{This model is attractive for three reasons: (i) empirically, it fits the data well, (ii) it captures the features that we expect, i.e., that of a process that grows at first, and then settles to a steady state, and (iii) if we locally approximate the distribution with an exponential distribution with rate parameter $\lambda_t$, then equation~(\ref{eq:kf-int}) collapses to a logistic differential equation for $\lambda_t$, whose solution is the logistic curve function. Note that this approximation can only be valid locally and does not provide a global solution to the partial differential equation~(\ref{eq:kf-int}).} From this parametric approximation, I estimate the time derivative of $\log(1-F_t(w))$. Finally, I retrieve the quantity of interest using the fact that $\partial_t F_t(w) = (1 - F_t(w)) \partial_t \log(1-F_t(w))$.

\paragraph{Other Processes}

I separately estimate the effects of income, demography, inheritance, marriage, and divorce, using the microdata and, when needed, microsimulations based on the microdata. For income, I directly calculate the mean and the variance within each wealth bin and use this to estimate the drift and the mobility induced by income. For demography, inheritance, marriage, and divorce, I simulate the processes in the microdata and use the difference between the \glspl{cdf} before and after simulation to estimate the effects.

\paragraph{Estimation of Drift and Mobility}

Having estimated all the observable components of equation~(\ref{eq:estimation-complete}), we can plot the empirical counterpart to the phase portrait in Figure~\ref{fig:phase-portrait}, for every wealth bin. The result, for each year and for the bin corresponding to a wealth level of 500 times the average income, is shown in Figure~\ref{fig:phase-portrait-data}. Due to the inverse hyperbolic sine transform of wealth and to the inclusion of effects besides drift and mobility, the interpretation of the parameters is not as straightforward as in Figure~\ref{fig:phase-portrait}. But the fundamental linear relationship should hold.

\begin{figure}[ht]
\captionsetup[subfigure]{justification=centering}
\begin{center}
    \begin{subfigure}[t]{0.499\textwidth}
        \centering
        \includegraphics[width=\textwidth]{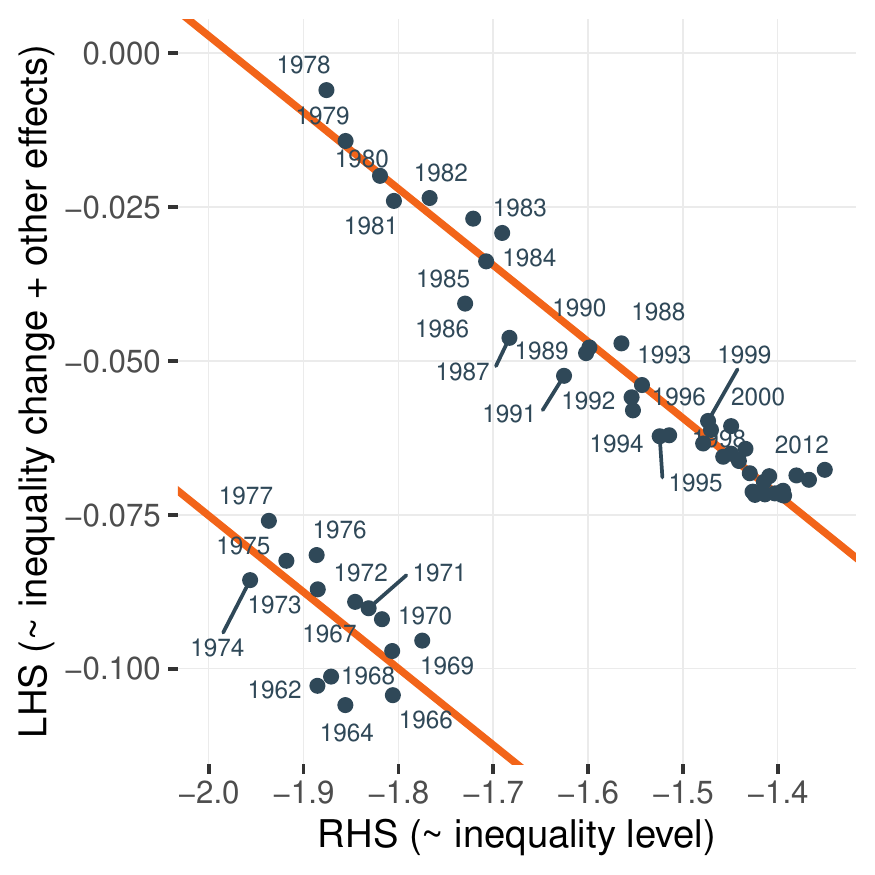}
        
        \caption{Empirical Phase Portrait}
        \label{fig:phase-portrait-data}
        \vspace{1em}
        
        \begin{minipage}{0.9\linewidth}
        \footnotesize \textit{Source:} Author's estimations. \textit{Note:} This scatter plot shows the empirical counterpart to the phase portrait described by equation~(\ref{eq:estimation-complete}), for a level of wealth that correspond to 500 times the average national income (around \$40m in 2022). The two linear lines, fitted for 1962--1977 and 1978--2019, correspond to the relationship defined by equation~(\ref{eq:estimation-complete}). See main text and Appendix~\ref{sec:detailed-estimations} for details on the estimation procedure.
        \end{minipage}
    \end{subfigure}%
    \begin{subfigure}[t]{0.499\textwidth}
    \centering
        \includegraphics[width=\textwidth]{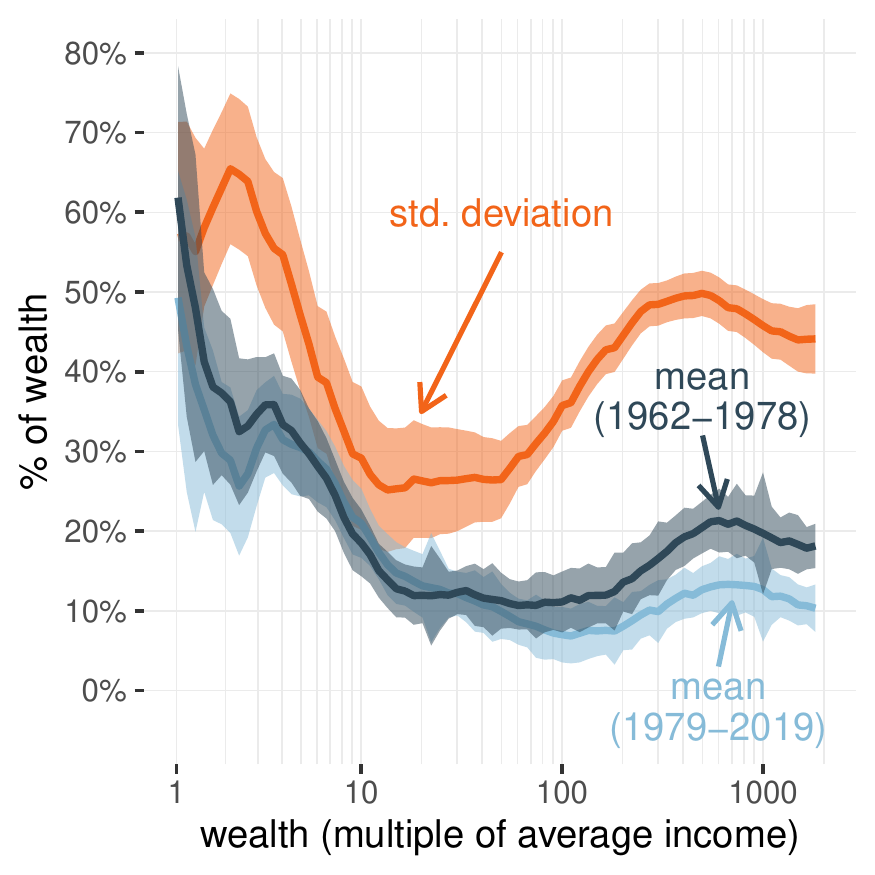}
        
        \caption{Consumption by Wealth}
        \label{fig:diffu-drift-ci-top}
        \vspace{1em}
        
        \begin{minipage}{0.9\linewidth}
        \footnotesize \textit{Source:} Author's estimations. \textit{Note:} This graph shows the average and the standard deviation of consumption by wealth, as estimated from the slope and the intercept of the linear relationships in the left panel, for every wealth bin. Parameters have been adjusted to account for the data transformations, as explained in Appendix~\ref{sec:transform-param-asinh}. Areas around the lines indicate 95\% confidence intervals, estimated using a bootstrap procedure described in Appendix~\ref{sec:std-err}.
        \end{minipage}
    \end{subfigure}
\caption{Estimation of Consumption by Wealth}
\end{center}
\label{fig:propensity-consume}
\end{figure}

Two facts stand out. First, there has been a structural break between 1962--1977 and 1978--2019. Indeed, it is impossible to account for wealth's evolution during both periods by assuming the same linear relationship on the phase portrait: the underlying accumulation process (i.e., the parameters of the propensity to consume) must have changed. Second, within each of these periods, the relationship between the left-hand side and the right-hand side of equation~(\ref{eq:estimation-complete}) is indeed linear. Therefore, a parsimonious model, with a constant mean and variance of consumption by wealth, can account for the trajectory of wealth since 1978, and, separately, for the trajectory between 1962 and 1977. If we focus, for example, on the post-1978 period, we can see the dynamics described in Figure~\ref{fig:phase-portrait} at play. We start in 1978, with a low but rapidly increasing inequality level. But as inequality goes up, the pace at which it increases slows down progressively.

Note that, while there is unmistakable evidence that the intercept of the linear relationship (which captures drift) has changed between periods, there is no clear sign that the slope (which captures mobility) is different.\footnote{This is partly the result of a smaller sample size over 1962--1977.} In light of this, and to get more robust estimates, I assume the same mobility parameter over both periods and only let the drift vary. I apply the same model within all wealth bins: for each of them, I fit two linear relationships with the same slope. I use \citet{deming_statistical_1943} regressions to account the presence of error terms on both sides of equation~(\ref{eq:estimation-complete}). Appendix~\ref{sec:deming-estimation-procedure} provides details of the procedure, alongside robustness checks. I extract the coefficients from these regressions and transform them so that they can be interpreted in terms of the mean and variance of consumption (see Step~2 in Section~\ref{sec:identification-final}, as well as Appendix~\ref{sec:transform-param-asinh} for additional adjustments to account for the inverse hyperbolic sine transform of wealth). This finally allows me to plot Figure~\ref{fig:diffu-drift-ci-top}, the profile of the mean and the variance of consumption by wealth. This figure also displays 95\% confidence intervals, calculated using a bootstrap procedure, which accounts for the presence of error terms on both sides of equation~(\ref{eq:estimation-complete}), as well as autocorrelations across years and wealth bins, described in Appendix~\ref{sec:std-err}.

Several findings emerge from Figure~\ref{fig:diffu-drift-ci-top}. First, the variance of consumption is large, which implies a significant role for mobility in the wealth distribution. Second, on average, people consume a significant fraction of their wealth, even at the top, and even in periods of increasing wealth inequality. This matters, in particular, for our understanding of the wealth distribution in the steady state. Significant consumption levels at the top --- in general exceeding income --- create a tendency for large wealth holdings to reverse toward the mean. In the long run, the reversion towards the mean counterbalances mobility's effect, making it possible for a steady-state distribution to emerge.\footnote{The presence of demographic effects also contributes to the existence of a nondegenerate steady-state.} If consumption at the top were too low, then wealth at the top would grow without bounds, and so would inequality. Note that at no point did I restrict the parameter values to force the existence of a steady state: a nondegenerate steady state arises naturally from the data in a model with constant drift and constant mobility. Finally, we see that changes in the average consumption between 1962--1978 and 1979--2019 are most significant at the top of the distribution (i.e., wealth above 50 times the average income), which aligns with the view that top wealth holders have been the primary drivers of rising wealth inequality. 

\paragraph{Relation to the Rest of the Literature}

\begin{figure}[ht]
    \input{figures/diagram-models.tex}
    \caption{Estimated Parameters and Their Relationship to the Literature}
    \label{fig:diagram-models}
\end{figure}
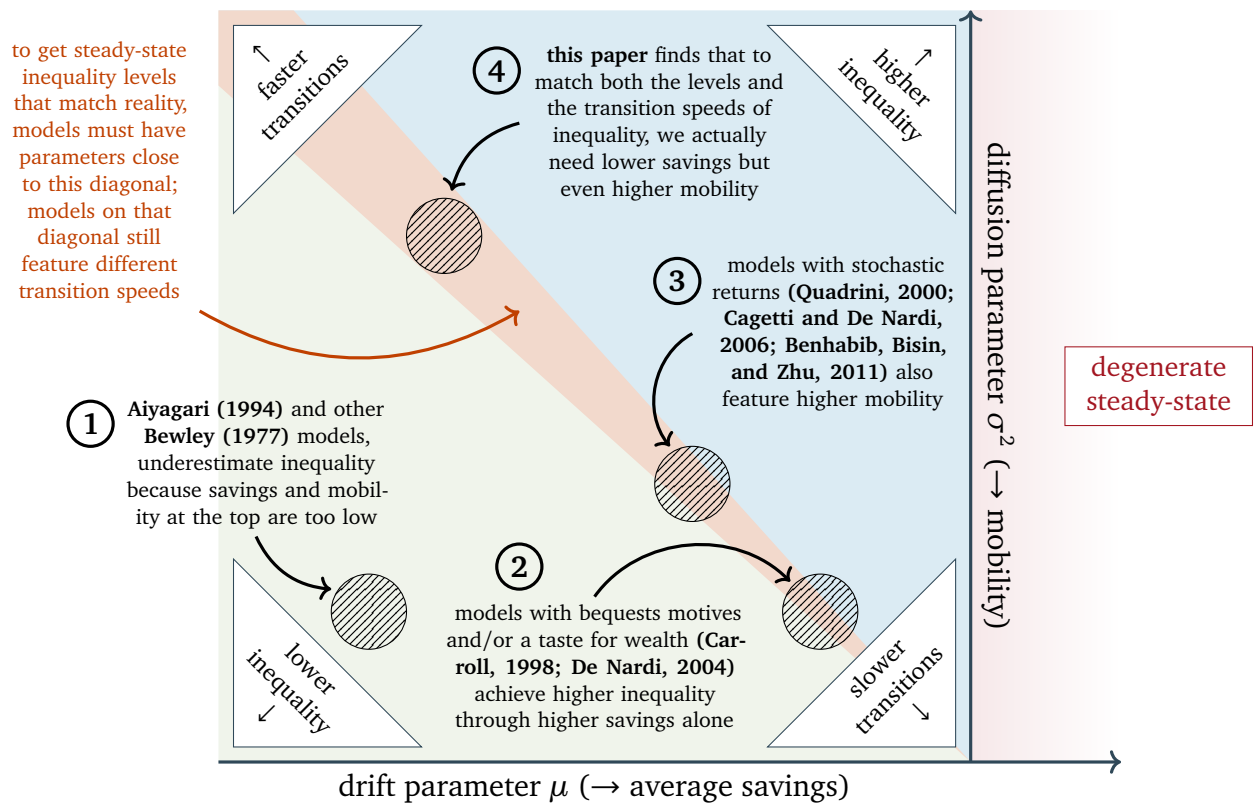

How do these estimates relate to the rest of the literature, and what can we learn from them? We can simplify the situation by focusing solely on the top of the distribution and on the two key parameters (the drift and the mobility) while ignoring the other, less important phenomenons (mobility gradient, demography, etc.) Figure~\ref{fig:diagram-models} summarizes the situation. The various models of the literature can be schematically represented on a two-dimensional plane, where the $x$-axis corresponds to the amount of drift ($\mu$), and the $y$-axis corresponds to the amount of mobility ($\sigma^2$). In this representation, all models with a nondegenerate steady-state lie on the top left quadrant, pictured in Figure~\ref{fig:diagram-models}. The bottom half is not meaningful because it implies negative mobility; the top right quadrant implies an infinite steady-state inequality because there is no reversion towards the mean at the top.\footnote{The drift term $\mu$ is normalized by the economy's growth rate, so it is still possible to have a nondegenerate steady-state if people at the top experience positive wealth growth on average as long as that growth remains below the economy's growth rate. Demography and the mobility gradient are other phenomenons that make it possible to sustain a steady state with positive drift at the top. In any case, it remains true that the emergence of a steady-state requires limited wealth growth at the top.}

What wealth distribution is implied by the different points? At the steady-state, the derivative of the wealth distribution with respect to time is zero, and therefore equation~(\ref{eq:kf-int}) becomes:
\begin{equation}\label{eq:kf-int-steady}
0 = \mu(w) - \frac{1}{2}\sigma^2(w) \frac{\partial_w f(w)}{f(w)}
\end{equation}
This equation characterizes a set of straight diagonal lines for every distribution of wealth, passing through the origin of the plane. Each of them is an ``isoinequality'' line, defining the set of parameter values that lead to the same steady-state distribution of wealth. These isoinequality lines indicate that models can attain any long-run inequality level, either using a high-savings/low mobility regime (points close to the origin) or using a low-savings/high mobility regime (points far away from the origin).

The set of lines that roughly match the inequality levels typically seen in the United States is colored in orange in Figure~\ref{fig:diagram-models}. Combinations of parameters above this line correspond to higher inequality; combinations that lie below, to lower inequality. Importantly, models on the same isoinequality line still differ when it comes to dynamics. High-savings/low-mobility regimes feature slow transitions between steady-states, while the opposite holds for low-savings/high-mobility regimes.\footnote{Figure~\ref{fig:phase-portrait} can demonstrate this. Increasing the slope of the line while keeping the same point of intersection with the $x$-axis leads to the same steady-state but with higher derivatives of the distribution with respect to time (on the $y$-axis), and therefore faster transitions.}

We can now study where the different models in the literature stand and compare them to this paper's estimate. Start the from \citet{aiyagari_uninsured_1994} and similar \citet{bewley_permanent_1977} models (item~1, Figure~\ref{fig:diagram-models}). These models notoriously underestimate inequality for two reasons. First, people in these models accumulate wealth only for precautionary or consumption smoothing motives, so they have no reasons to accumulate the type of large wealth holdings we observe in practice. Second, because everyone earns the same rate of return, mobility at the top is only the result of labor income shocks. Since labor income is small compared to wealth at the top of the distribution, there is limited mobility as well. These facts put these models squarely in the bottom left corner of Figure~\ref{fig:diagram-models}. To fix this problem, a second set of models (item~2, Figure~\ref{fig:diagram-models}) introduced additional saving motives \citep{carroll_why_1998,de_nardi_wealth_2004} such as a taste for wealth, or for bequests. These models manage to match observed inequality levels by increasing savings at the top but do not fundamentally change the extent of wealth mobility. This puts them within the area of realistic steady-state inequality levels (orange diagonal) by moving them to the right of \citet{aiyagari_uninsured_1994} models in the bottom right corner. A third set of models (item~3, Figure~\ref{fig:diagram-models}) also introduces idiosyncratic stochastic returns \citep{quadrini_entrepreneurship_2000,cagetti_entrepreneurship_2006,benhabib_distribution_2011}. This increases mobility at the top because people with the same initial wealth may now move up or down the distribution depending on whether they get high or low returns. That being said, mobility remains quite limited because it is only the result of heterogeneous labor and capital income. Conditional on wealth, however, consumption remains essentially homogeneous because of consumption smoothing. This limited amount of mobility implies slow dynamics, as was identified by \citet{gabaix_dynamics_2016} for income inequality. This paper (item~4, Figure~\ref{fig:diagram-models}) finds that to match the dynamics of inequality that we observe, we need even higher mobility (and consequently lower savings).\footnote{This is a parsimonious alternative to the solutions suggested by \citet{gabaix_dynamics_2016}, which involve the introduction of additional short-run dynamics at the beginning of transition periods.}

We can also the synthetic savings method \citep{saez_wealth_2016,kuhn_income_2020,garbinti_accounting_2021} to this paper. In equation~(\ref{eq:kf-int}), define the synthetic saving as $\tilde{\mu}_t(w) \equiv \mu_t(w) - \frac{1}{2}\sigma_t^2(w) \partial_w f_t(w)/f_t(w)$, and then apply the change of variable $w = Q_t(p)$, where $0<p<1$ is a fractile and $Q_t = F_t^{-1}$ is the quantile function. We get:
\begin{equation*}
\partial_t Q_t(p) = \tilde{\mu}_t(Q_t(p))
\end{equation*}
which indeed corresponds to the traditional definition of synthetic savings. Note, however, that the definition of $\tilde{\mu}_t(w)$ depends on the distribution of wealth, so a more accurate formula would be $\partial_t Q_t(p) = \tilde{\mu}_t(Q_t(p),\partial_p Q_t(p))$. Synthetic saving rates methods can either choose to ignore the dependency on $\partial_p Q_t(p)$, or explicitly eliminate it by setting $\sigma_t(w) = 0$. 

\subsection{Validation}\label{sec:validation}

We can assess the validity and consistency of the model in two different ways. First, we can look at its \textit{internal consistency}. (If we simulate the evolution of the wealth distribution using the estimated parameters, do we reproduce the observed data?) Second, we can look its \textit{external consistency}. (Are the estimated parameters consistent with external observations?) In this section, I address both questions.

\paragraph{Replication of Observed Wealth Inequality Dynamics}

Starting from the distribution of wealth in 1962, and assuming that all the factors which affect the wealth distribution remain at their observed value, I can use the mean and variance of the propensity to consume estimated from the model to simulate the evolution of the wealth distribution. An elementary requirement for the general validity of the approach is that the evolution of the simulated wealth distribution matches the one observed in reality.

\begin{figure}[ht]
\captionsetup[subfigure]{justification=centering}
\begin{center}
    \begin{subfigure}[t]{0.499\textwidth}
        \centering
        \includegraphics[width=\textwidth]{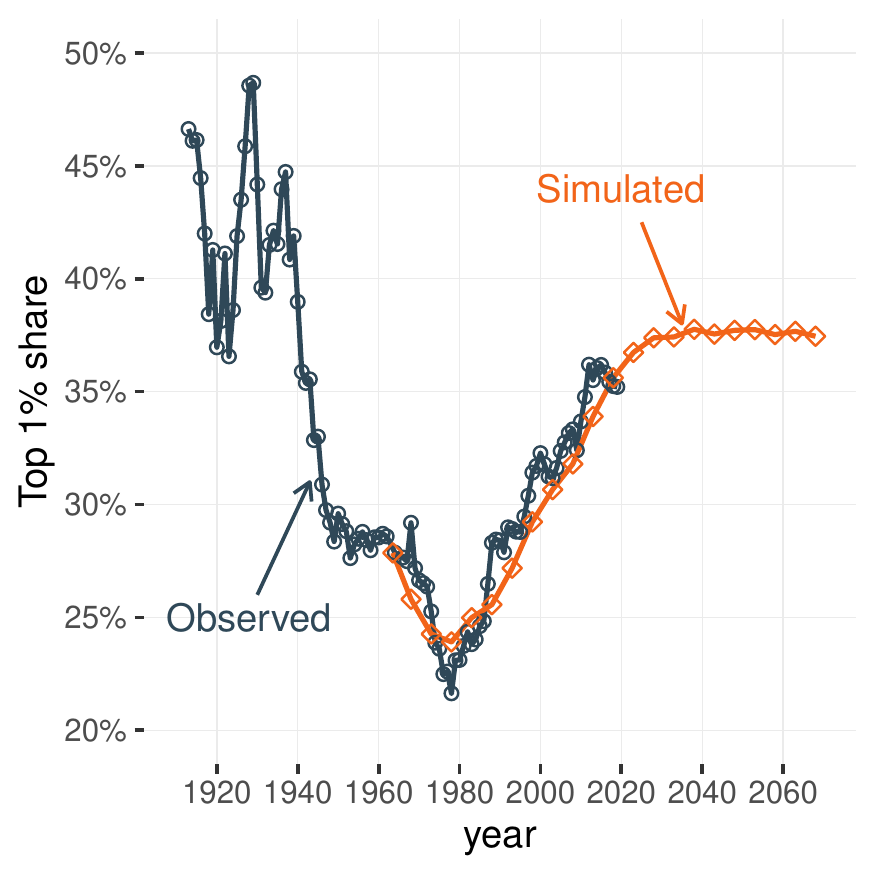}
        \caption{Top 1\% Share}
        \label{fig:actual-simul-top1}
    \end{subfigure}%
    \begin{subfigure}[t]{0.499\textwidth}
        \centering
        \includegraphics[width=\textwidth]{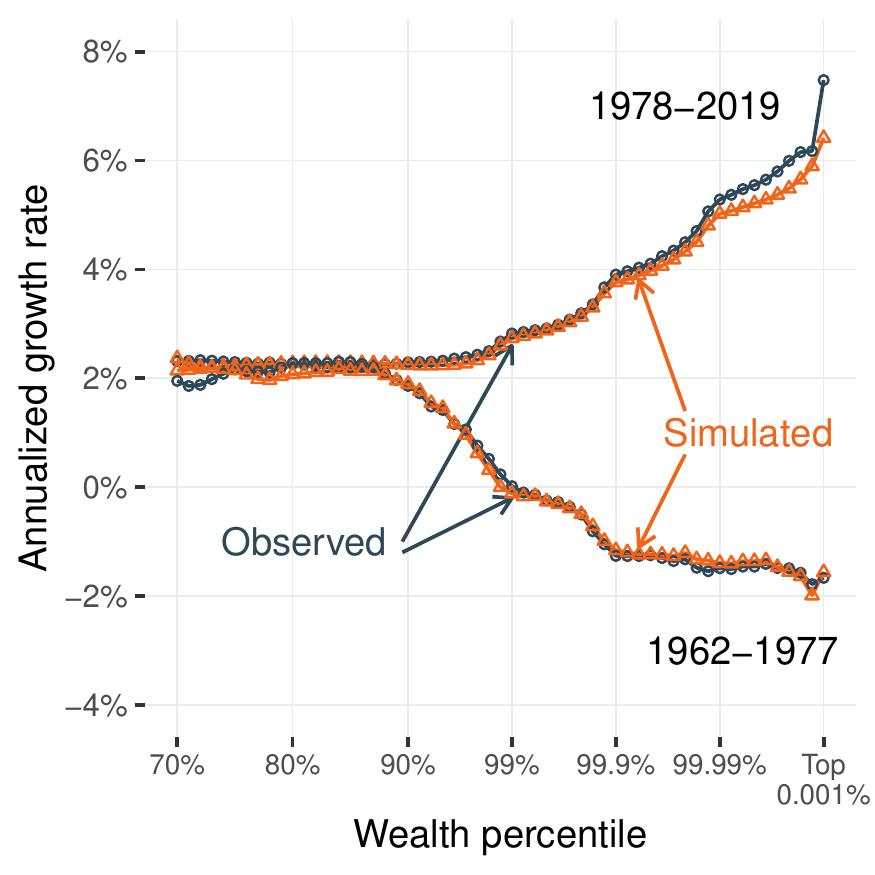}
        \caption{Growth Incidence Curves}
        \label{fig:actual-simul-gic}
    \end{subfigure}

\vspace{1em}
\begin{minipage}{0.8\linewidth}
\footnotesize \textit{Note:} The simulation of the model involves randomly simulated values: to filter out the resulting statistical noise, I simulate the model five times and take the median of the simulations. See main text for details. After 2019, the simulation use the demographic projection (medium variant) from the World Population Prospects \citep{united_nations_world_2019} and otherwise assumes that economic parameters remain fixed at their latest observed values.
\end{minipage}

\caption{Comparison of the Model with Observed Dynamics}
\label{fig:actual-vs-simul}
\end{center}
\end{figure}

Figure~\ref{fig:actual-vs-simul} confirms that this is the case. Figure~\ref{fig:actual-simul-top1} compares the evolution of the top 1\% wealth share in real and simulated data, and shows that we reproduce both the decrease in wealth inequality over 1962--1978, and the increase over 1979--2019. Figure~\ref{fig:actual-simul-gic} goes further by showing the \glspl{gic} for the top 30\% (a group that has consistently owned about 90\% of total wealth since the 1960s).\footnote{Average wealth below the 70th percentile is very low, even zero or negative for some percentiles, and therefore it is not meaningful to calculate their growth rates.} Again, we reproduce the observed growth rates for every percentile (and for fractions of a percentile within the top 1\%), both over 1962--1977 and over 1978--2019.

\paragraph{Consumption}

To evaluate the external validity of the model, I now compare the model's estimates of consumption to external data. It is important to note that direct evidence on the value of these parameters is scant --- which is, in fact, a central motivation for the indirect approach taken in this paper. Nonetheless, we can use two surveys to shed some light on these values. First, there is the \gls{scf}, which is usually cross-sectional but had a panel wave between 2007 and 2009. Then there is the \gls{psid}, which has been recording wealth every five years since 1984.

\begin{table}[ht]
    \centering
    \begin{threeparttable}
    \input{tables/propensity-consumption.tex}
    \begin{tablenotes}
    \footnotesize \textit{Source:} Author's calculations using the \glsfirst{scf} and the \glsfirst{psid}. \textit{Notes:} All numbers are yearly values: data from the \gls{scf} and the \gls{psid}, which are calculated between several years, are rescaled by a factor $\Delta t$ (for the means) and $\sqrt{\Delta t}$ (for the standard deviation) where $\Delta t$ is the number of years, to make estimates comparable. I winsorize the bottom and the top $2.5\%$ of survey-based consumption estimates to limit the impact of measurement error. For the model, values may differ from Table~\ref{tab:decomposition-top1} because they are population-weighted, rather wealth-weighted.
    \end{tablenotes}
    \end{threeparttable}
    \caption{Distribution of the Propensity to Consume: Comparison With Other Sources}
    \label{tab:propensity-consumption}
\end{table}

Because these surveys record wealth and income longitudinally, I can use them to estimate the consumption of each respondent. In principle, I can calculate consumption as the difference between income and the variation of wealth between two consecutive interviews. In practice, doing so involves considerable difficulties, and estimates based on that approach only exist to give rough orders of magnitudes.\footnote{First, the surveys do not record the income earned between interview years, so I have to assume that the respondent's incomes have not changed between interviews. Second, some forms of income, especially accrued capital gains, are not always properly recorded in the surveys, so changes in asset prices might affect the results. Third, since we are calculating consumption as a residual, we are more sensitive to measurement error, which can be a significant issue in surveys. Fourth, surveys record their data at time intervals that are less frequent than this paper's yearly data. I rescale the survey estimates of the mean and the variance of consumption by a factor $\Delta t$, where $\Delta t$ is the time between interviews, to annualize all numbers and make them comparable. But that doesn't fix the fact that we are discretizing the underlying continuous-time process and that the larger time steps lead to coarser approximations.} Table~\ref{tab:propensity-consumption} nonetheless provides the mean and variance of consumption in the surveys, as a fraction of wealth, for three brackets, estimated so as to be as comparable as possible to the model's parameters. Overall, we observe broadly similar numbers. In particular, this exercise confirms our two notable findings: significant consumption levels (including at the top) and important mobility throughout the distribution. One of the largest discrepancies between the model and the surveys concern the mean consumption of the top 1\% which is twice as high in the \gls{scf} than according to the model. This difference could be explained by the fact that the survey was conducted over the great recession and that the consumption estimate in the \gls{scf} might be polluted by asset price declines.

\paragraph{Wealth Mobility}

Another way to check the external validity of my approach is to compare the mobility implied by the model with the mobility we find in the survey data. This approach is less detailed, but more robust than attempts to estimate consumption, because it does not require information on income. Figure~\ref{fig:mobility-rank} compares wealth mobility in the model with the \gls{scf} and the \gls{psid}. On the $x$-axis, I group observations according to their wealth rank, and on the $y$-axis, I plot the distribution of the wealth ranks for each group in the following survey wave, using the median rank and the interquartile range. Then I estimate comparable quantities using the model. For the \gls{scf} (Figure~\ref{fig:mobility-rank-scf}), I estimate mobility over two years to match the frequency of the survey and show results up to the top $0.1\%$. For the \gls{psid}, the interval is five years, and given the smaller sample size and lack of oversampling at the top, I only go up to the top 5\%.

\begin{figure}[ht]
\begin{center}
    \begin{subfigure}[t]{0.499\textwidth}
        \includegraphics[width=\textwidth]{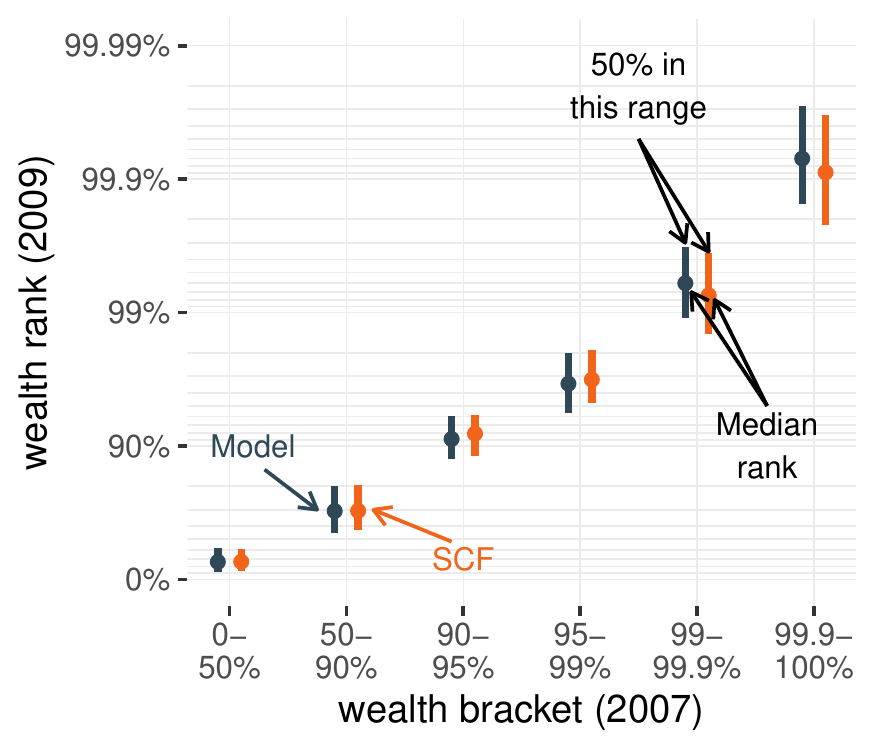}
        \caption{Panel \gls{scf} (2007--2009)}
        \label{fig:mobility-rank-scf}
    \end{subfigure}%
    \begin{subfigure}[t]{0.499\textwidth}
        \includegraphics[width=\textwidth]{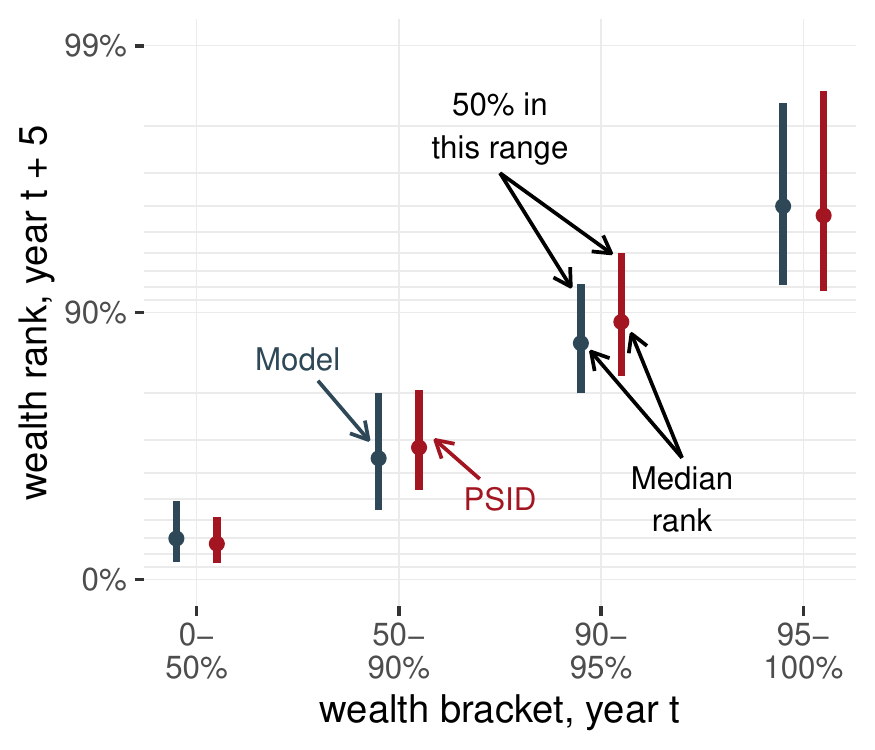}
        \caption{\gls{psid} (1984--2019)}
        \label{fig:mobility-rank-psid}
    \end{subfigure}
    \begin{minipage}{0.8\linewidth}
    
    \vspace{1em}
    \footnotesize \textit{Source:} Own computation using the Panel \glsfirst{scf} (2007--2009) and the \glsfirst{psid} (1984--2019).
    \end{minipage}
\end{center}
\vspace{-1em}
\caption{Comparison of Mobility in the Model with Panel Survey Data}
\label{fig:mobility-rank}
\end{figure}

Once again, the model is broadly consistent with the panel survey data. There is a fair amount of persistence in the wealth rank over time. On average, an observation in wave $n+1$ remains close to its rank in wave $n$. But there is noticeable variability around this central tendency, which shows that there is still significant movement in the wealth distribution over time. The magnitude of this variability, as shown by the interquartile ranges in Figure~\ref{fig:mobility-rank}, is similar in the data and the model.

%% file: figures/diagram-models.tex
\centering
\begin{tikzpicture}

\fill[color=white,left color=RubyRed!10, right color=white] (0,0) rectangle (2,10);
\fill[color=DarkSkyBlue!30] (0, 0) -- (-10, 10) -- (0, 10);
\fill[color=MaximumGreen!10] (0, 0) -- (-10, 0) -- (-10, 10);


\fill[color=Mahogany!20] (0, 0) -- (-9, 10) -- (-10, 10) -- (-10, 9) -- cycle;

\draw[color=Charcoal,fill=white] (-9.8,0.2) -- (-9.8,2.7) -- (-7.3,0.2) -- cycle;
\draw[color=Charcoal,fill=white] (-0.2,9.8) -- (-2.7,9.8) -- (-0.2,7.3) -- cycle;
\draw[color=Charcoal,fill=white] (-9.8,9.8) -- (-7.3,9.8) -- (-9.8,7.3) -- cycle;
\draw[color=Charcoal,fill=white] (-0.2,0.2) -- (-2.7,0.2) -- (-0.2,2.7) -- cycle;

\node[rotate=-45,text width=2cm,align=center,below,scale=0.8] at (-8.6, 1.4) {lower inequality \\ $\downarrow$};
\node[rotate=-45,text width=2cm,align=center,above,scale=0.8] at (-1.4, 8.6) {$\uparrow$ \\ higher inequality};
\node[rotate=45,text width=2cm,align=center,below,scale=0.8] at (-1.42, 1.42) {slower transitions \\ $\downarrow$};
\node[rotate=45,text width=2cm,align=center,above,scale=0.8] at (-8.7, 8.7) {$\uparrow$ \\ faster transitions};

\node[rectangle,scale=0.9,color=RubyRed,fill=white,draw=RubyRed,text width=2.5cm,align=center] at (2.5, 5) {degenerate steady-state};

\draw[->,very thick,color=Charcoal] (-10,0) -- (2,0);
\node[below] at (-5, 0) {drift parameter $\mu$ ($\rightarrow$ average savings)};
\draw[->,very thick,color=Charcoal] (0,0) -- (0,10);
\node[above,rotate=-90] at (0, 5) {diffusion parameter $\sigma^2$ ($\rightarrow$ mobility)};

\draw[->,very thick,color=Mahogany] (-10.25, 6) node[above left,text width=3.5cm,align=center,scale=0.7]{to get steady-state inequality levels that match reality, models must have parameters close to this diagonal; models on that diagonal still feature different transition speeds} to[bend right] (-6, 6);

\draw[very thick] (-11.7, 4.5) node{\textbf{1}} circle (0.3);
\draw[pattern=north east lines] (-8, 2) circle (0.5);
\draw[->,very thick] (-9.5, 3) node[above,text width=5cm,align=center,scale=0.7]{\textbf{\AtNextCite{\color{black}}\citet{aiyagari_uninsured_1994}} and other \textbf{\AtNextCite{\color{black}}\citet{bewley_permanent_1977}} models, underestimate inequality because savings and mobility at the top are too low} to[bend right] (-8.5, 2.3);

\draw[very thick] (-6, 2.6) node{\textbf{2}} circle (0.3);
\draw[pattern=north east lines] (-2, 2) circle (0.5);
\draw[->,very thick] (-5, 2.2) node[below,text width=6.5cm,align=center,scale=0.7]{\AtNextCite{\color{black}} models with bequests motives and/or a taste for wealth \textbf{\citep{carroll_why_1998,de_nardi_wealth_2004}} achieve higher inequality through higher savings alone} to[bend left=50] (-2.4,2.4);

\draw[very thick] (-3.9, 6.4) node{\textbf{3}} circle (0.3);
\draw[pattern=north east lines] (-3.7, 3.7) circle (0.5);
\draw[->,very thick] (-3.7, 5.7) node[right,text width=5cm,align=center,scale=0.7]{\AtNextCite{\color{black}} models with stochastic returns \textbf{\citep{quadrini_entrepreneurship_2000,cagetti_entrepreneurship_2006,benhabib_distribution_2011}} also feature higher mobility} to[bend right=50] (-4, 4.2);

\draw[very thick] (-6.3, 9.2) node{\textbf{4}} circle (0.3);
\draw[pattern=north east lines] (-7, 7) circle (0.5);
\draw[->,very thick] (-6, 8.5) node[right,text width=5cm,align=center,scale=0.7]{\textbf{this paper} finds that to match both the levels and the transition speeds of inequality, we actually need lower savings but even higher mobility} to[bend right=30] (-6.9, 7.6);


\end{tikzpicture}

%% file: tables/propensity-consumption.tex
\begin{tabular}{cccccc}
\toprule
                      &          & \multicolumn{2}{c}{Model}                                                               & SCF                                               & PSID                                               \\ 
                      &          & 1962--1977                                 & 1978--2019                                 & 2007--2009                                        & 1984--2019                                         \\ \midrule
\multirow{4}{*}{\begin{tabular}[c]{@{}c@{}}Mean\\ (\% of wealth)\end{tabular}} & 50--90\% & 39\% & 31\% & 32\% & 38\% \\ \cmidrule(l){2-6}
                      & 90--99\% & 15\% & 15\% & 17\% & 14\% \\ \cmidrule(l){2-6}
                      & Top 1\%  & 13\% & 9\% & 21\% & 12\% \\ \midrule
\multirow{4}{*}{\begin{tabular}[c]{@{}c@{}}Std. Dev.\\ (\% of wealth)\end{tabular}}   & 50--90\% & \multicolumn{2}{c}{50\%}                    & 42\%   & 61\%   \\ \cmidrule(l){2-6}
                      & 90--99\% & \multicolumn{2}{c}{27\%}                    & 31\%   & 41\%   \\ \cmidrule(l){2-6}
                      & Top 1\%  & \multicolumn{2}{c}{37\%}                    & 35\%   & 31\%   \\ \bottomrule
\end{tabular}

%% file: content/5-wealth-inequality.tex
\subsection{Decomposition of Wealth Growth}

I can use equation~(\ref{eq:estimation-complete}) of the model to get a straightforward decomposition of the growth of any part of the wealth distribution. This decomposition is similar to that of \citet{gomez_decomposing_2022}, with the difference being that it is estimated directly from comprehensive wealth data in the United States, and accounts for more factors. To understand the decomposition, define $Q_t(p)$, the p-$th$ quantile of the wealth distribution at time $t$. By definition, $F_t(Q_t(p)) = p$. Taking the derivative of this expression with respect to time, and letting $w=Q_t(p)$ we get that $\partial_t Q_t(p) = -\partial_t F_t(w)/f_t(w)$. Therefore, the left-hand side of equation~(\ref{eq:estimation-complete}) is equal to the variation of the $p$-th quantile. Let $p=99\%$ and let $W_t(p)$ be the average wealth of the top 1\%. We can write the rate of growth of $W_t(p)$ as the average of the growth of the individual fractiles that make up the top 1\%, i.e., $\partial_t W_t(p) = \frac{1}{1-p}\int_p^{+\infty} \partial_t Q_t(r)\,\dif r$. Therefore, if we average the effects on the right-hand side of equation~(\ref{eq:estimation-complete}), we can decompose the growth of the top percentile of wealth.

\begin{table}[ht]
    \centering
    \begin{threeparttable}
    \input{tables/decomposition.tex}
    \begin{tablenotes}
    \footnotesize \textit{Source:} Author's calculations. \textit{Notes:} Growth is adjusted for inflation. The growth numbers correspond to the average annual growth of the top 1\% of the wealth distribution over each period. Values differ from Table~\ref{tab:propensity-consumption} because they are wealth-weighted, rather population-weighted.
    \end{tablenotes}
    \caption{Decomposition of Wealth Growth in the Top 1\%}
    \label{tab:decomposition-top1}
    \end{threeparttable}
\end{table}

Table~\ref{tab:decomposition-top1} shows this decomposition separately for the period of decreasing inequality (1962--1978) and increasing inequality (1979--2019). For each period, I express the total rate of growth of the top wealth percentile as the sum of its different components. One virtue of this presentation is that it puts all the effects we consider here on the same scale --- although they are all conceptually very distinct. We can say, for example, that demography over 1979--2019 had the same effect as a 1.8~pp. decrease in the rate of return would have had.

Two facts stand out in Table~\ref{tab:decomposition-top1}. First, drift and mobility dominate the other factors by far. Each of them have an opposite effect, so when combined, they partially cancel out. But this should not make us overlook that these effects are each very sizable on their own and that the trajectory of wealth inequality is effectively determined by temporary imbalances between the two. Second, the increase in the drift --- which accounts for most of the change in the growth rate of the top 1\% --- is primarily driven by two factors: an increase in capital gains and a decrease in average consumption. The increase in capital gains comes largely from the fact that wealth holders experienced capital losses over 1962--1978. On the other hand, regular capital incomes cannot account for the rise of wealth inequality since they have been lower during 1979--2019 than during 1962--1978.

While Table~\ref{tab:decomposition-top1} provides a valuable overview of the drivers of top wealth growth, it is limited in its ability to tell us how the different factors have truly impacted wealth inequality. Indeed, many of the terms in equation~(\ref{eq:estimation-complete}) are endogenous to the distribution of wealth itself, and wealth evolves through several feedback loops between each side of equation~(\ref{eq:estimation-complete}). For this reason, we cannot simply calculate a counterfactual growth rate for the wealth of the top 1\% by changing specific terms in the decomposition presented in Table~\ref{tab:decomposition-top1}. The next section goes deeper by performing counterfactual simulations where we change some parameters.

\subsection{Counterfactuals}\label{sec:counterfactuals}

In this section, I change various parameters of the wealth accumulation process and observe how the distribution of wealth would have evolved under these different circumstances. I stress that these counterfactuals exist to explore the direct, proximate consequences of the effects under study. In particular, when I consider changes in, say, income or taxation, I do so while explicitly keeping the consumption unchanged, so that I can explore one specific channel at a time. Therefore, this exercise should not be interpreted as a full counterfactual, which would incorporate both direct and indirect effects. It remains a powerful way to explore many mechanisms and can be used to clarify the facts that more exhaustive models would need to match.

\begin{figure}[p]
\vspace*{-2em}
\begin{center}
    \begin{subfigure}[t]{0.499\textwidth}
        \includegraphics[width=\textwidth]{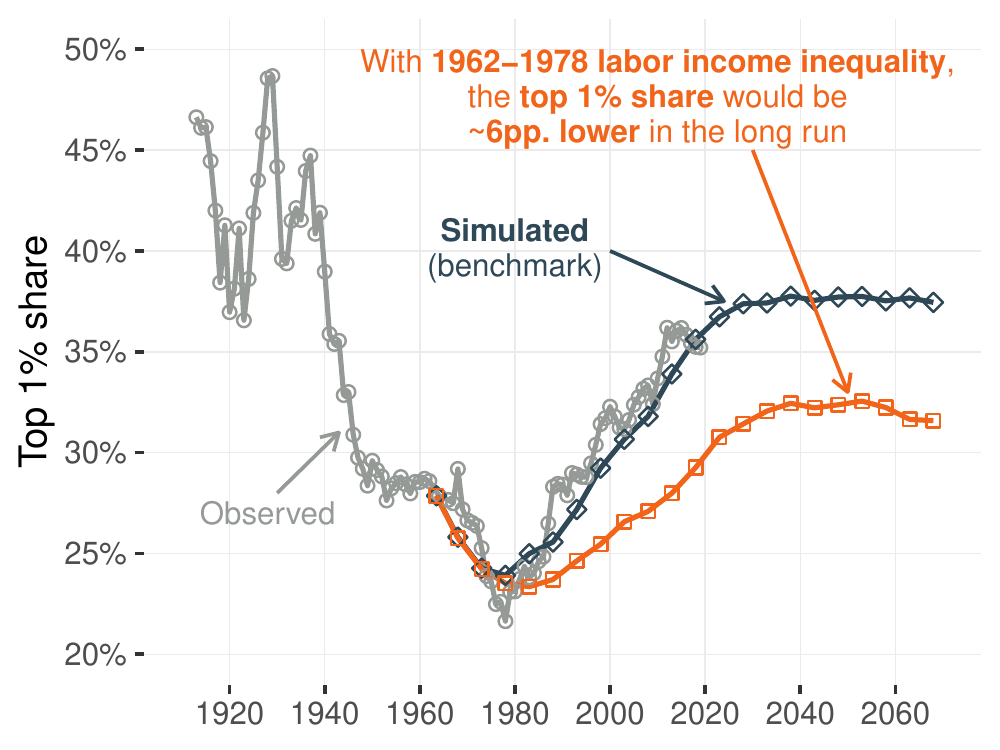}
        \caption{1962--1978 Labor Income Inequality}
        \label{fig:wealth-counterfactuals-labor-income}
    \end{subfigure}%
    \begin{subfigure}[t]{0.499\textwidth}
        \includegraphics[width=\textwidth]{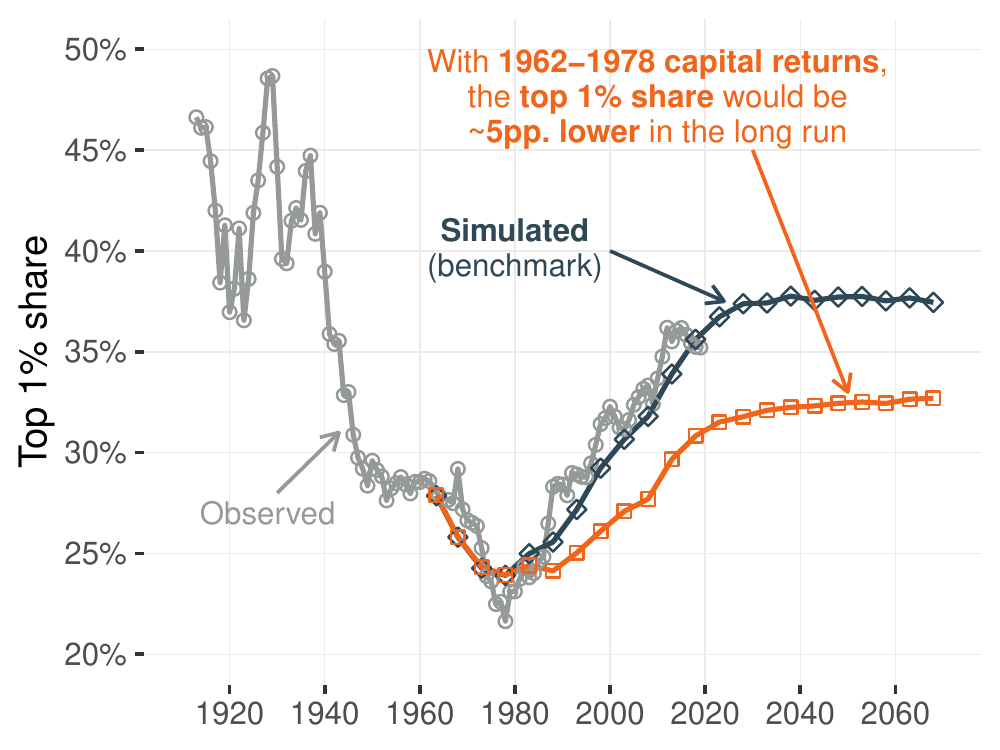}
        \caption{1962--1978 Rates of Return on Wealth}
        \label{fig:wealth-counterfactuals-capital}
    \end{subfigure}
    
    \begin{subfigure}[t]{0.499\textwidth}
        \includegraphics[width=\textwidth]{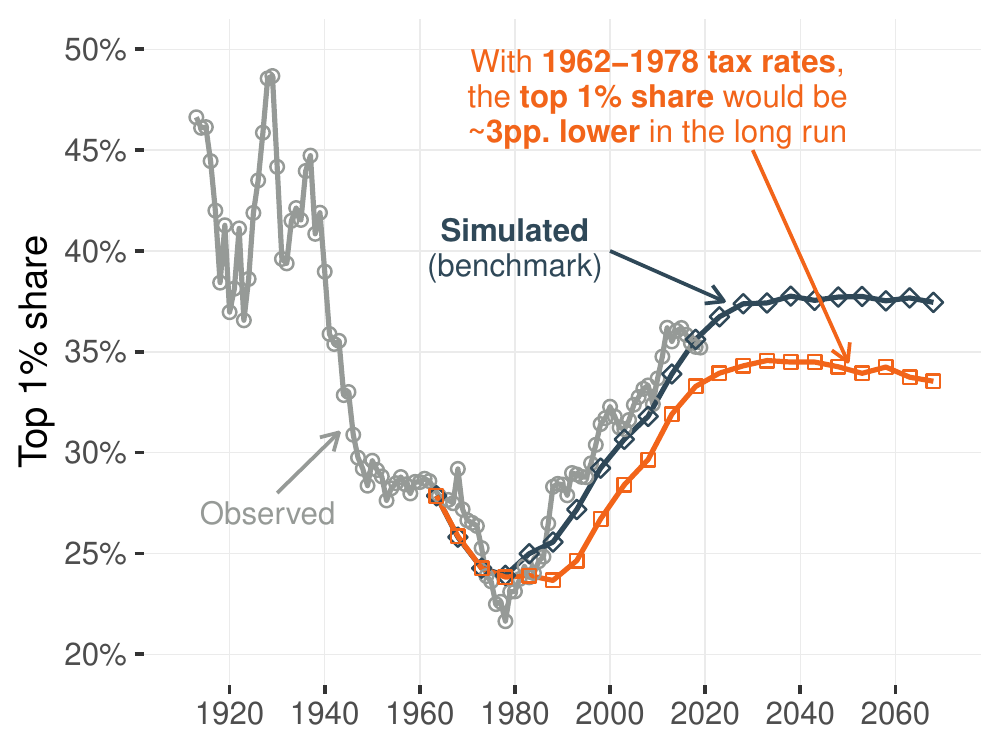}
        \caption{1962--1978 Effective Tax Rates}
        \label{fig:wealth-counterfactuals-taxes}
    \end{subfigure}%
    \begin{subfigure}[t]{0.499\textwidth}
        \includegraphics[width=\textwidth]{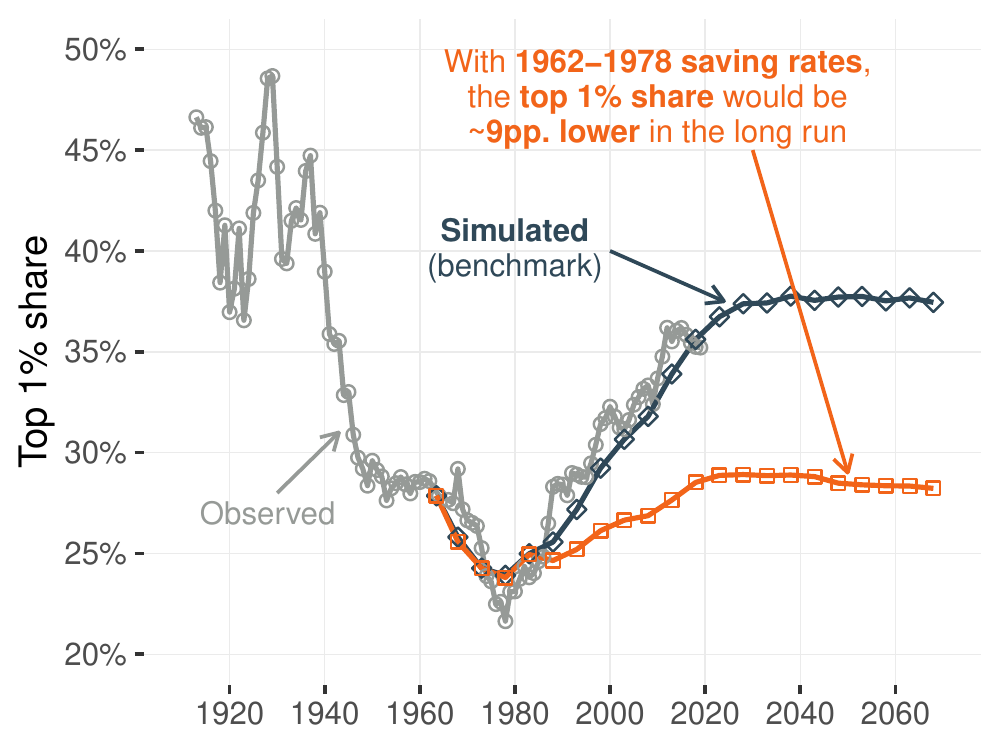}
        \caption{1962--1978 Savings}
        \label{fig:wealth-counterfactuals-conso}
    \end{subfigure}
    
    \begin{subfigure}[t]{0.499\textwidth}
        \includegraphics[width=\textwidth]{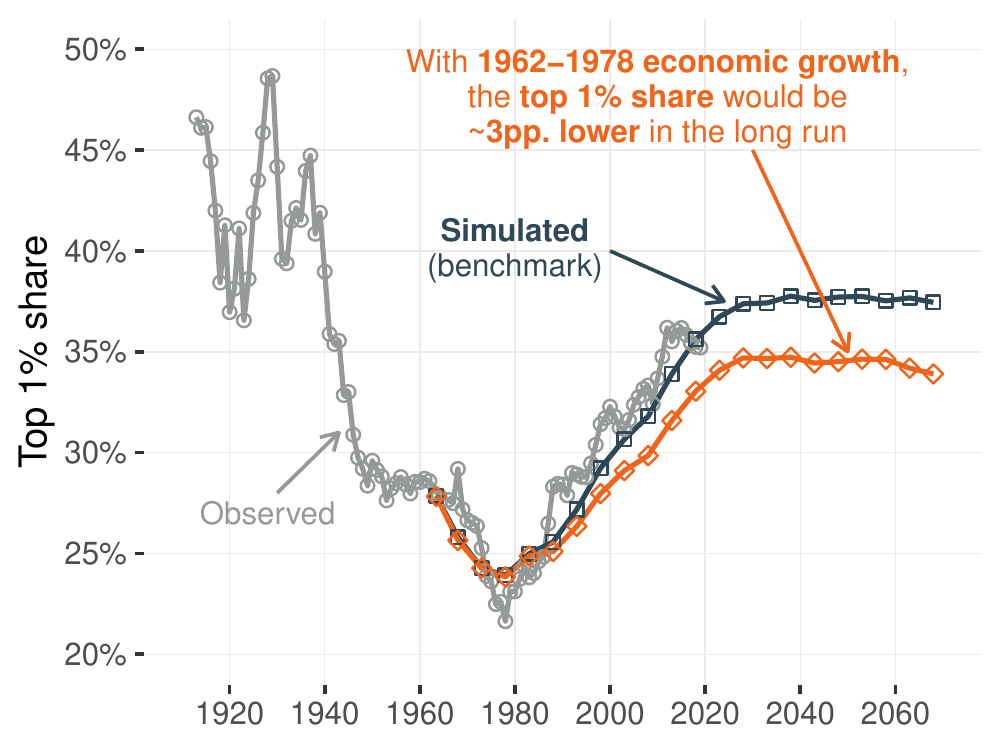}
        \caption{1962--1978 National Income Growth}
        \label{fig:wealth-counterfactuals-growth}
    \end{subfigure}%
    \begin{subfigure}[t]{0.499\textwidth}
        \includegraphics[width=\textwidth]{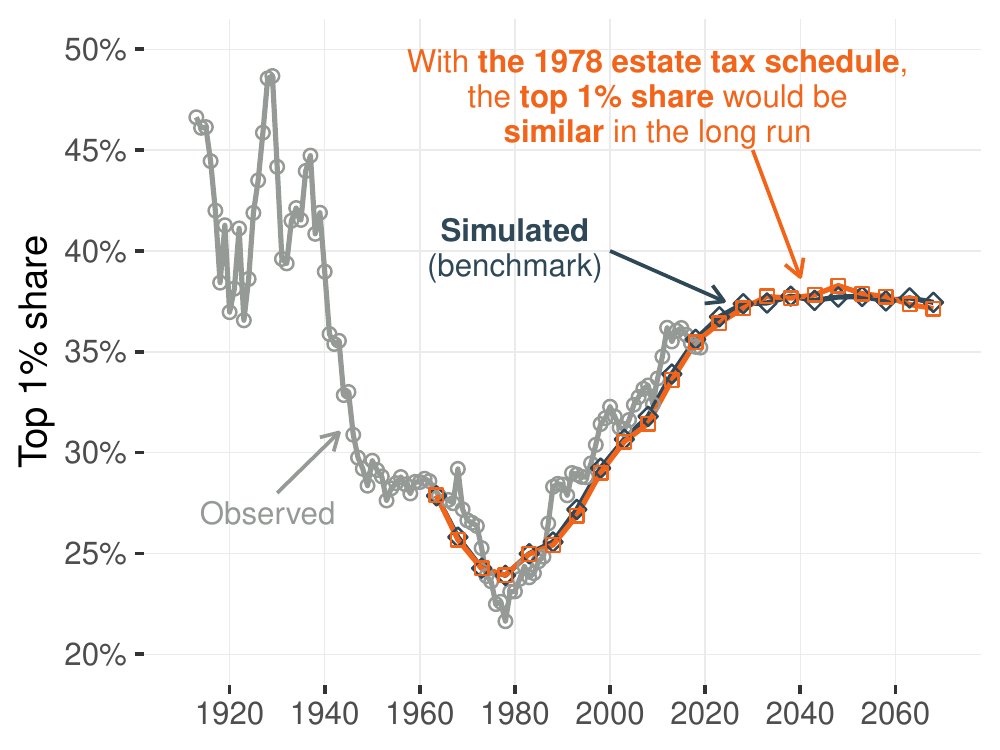}
        \caption{1962--1978 Estate Tax}
        \label{fig:wealth-counterfactuals-estate-tax}
    \end{subfigure}

\vspace{1em}
\begin{minipage}{\linewidth}
\footnotesize \textit{Note:} See main text for details. The benchmark simulation use the demographic projection (medium variant) from the World Population Prospects \citep{united_nations_world_2019} and otherwise assumes that economic parameters remain fixed at their latest observed values. The counterfactual projections change a parameter but leave the others constant. Each model simulation involves randomly simulated values: to filter out the resulting statistical noise, I simulate the model five times and take the median of the simulations.
\end{minipage}

\vspace{0.5em}
\caption{Counterfactual Evolution of Wealth Inequality}
\label{fig:wealth-counterfactuals}
\end{center}
\end{figure}

\paragraph{Benchmark}

Figure~\ref{fig:wealth-counterfactuals} shows the evolution of the top 1\% wealth share according to the model under different scenarios. I run the simulation from the beginning of the data (in 1962) to the year 2070. In each case, I compare the result to a benchmark scenario, which estimates the future evolution of wealth assuming the characteristics of the economy are held fixed after 2019. Specifically, I assume that economic growth remains at its 2010--2019 average and that the distribution of labor income and capital rates of return remain at their 2019 value. I simulate the effect of demography in the future using the projections (medium variant) from the World Population Prospects \citep{united_nations_world_2019}, and otherwise assume that the correlation between age and wealth remains constant after 2019. In this benchmark scenario, the top 1\% wealth stabilizes around 37--38\% in the long run.

\paragraph{Labor Income Inequality}

In Figure~\ref{fig:wealth-counterfactuals-labor-income}, I estimate what the distribution of wealth would look like today if the distribution of labor income had stayed the same after 1978 as it was over the 1962--1978 period. That is, I give people with a given rank in the wealth distribution after 1978 the average mean and variance of labor income from people with the same rank over 1962--1978. This implies that the distribution of labor income is held fixed. I find that the increase in the top 1\% wealth since the 1980s would have been significantly lower under those conditions and would be about 6pp. lower in the long run. The distribution of labor income has thus been an important driver of rising wealth inequality, but remains far from explaining the full rise.

\paragraph{Rates of Return of Wealth}

In Figure~\ref{fig:wealth-counterfactuals-capital}, I maintain the rates of return on capital at their average 1962--1978 value. I find an effect on the top 1\% wealth share that is similar to that of labor income. Importantly, capital gains drive the entire effect. Indeed, as shown in Table~\ref{tab:decomposition-top1}, regular capital returns have actually been lower since 1979 than during 1962--1978. But overall wealth returns have been higher, particularly due to the large capital losses observed during the 1970s. The effect seen in Figure~\ref{fig:wealth-counterfactuals-capital} is in fact similar to what we would obtain by assuming zero capital gains since 1978.

\paragraph{Taxation}

In Figure~\ref{fig:wealth-counterfactuals-taxes}, I focus my attention on taxation. Note that tax rates do not explicitly intervene in my decomposition since its parameters depend directly on the post-tax distribution of labor income and post-tax rates of return. To explore the effects of taxation, I therefore construct alternative distributions of labor income and capital return, for which I assume that pretax distributions evolved as they did in real life, but for which average effective tax rates by wealth percentile are maintained at their average 1962--1978 level. Because tax progressivity has decreased since the 1980s \citep{saez_trends_2020}, the wealthiest people in this scenario end up with less after-tax labor income and lower after-tax rates of return. As a result, they accumulate less wealth. Figure~\ref{fig:wealth-counterfactuals-taxes} shows the outcome of this process. I find that this direct effect of taxation is real but also significantly more muted than the overall effect of labor income or rates of return. This aligns with the fact that the rise in post-tax income inequality is mostly explained by rising pretax inequality and that the decrease in tax progressivity only plays a secondary role.

\paragraph{Savings}

Figure~\ref{fig:wealth-counterfactuals-conso} considers the role of the changes in savings estimated by the model. Of the different effects I consider, this one plays the most prominent role. Assuming the same average consumption as in 1962--1978, the top 1\% wealth share would be about 9pp. lower in the long run. This implies that savings have played a large role in explaining today's levels of wealth inequality: the wealthy are wealthier today than in the 1960s and 1970s in large part because they have saved more.\footnote{Naturally, such changes in consumption could themselves be the result of changes to income and taxation. But the current exercise maintains the distinction between the direct and the indirect effects of income and taxes on wealth inequality.}

\paragraph{Economic Growth}

Economic growth also plays an important role in equation~(\ref{eq:estimation-complete}). Indeed, people at the top of the wealth distribution derive most of their income from capital, so the rate at which they accumulate wealth follows the rate of return $r_t$. On the other hand, people at the bottom of the wealth distribution derive most of their income from labor, so their ability to accumulate wealth follows the growth of labor income which matches $g_t$ in the long run. This mechanism explicitly manifests itself through the fact that wealth accumulation in equation~(\ref{eq:income-consumption}) depends on $r_{it} - g_t$. This dependence of wealth inequality on $r-g$ was emphasized by \citet{piketty_capital_2014,piketty_wealth_2015}. In Figure~\ref{fig:wealth-counterfactuals-growth}, I explore the role that the slowdown of growth since the 1980s had on the wealth distribution. To that end, I increase economic growth by a constant factor over 1979--2019 to match the average growth over 1962--1978. I find a limited but noticeable effect: had the economic growth not slowed down, the top 1\% wealth share would be about 3pp. lower.

\paragraph{Estate Taxation}

Estate taxation has undergone many reforms, leading to a top marginal rate that is half as high today as it was in the 1960s. What role has estate taxation played in the increase in wealth inequality? I explore this issue in Figure~\ref{fig:wealth-counterfactuals-estate-tax}, in which I freeze the estate tax schedule after 1978. I find virtually no impact of this change on the wealth distribution. Two factors account for this finding. First, the evolution of the overall progressivity of the estate tax over the 20th century has actually been more ambiguous than what the trajectory of the top marginal tax rate would suggest (see Figure~\ref{fig:estate-tax} in appendix). The very high marginal tax rates of the 1960s did not kick in until extremely high levels of wealth. Following the reforms in the 1980s, estates in the order of tens of millions of dollars were actually taxed more heavily. Estates needed to reach hundreds of millions of dollars to benefit from the reforms. It wasn't until the 2000s that the estate tax became unambiguously less progressive.

Second, estate taxation has an intrinsically weak impact on the wealth distribution. Weaker, say, than an annual wealth tax of seemingly comparable magnitude. As an illustration, running the model with a radical estate tax (100\% tax on estates above \$100k) would only reduce the top 1\% wealth share by $1.5$pp. in the long run. A naive view would suggest that taxing wealth every year at a net-of-tax rate $1 - \tau$ is equivalent to taxing wealth every $n$ years at a rate $(1-\tau)^n$. This exercise shows this is not the case: taxing wealth once every generation, as the estate tax does, fundamentally alters the nature of the tax. The typical lifetime of a generation is the main determinant of this fact: the longer generations live, the more different the estate tax is from a wealth tax. Section~\ref{sec:inheritance-vs-wealth-tax} develops a simplified model where this finding can be derived analytically.

\paragraph{Other Effects}

The other effects considered in the model (demography, assortative mating) have had a negligible impact on the distribution of wealth.

%% file: tables/decomposition.tex
\begin{tabular}{lrrr}
  \toprule
Effect & 1962--1978 & 1979--2019 & Difference \\ 
  \midrule
Demography & $-1.6\%$ & $-1.8\%$ & $-0.2\%$ \\ 
   [0.5em]Inheritance & $1.6\%$ & $1.1\%$ & $-0.6\%$ \\ 
   [0.5em]Marriages \& Divorces & $-0.07\%$ & $-0.02\%$ & +$0.05\%$ \\ 
   [0.5em]Drift & $-10.1\%$ & $-4.1\%$ & +$6.1\%$ \\ 
  \quad \textit{incl. Labor income} & $2.2\%$ & $2.2\%$ & $-0.007\%$ \\ 
  \quad \textit{incl. Capital income} & $4.7\%$ & $4.5\%$ & $-0.3\%$ \\ 
  \quad \textit{incl. Capital gains} & $-1.8\%$ & $1.0\%$ & +$2.9\%$ \\ 
  \quad \textit{incl. Consumption} & $-15.2\%$ & $-11.8\%$ & +$3.5\%$ \\ 
   [0.5em]Mobility & $9.4\%$ & $9.0\%$ & $-0.4\%$ \\ 
  \quad \textit{incl. Income} & $0.6\%$ & $0.4\%$ & $-0.2\%$ \\ 
  \quad \textit{incl. Consumption} & $8.8\%$ & $8.6\%$ & $-0.2\%$ \\ 
   \midrule
Equals: Average annual growth & $-0.8\%$ & $4.2\%$ & +$5.0\%$ \\ 
   \bottomrule
\end{tabular}

%% file: content/6-taxation.tex
In this section, I use the model of this paper to assess the long-run effect of wealth taxes at the top of the distribution. Understanding this effect is crucial to understanding capital taxation in general and wealth taxation in particular, as recent studies have argued that the long-run elasticity of wealth with respect to the net-of-tax rate is a sufficient statistic for optimal capital taxation \citep{saez_simpler_2018,piketty_theory_2013}. However, while short-run responses to wealth taxation can be measured empirically \citep[e.g.,][]{brulhart_taxing_2016,seim_behavioral_2017,jakobsen_wealth_2020,zoutman_elasticity_2018,ring_wealth_2020,londono-velez_enforcing_2021}, the long-run responses are more elusive.

This paper provides a powerful framework for addressing that question. Indeed, it can estimate counterfactual steady-state wealth distributions, in which we assume different rates of return or different savings. This problem is, effectively, equivalent to estimating a counterfactual wealth distribution under a wealth tax (which decreases post-tax rates of return) with or without behavioral responses (which modify saving rates). A key advantage of this approach --- which is often absent from studies of optimal capital taxation --- is the presence of mobility. Not only is mobility a desirable feature on its own, but it also ensures that the model is well-behaved under a wide range of economic behaviors, because it naturally leads to nondegenerate steady-states for the wealth distribution. This makes it a realistic, easy and tractable way to explore the equity-efficiency trade-offs that are involved in wealth taxation. In contrast, many models without mobility generate infinite responses of capital supply to taxation; under these conditions, the efficiency concerns always dominate the equity concerns, which leaves no room to discuss trade-offs.

One way to use this paper's model is to fully simulate the evolution of the wealth distribution under various wealth tax assumptions, similarly to what was done in Section~\ref{sec:counterfactuals}. This solution is the most complete but also the most computationally demanding. This section pursues an alternative path, where I focus on the long run and use a slightly simplified model. In these conditions, it becomes possible to derive analytical formulas for how the wealth distribution would eventually react to any wealth tax. This makes it easy to draw Laffer curves for wealth taxation and determine what would be the revenue-maximizing wealth tax rate at the top.\footnote{This section does not develop a complete theory of optimal wealth taxation, which would be beyond the scope of this paper. Instead, it focuses on the observable responses of the wealth distribution to a wealth tax. I leave the issue of integrating this framework with a complete optimal tax model for future research.}

\subsection{Theoretical Results}

\subsubsection{Wealth Taxation Without Behavioral Responses}\label{sec:wealth-taxation-no-behav}

Consider the introduction of a nonlinear wealth tax $\tau(w)$ on wealth $w$. I will first assume that there is no behavioral response following the introduction of the wealth tax, so the evolution of wealth is simply characterized by $\dif w_{it} = (\mu(w_{it}) - \tau(w_{it}))\dif t + \sigma_t(w_{it}) \dif B_{it}$. I can state the following result.

\begin{prop}\label{prop:wealth-tax}
Let $f$ be the steady-state density of wealth without a wealth tax. Then the steady-state density $f^*$ with a wealth tax is equal to the steady-state density without a wealth tax, reweighted by a factor $\theta(w)$:
\begin{equation*}
f^*(w) \propto \theta(w)f(w) \qquad \text{where} \qquad \theta(w) = \exp\left\{-\int_{-\infty}^w \frac{2\tau(s)}{\sigma^2(s)}\,\dif s\right\}
\end{equation*}
In particular, if $\tau(w) = \tau(w-w_0)_+$ (i.e., the wealth tax is linear with rate $\tau$ above a threshold $w_0$), and if $\sigma(w) = \sigma w$ for $w \geq w_0$ (i.e., diffusion is proportional to wealth above $w_0$) then $\theta(w)$ simplifies to:
\begin{equation*}
\forall w>w_0 \qquad \theta(w)=\exp\left\{-{\frac{2\tau}{\sigma^2}}\left(\frac{w_0}{w}-1\right)\right\}\left(\frac{w}{w_0}\right)^{-{2\tau/\sigma^2}}
\end{equation*}
and $\theta(w) = 1$ otherwise.
\end{prop}
\begin{proof}
See Appendix~\ref{sec:proof-wealth-tax}.
\end{proof}

That result makes it possible to estimate how the tax base would react to a wealth tax in the long run, effectively by reweighting the steady-state distribution of untaxed wealth using the function $\theta$. The setting mentions the introduction of a new wealth tax where there previously was none, but we could apply the same result to an increase or a decrease of an existing wealth tax by redefining $\tau$ as a change in the rate of the wealth tax.\footnote{This result does not consider how the government uses the wealth tax. We can adapt the formula to include the redistribution of a lump-sum amount and then solve an equation numerically to ensure that we redistribute as much as we tax. In practice, the impact would be negligible as long as we focus on the top of the distribution. Indeed, much more wealth is taxed in that part of the distribution than would be redistributed. The lump sum rebate would have an impact at the bottom, however. Section~\ref{sec:lump-sum-rebate} in appendix addresses this question.}

The result emphasizes the role of mobility, as explained by \citet{saez_progressive_2019}. The impact on the tax base depends on $\tau(w)/\sigma^2(w)$ and not just $\tau(w)$. Doubling the parameter $\sigma(w)$ quadruples the parameter $\sigma^2(w)$, which implies that a tax rate four times as high would lead to the same change in the tax base. The intuition is the same as in \citet{saez_progressive_2019}: high mobility means that people only get taxed for a short period and that new, previously untaxed wealth keeps entering the tax base. As a result, the tax base does not react too much to wealth taxation. When mobility goes to zero, however, the same wealth from the same people is taxed repeatedly so that the tax base eventually goes to zero.

Proposition~\ref{prop:wealth-tax} carries one important difference compared to the result of \citet{saez_progressive_2019}. In the second part of the result, the reweighting factor is the product of two terms: $\exp\{-{{2\tau}({w_0}/{w}-1)/{\sigma^2}}\}$ and $({w}/{w_0})^{-{2\tau/\sigma^2}}$. The first impacts the distribution near the threshold $w_0$ while the second impacts the distribution away from the threshold. The discrete-time formula of \citet{saez_progressive_2019} only contains an equivalent of the second term, $({w}/{w_0})^{-{2\tau/\sigma^2}}$. This difference is because they consider the taxation of billionaires not at a marginal rate $\tau$, but at an average rate $\tau$. This distinction is important because the long-run mechanical effect of a wealth tax explicitly depends on the average tax rate and not, like behavioral responses, on the marginal rate. At the very top of the distribution, the distinction between the average and the marginal rate becomes negligible, hence the term $({w}/{w_0})^{-{2\tau/\sigma^2}}$ similar to the formula in \citet{saez_progressive_2019}. But close to the threshold $w_0$, the distinction does matter and, in fact, has important consequences for the overall response of the tax base to the tax, as we will see in Section~\ref{sec:calibrations-us}.

\subsubsection{Behavioral Responses}

\paragraph{Tax Evasion and Tax Avoidance} 

People can react to a wealth tax by hiding some of their wealth, either through tax evasion or tax avoidance. Assume that, in response to a marginal tax rate $\tau'(w)$, people only report a fraction $\alpha(w)=[1-\tau'(w)]^\varepsilon$ of their wealth. The parameter $\varepsilon$ is the elasticity of declared wealth to the marginal net-of-tax rate $1-\tau'(w)$. For a small rate $\tau'(w) \ll 1$, people react by approximately hiding a fraction $\tau'(w)\varepsilon$ of their wealth. When $\varepsilon = 0$, people truthfully report all of their wealth. As $\varepsilon$ goes to infinity, people start hiding all of their wealth to avoid paying the tax. With tax avoidance, people that own $w$ in wealth pay $\tau(w)\alpha(w)$
instead of $\tau(w)$. In effect, this is equivalent to having a wealth tax with a lower rate. Therefore, the results for the purely mechanical model hold with minimal modifications.

\paragraph{Consumption}

People may also react to a wealth tax by accumulating less wealth. Changes to savings have different implications than tax evasion. Indeed, tax evasion affects both the dynamic of wealth and the tax base. Savings, on the other hand, affect the dynamic of wealth but do not directly reduce the tax base.

Theory provides few constraints regarding how a wealth tax ought to affect saving rates, given the many settings and mechanisms we could consider. In broad terms, there can be a substitution effect that decreases savings (because a wealth tax makes deferring consumption more expensive). And there can be an income effect that increases savings (because a wealth tax makes people poorer, and they compensate by investing more). In the case of, say, labor supply, the widely accepted view that substitution effects dominate. No such consensus exists for savings.

The following reduced-form specification can nonetheless account for the overall effect in a direct and intuitive way. Assume that, in response to a tax with marginal rate $\tau'(w)$ on wealth, people increase their consumption by a factor $\beta(w)=[1 - \tau'(w)]^{-\eta}$. The parameter $\eta$ captures the elasticity of consumption with respect to the marginal net-of-tax rate $1-\tau'(w)$.\footnote{I will ignore the cases where $\eta<0$ (i.e., income effects dominate), even though they are a theoretical possibility and even though some studies find this result \citep{ring_wealth_2020}. Indeed, it is problematic to assume in a taxation context that the tax base responds positively to the tax. Moreover, the elasticity has to change sign at some point, otherwise, a $100\%$ wealth tax would correspond to infinite savings. However, if true, it would imply that wealth tax rates could be higher.} The drift in the dynamic of wealth now includes an additional term $c(w)[1 - \beta(w)]$ where $c(w)$ is the average propensity to consume. The savings behavioral response amplifies the impact of the wealth tax.

\paragraph{Complete Model}

Behavioral responses change the drift term, which is analogous to a change in the effective rate of the wealth tax. Therefore Proposition~\ref{prop:wealth-tax} can be directly extended to account for behavioral effects. The reweighting factor in the full model becomes:
\begin{equation*}
\theta(w) = \exp\left\{-\int_{-\infty}^w \frac{2\tau(s)\alpha(s)}{\sigma^2(s)}\,\dif s -\int_{-\infty}^w \frac{2c(s)[1 - \beta(s)]}{\sigma^2(s)}\,\dif s\right\}
\end{equation*}

\subsection{Calibrations for the United States}\label{sec:calibrations-us}

\paragraph{Baseline Wealth Model}

To illustrate the formulas of this section, I will consider the case of a linear tax on estates above \$50m. I consider a model of wealth accumulation with a mobility parameter and an average propensity to consume that match this paper's estimates for the United States over 1979--2019. I directly use the 2019 data as an estimate of the steady-state wealth distribution, given that the simulation in Section~\ref{sec:counterfactuals} suggest that wealth inequality is close to its long-run value.

\paragraph{Behavioral Elasticities}

To calibrate $\varepsilon$ and $\eta$, I rely on the recent empirical literature that exploit various quasi-experimental settings to assess behavioral reactions to a wealth tax.

Several of these papers present bunching evidence \citep{seim_behavioral_2017,londono-velez_enforcing_2021,jakobsen_wealth_2020}. Bunching provides the cleanest estimates of pure tax avoidance elasticity. Indeed, the true value of wealth in the short run tends to follow unpredictable asset movements so it would be hard for a household to precisely bunch at kink points. \citet{seim_behavioral_2017} finds an elasticity of $0.5$ in Sweden, and \citet{jakobsen_wealth_2020} find elasticities that are even lower in Denmark. \citet{londono-velez_enforcing_2021} find a higher estimate (2--3) in Colombia.

As their main identification strategy, \citet{jakobsen_wealth_2020} pursue a difference-in-difference approach that exploit various tax reforms. This allows them to compute elasticities that incorporate dynamic and saving responses over larger time spans. Over an 8-year time frame, they find a sizable elasticity at the top of about $18$ with respect to the net-of-tax rate. The authors argue that most ($90\%$) of it can be attributed to a behavioral effect (as opposed to a mechanical effect). Assuming that the elasticities cumulate multiplicatively over time, this would correspond to a yearly behavioral elasticity of $1.4$ for both the saving and tax avoidance response. \citet{zoutman_elasticity_2018}, also using a difference-in-difference strategy, finds a much higher elasticity of almost 14. \citet{seim_behavioral_2017} also analyzes saving responses to a wealth tax but does not find any. \citet{ring_wealth_2020} exploits geographic discontinuities in the exposure to wealth taxation to estimate savings responses and actually finds an \textit{increase} in savings in response to wealth taxation. \citet{brulhart_taxing_2016} find a much higher overall elasticity (23--34) in Switzerland using both between canton variations of the tax rate and within variation in the Bern canton. They also look at bunching evidence but find much lower effects there.

Note that the tax avoidance elasticity is not a pure structural parameter, but also results from how strongly a wealth tax is enforced. For the baseline calibration, I will consider a limited tax avoidance response ($\varepsilon = 1$), which is around the values found by \citet{seim_behavioral_2017}, \citet{londono-velez_enforcing_2021} and \citet{jakobsen_wealth_2020}. I will also consider a medium savings response ($\eta = 1$), in line with \citet{jakobsen_wealth_2020}, and around the median of existing studies. I will also consider an alternative scenario with a higher saving response ($\eta = 3$) and a higher tax avoidance response ($\varepsilon = 3$).\footnote{Very large behavioral responses, as found by \citet{brulhart_taxing_2016}, are of limited interest for the current exercise, since they make long-run dynamic effects negligible compared to the immediate static effect of behavioral responses.}

\paragraph{Results}

\begin{figure}[ht]
     \centering
     \begin{subfigure}[b]{0.499\textwidth}
         \centering
         \includegraphics[width=\textwidth]{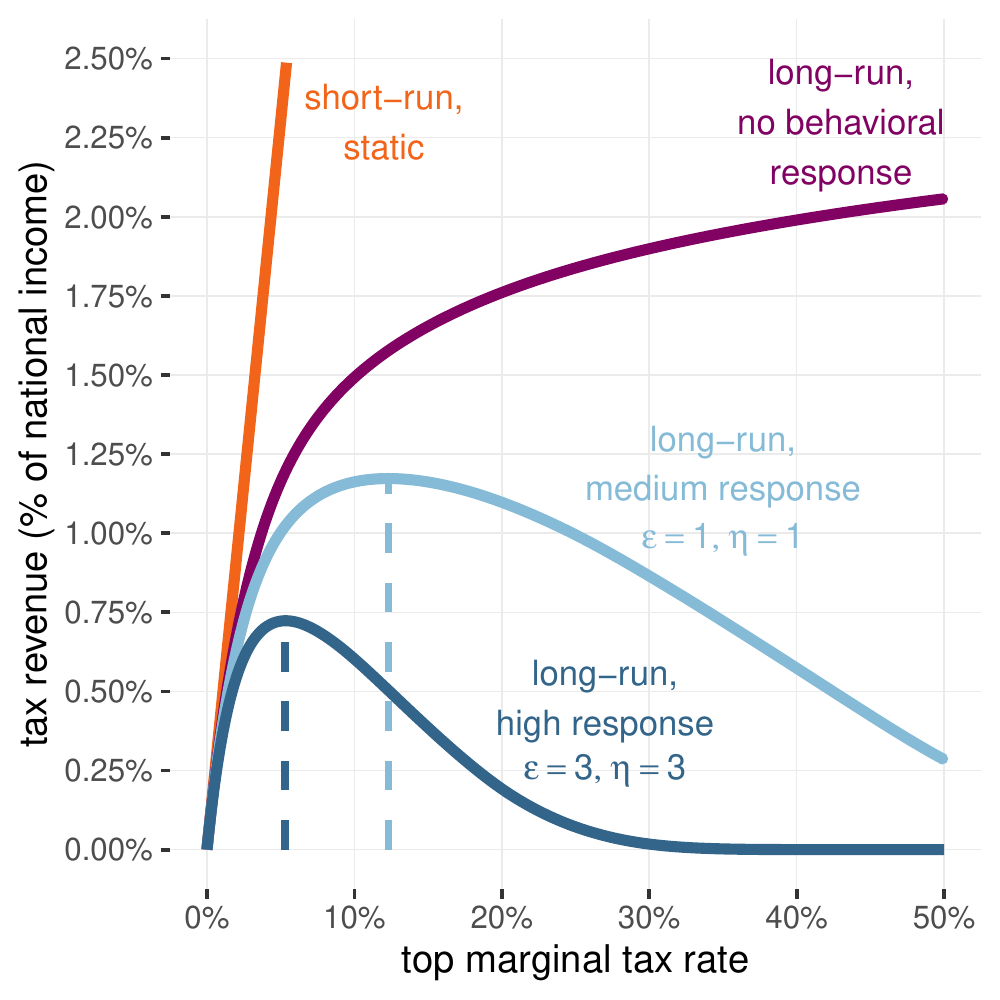}
         \caption{Laffer Curves}
         \label{fig:laffer-curves}
     \end{subfigure}%
     \begin{subfigure}[b]{0.499\textwidth}
         \centering
         \includegraphics[width=\textwidth]{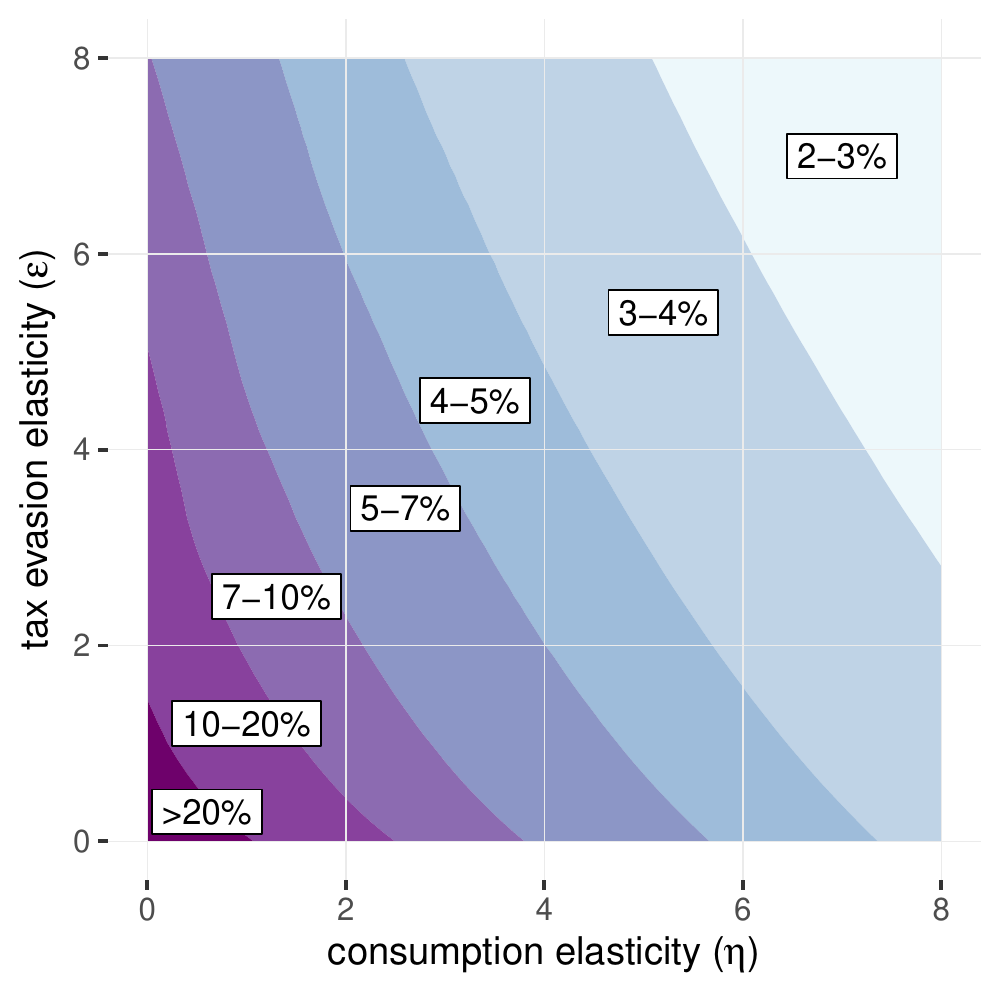}
         \caption{Revenue-Maximizing Tax Rates}
         \label{fig:abacus}
     \end{subfigure}
    \caption{Effects of a Wealth Tax on Estates Above \$50m}
    \label{fig:wealth-tax}
\end{figure}

Figure~\ref{fig:laffer-curves} shows the Laffer curves for the wealth tax, i.e., the amount of tax revenue raised as a function of the tax rate. As a point of reference, we can start from the short-run, static revenue estimate (in orange). This estimate assumes the tax base to be entirely nonreactive, and as a result, the wealth tax can raise a very large amount of revenue. The second curve (in purple) accounts for how the wealth tax progressively eats away its tax base in the long run, even without behavioral responses. As argued by \citet{saez_progressive_2019}, mobility is the central factor determining that curve's shape. In a model without mobility \citep[e.g.,][]{jakobsen_wealth_2020} this curve would be equal to zero for all positive tax rates. This dynamic mechanical effect is sizable and significantly lowers tax revenue in the long run. However, it is not strong enough to produce the usual, inverted U-shaped Laffer curve: the tax base never goes to zero, even with a 100\% tax rate. As a consequence, that model is insufficient to generate guidance on the revenue-maximizing rate.\footnote{\citet{saez_progressive_2019}, who consider a somewhat analogous model, do obtain an estimate for the revenue-maximizing rate. This is because in their model, top wealth holders are taxed at an average, not marginal, rate $\tau$ (see also Section~\ref{sec:wealth-taxation-no-behav} for a discussion).} The reason behind the result is that the dynamic mechanical effect depends on the average tax rate and not, like behavioral responses, on the marginal rate. Hence, close to the exemption threshold, the impact of the tax remains limited even with very high rates. So, in theory, it remains possible to keep raising revenue by increasing the tax rate \textit{ad infinitum}. To get a more realistic and actionable model, we need to incorporate behavioral responses, which I do for the two last curves (in blue). In the benchmark calibration, the revenue-maximizing rate remains high ($\approx 12\%$), but the amount of revenue raised is only a fraction of what a static scoring finds.

There is still a fair amount of uncertainty regarding the value of behavioral elasticities, so in Figure~\ref{fig:abacus}, I estimate the revenue-maximizing rates for a wide range of elasticity values. A notable result here is that the consumption elasticity has a stronger impact than the tax avoidance elasticity. This is because tax avoidance actually has an ambiguous effect on the tax base. On the one hand, it directly lowers the tax base since people underreport their assets. But on the other hand, it increases the post-tax rate of return, allowing people to accumulate more, which grows the tax base in the long run.

\subsection{Comparison with Inheritance Taxation}\label{sec:inheritance-vs-wealth-tax}

How does the annual taxation of wealth compare to inheritance taxation? In simulations, I found in Section~\ref{sec:counterfactuals} that the estate tax has a limited impact on the distribution. To illustrate and clarify this finding, I develop a simple model that compares the mechanical effects of an annual wealth tax to the effects of an inheritance tax. Assume that wealth $w$ evolves according to the following \gls{sde}, where wealth is taxed at a constant rate $\tau$, and $\mu$ is the rate of wealth growth without tax:
\begin{equation*}
\dif w = (\mu - \tau) w \dif t + \sigma w \dif B_t
\end{equation*}
Assume a reflecting barrier at the bottom, so that wealth cannot go below $w_0$. This is the simplest model that generates a Pareto distribution of wealth, with a \gls{pdf} of the form $f(w) = \frac{\alpha w_0^\alpha}{w^{\alpha+1}}$ \citep{gabaix_power_2009}.

Assume that people die at a Poisson rate $\delta$. When they do, their only child pays the estate tax, inherits their wealth, and replaces them in the distribution. Let $(1-\chi)^{1/\delta}$ be the net-of-tax rate of the estate tax. (This specification makes the rates $\tau$ and $\chi$ \textit{a priori} comparable.) At the steady-state, the wealth density $f$ satisfies the \citet{kolmogorov_uber_1931} equation with the Poisson process included, and with the time the derivative set to zero:
\begin{equation*}\small
0 =
\,-\, \underbrace{(\mu - \tau)\partial_w w f(w)}_{\text{drift}}
\,+\, \underbrace{\frac{1}{2}\sigma^2\partial^2_w w^2 f(w)}_{\text{mobility}}
\,-\, \underbrace{\delta f(w)}_{\text{deaths}}
\,+\, \underbrace{\frac{\delta}{(1-\chi)^{1/\delta}} f\left(\frac{w}{(1-\chi)^{1/\delta}}\right)}_{\text{births}}
\end{equation*}
The steady-state still follows a Pareto distribution. Substituting $f$ with its value $f(w) = \frac{\alpha w_0^\alpha}{w^{\alpha+1}}$, we get the following equation for the Pareto coefficient $\alpha$, which is a measure of inequality (higher values mean less inequality):
\begin{equation}\label{eq:alpha}
0 = \mu +\frac{1}{2} (\alpha -1)  \sigma ^2 - \tau - \frac{\delta}{\alpha}[1 - (1 - \chi)^{\alpha/\delta}]
\end{equation}

\begin{figure}[ht]
    \centering
    \begin{minipage}{0.7\linewidth}
    \includegraphics[width=\linewidth]{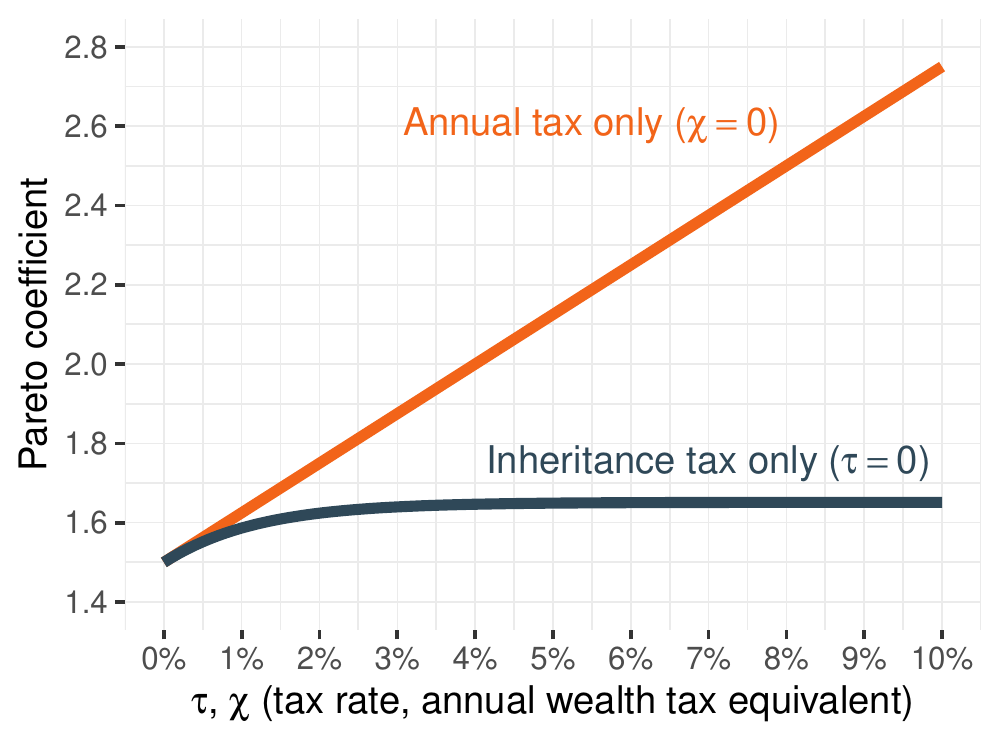}
    \caption{Pareto Coefficients under an Annual Wealth vs. an Inheritance Tax of Comparable Magnitude}
    \label{fig:simple-model-inheritance}
    \end{minipage}
\end{figure}

Consider the following calibration: $\mu=-0.04$, $\sigma=0.4$ and $\delta=1/50$. We can solve this equation numerically to get the steady-state Pareto coefficient under any value of $\tau$ and $\chi$. Figure~\ref{fig:simple-model-inheritance} shows the results from this exercise. It compares the Pareto coefficients for two scenarios: an annual wealth tax (without inheritance tax) and an inheritance tax (of comparable magnitude, without an annual wealth tax). With an annual wealth tax, the Pareto coefficient increases linearly with the tax rate (therefore, inequality decreases). With an inheritance tax, not only does the Pareto coefficient increase much less overall, but as the rate increases, the Pareto coefficients quickly reach an asymptote. Hence, an inheritance is incapable, no matter how large, of reducing inequality to the same extent that an annual wealth tax can.

Equation~(\ref{eq:alpha}) provides the intuition behind that result. For a small estate tax rate $\chi$, a first-order approximation gives $\frac{\delta}{\alpha}[1 - (1 - \chi)^{\alpha/\delta}] \approx \chi$, and therefore the estate tax has an effect similar to a wealth tax $\tau$. But for higher rates $\chi$, this is no longer true. In practice, the approximation does not work unless $\chi$ is extremely small because the ratio $\alpha/\delta$ is high. As the estate tax rate $\chi$ increases, the Pareto coefficients converge to a maximum value $\alpha_1$, which we obtain by solving~(\ref{eq:alpha}) for $\chi = 1$:
\begin{equation*}
\alpha_1 = \frac{1}{2}\left\{\alpha_0 + \sqrt{\alpha_0^2 + \frac{8\delta}{\sigma^2}}\right\}
\end{equation*}
where $\alpha_0 = 1 - 2(\mu - \tau)/\sigma^2$ is the Pareto coefficient in the absence of inheritance tax. The estate tax is, therefore, the least efficient when the ratio $8\delta/\sigma^2$ is small, in which case the Pareto coefficient for $\chi = 1$ is almost the same as the coefficient for $\chi = 0$. This can happen for two reasons: $\delta$ being low and $\sigma^2$ being high. When $\sigma^2$ is high, there is a lot of mobility within generations, so estate taxation is ineffective because wealth quickly recovers from it. But this effect also exists for annual wealth taxes (see Section~\ref{sec:wealth-taxation-no-behav}), so it does not drive the gap between annual and inheritance taxes.  When $\delta$ is low, people die and children inherit their wealth at a low rate, so estate taxation is ineffective because it happens too rarely. We can also see this effect by looking directly at~(\ref{eq:alpha}): taking the estate tax rate $\chi$ to 100\% has the same effect as an annual wealth tax with a rate $\delta/\alpha$. So for $\alpha \approx 1.5$, a 100\% estate tax, which occurs every 50 years on average, is similar to a 3\% annual wealth tax, and cannot get higher.\footnote{Intuitively, $\delta$ creates a cap on the potency of the estate tax: when it is low, wealth is rarely transmitted, and so its effect is limited. The parameter $\alpha$ modulates that cap. When $\alpha \rightarrow 1$, inequality tends to infinity. Only one dynasty owns meaningful wealth, and the estate tax is maximally efficient because it can repeatedly affect that one dynasty. On the other hand, when $\alpha \rightarrow +\infty$, inequality tends to zero. Everyone owns a similar amount of wealth. The estate tax struggles to discriminate between the richest and the poorest in these conditions.}

%% file: content/7-conclusion.tex
This paper introduces a new way of decomposing the evolution of wealth inequality. This decomposition accounts for the different processes affecting the wealth distribution, including mobility. It only requires repeated cross-sections and can therefore be applied to available historical data.

Applying this decomposition to the United States, I find that the rise of wealth inequality has been driven by higher savings at the top, higher rates of capital gains, and higher labor income inequality, with other factors playing a minor role. Applying the framework to the study of wealth taxation, I derive simple formulas that describe how the wealth distribution would react to a wealth tax in the long run. I use them to estimate revenue-maximizing tax rates for a linear tax on high estates. The website {\small\url{https://thomasblanchet.github.io/wealth-tax/}} provides a simulator to apply these formulas in the United States.

The approach developed in this paper offers several venues for future research. One would be to extend the decomposition of this paper alongside several dimensions. Recent work \citep{blanchet_real-time_2022} has started to estimate distributional national accounts --- including national wealth --- disaggregated by race. This paper's methodology could be applied to such data and provide new insights that combine our understanding of the dynamics of the racial wealth gap \citep{derenoncourt_wealth_2022} with our understanding of the dynamics of wealth inequality. To the extent that wealth can be meaningfully divided between household members \citep{fremeaux_inequalities_2020}, the same could be done for gender.

In future research, the study of capital taxation that is sketched out in this paper could also be further expanded into a full-fledged theory of optimal capital taxation, one that integrate our theories of the wealth distribution with various social objectives \citep{saez_generalized_2016} and with the broader effects of wealth taxation on the macroeconomy \citep{gaillard_wealth_2021}.

%% file: appendix/appendix-titlepage.tex
\emptythanks
\title{Appendix to ``Uncovering the Dynamics \\of the Wealth Distribution''}
\author{Thomas Blanchet}
\maketitle

\part{}
\vspace{-3em}
\localtableofcontents

\newpage

%% file: appendix/1-omitted-proofs.tex
\subsection{Proposition~\ref{prop:gyongy}}\label{sec:proof-gyongy}

The proof of proposition~\ref{prop:gyongy} is a straightforward consequence of \textcolor{QueenBlue}{\citeauthor{gyongy_mimicking_1986}'s~(\citeyear{gyongy_mimicking_1986})} theorem and of the rules of Itô calculus. Let us start by providing the full statement of \textcolor{QueenBlue}{\citeauthor{gyongy_mimicking_1986}'s~(\citeyear{gyongy_mimicking_1986})} theorem.

\begin{thm}[\cite{gyongy_mimicking_1986}]
Let $\vec{X}_t$ be a $n$-dimensional stochastic process satisfying:
\begin{equation*}
\dif \vec{X}_t = \vec{\alpha}_t \dif t + \vec{\beta}_t \dif \vec{B}_t
\end{equation*}
where $\vec{\alpha}_t$ and $\vec{\beta}_t$ are bounded and nonanticipative $n \times 1$ and $n \times m$ stochastic processes, respectively, $\vec{\beta}_t \vec{\beta}_t'$ is uniformly positive definite, and $\vec{B}_t$ is a $m$-dimensional Wiener process. Then there is a Markov process $\vec{Y}_t$ satisfying:
\begin{equation*}
\dif \vec{Y}_t = \vec{a}_t(\vec{Y}_t) \dif t + \vec{b}_t(\vec{Y}_t) \dif \vec{B}_t
\end{equation*}
where $\vec{X}_t$ and $\vec{Y}_t$ have the same marginal distributions for each $t$. We can construct $\vec{Y}_t$ by setting:
\begin{align*}
\vec{a}_t(\vec{y}) &= \mathds{E}[\vec{\alpha}_t|\vec{X}_t=\vec{y}] & \vec{b}_t(\vec{y}) &= \mathds{E}[\vec{\beta}_t\vec{\beta}_t'|\vec{X}_t=\vec{y}]^{1/2}
\end{align*}
\end{thm}
Define $\mu_{it} = (y_{it} + (r_{it} - g_t)w_{it} - c_{it})$ and $\sigma_{it} = \left(\upsilon^2_{it} + \phi^2_{it}w^2_{it} + \gamma^2_{it}\right)^{1/2}$. Recall that the wealth of person $i$ evolves according to $\dif w_{it} = \mu_{it}\dif t + \sigma_{it}\dif B_{it}$. Using \textcolor{QueenBlue}{\citeauthor{gyongy_mimicking_1986}'s~(\citeyear{gyongy_mimicking_1986})} theorem, we can write $\dif w_{it} = \mu_t(w_{it})\dif t + \sigma_t(w_{it})\dif B_{it}$ where $\mu_t(w) = \mathds{E}[\mu_{it}|w_{it}=w]$ and $\sigma_t^2(w) = \mathds{E}[\sigma^2_{it}|w_{it}=w]$. To simplify notations, consider all expectations conditional on $w_{it}=w$. Write $z_{it} = \mu_{it}\dif t + \sigma_{it}\dif B_{it}$. For the drift term, we have directly $\expc[z_{it}] = \expc[\mu_{it}]\dif t = \mu_t(w)\dif t$. For the diffusion term:
\begin{align*}
\var(z_{it}) &= \expc[(z_{it} - \mu_t(w)\dif t)^2] \\
&= \expc[(\mu_{it}\dif t - \mu_t(w)\dif t + \sigma_{it}\dif B_{it})^2] \\
& = \underbrace{\expc[(\mu_{it}\dif t - \mu_t(w)\dif t)^2]}_{\text{$=0$ because $(\dif t)^2=0$}} \quad\! + \underbrace{\expc[\sigma^2_{it}\dif B_{it}^2]}_{\substack{\text{$=\expc[\sigma^2_{it}]\dif t$} \\ \text{because $\dif B_{it}^2=\dif t$}}} + \quad\! 2\underbrace{\expc[\sigma_{it}(\mu_{it} - \mu_t(w))\dif B_{it}\dif t]}_{\text{$=0$ because $\dif B_{it}\dif t = 0$}}
\end{align*}
Therefore, $\mu_t(w)\dif t = \expc[z_{it}]$ and $\sigma_t^2(w)\dif t = \var(z_{it})$. \hfill $\qed$

Note that the proof rests on an approximation made possible by this paper's continuous-time formalism. Namely, it allows us to identify the average variance of individual shocks $\mathds{E}[\sigma^2_{it}]$ with the overall variance of wealth growth $\var(z_{it})$. This is because over an interval of time $\dif t$, the part of the variance that is explained by $w$ is of order $(\dif t)^2$ while the rest of the variance is of order $\dif t$, and so as $\dif t \rightarrow 0$, the former is negligible compared to the latter.

\subsection{Proposition~\ref{prop:wealth-tax}}\label{sec:proof-wealth-tax}

Recall that wealth evolves according to $\dif w_{it} = (\mu(w) - \tau(w))\dif t + \sigma_t(w) \dif B_{it}$ where $\tau(w)$ is the wealth tax. Let $f^*$ be the steady-state density of wealth with the wealth tax. It has to obey the Kolmogorov forward equation with the time derivative terms set to zero:
\begin{equation*}
0 = -\partial_{w}[(\mu(w) - \tau(w))f^*(w)] + \frac{1}{2}\partial^2_{w}[\sigma^2(w) f^*(w)]
\end{equation*}
Solving this differential equation, we can write:
\begin{equation*}
f^*(w) \propto \exp\left\{-2\int_0^w\frac{\sigma(s)\sigma'(s) - \mu(s)}{\sigma^2(s)}\,\dif s\right\}\exp\left\{-\int_{0}^w\frac{2\tau(s)}{\sigma^2(s)}\,\dif s\right\}
\end{equation*}
Note that the steady-state density $f$ without the wealth tax corresponds to the case $\tau(w)\equiv 0$, and therefore $f^*(w)\propto f(w)\theta(w)$ where $\theta(w)\equiv \exp\left\{-\int_{w_0}^w\frac{2\tau(s)}{\sigma^2(s)}\,\dif s\right\}$, which gives the first part of the result. For the second part, assume that $\tau(w) = \tau(w-w_0)_+$ and that $\sigma(w) = \sigma w$ for $w > w_0$. Then $\theta(w)$ simplifies to:
\begin{equation*}
\theta(w) = \exp\left\{-\frac{2\tau}{\sigma^2}\int_{w_0}^w\frac{(s - w_0)_+}{s^2}\,\dif s\right\} = \exp\left\{-\frac{2\tau}{\sigma^2}\left(\frac{w_0}{w}-1\right)\right\}\left(\frac{w}{w_0}\right)^{-2\tau/\sigma^2}
\end{equation*}
which proves the second part of the result. \hfill $\qed$

%% file: appendix/2-data.tex
\subsection{Income and Wealth}\label{sec:income-wealth}

For income and wealth data, I primarily rely on the \gls{dina} microdata from \citet{piketty_distributional_2018}, which are based on the \gls{irs} individual public-use microdata files, and which have been continuously updated to reflect methodological changes and improvements \citep{saez_rise_2020,saez_trends_2020}. These files are annual (except for 1963 and 1965) since 1962. Each observation corresponds to an adult individual (20 or older), and each variable corresponds to an item of the national accounts distributed to the entire adult population. These files distribute the entirety of the income and wealth of the United States.

This data has several advantages. It provides distributional estimates that are consistent with macroeconomic aggregates. It has rather large samples (from about 175\,000 in the 1960s to about 300\,000 today), with oversampling of the richest. And because it is based on tax data, it captures the top tail of the distribution well. But it does have some drawbacks. First, it has limited socio-demographic information: in particular, age information is only available in the form of very broad age groups. Second, the data does not include capital gains because they are not part of national income as defined by the national accounts. For these reasons, I make some adjustments and imputations to these data, using the \gls{scf} and national accounts.

I use post-tax national income as my income concept of reference. It corresponds to income after all taxes and transfers. It also distributes government expenditures and the income of the corporate sector to individuals so as to sum up to net national income.

\paragraph{Capital Gains}

\begin{figure}[ht]
\begin{center}
    \begin{subfigure}[t]{0.499\textwidth}
        \includegraphics[width=\textwidth]{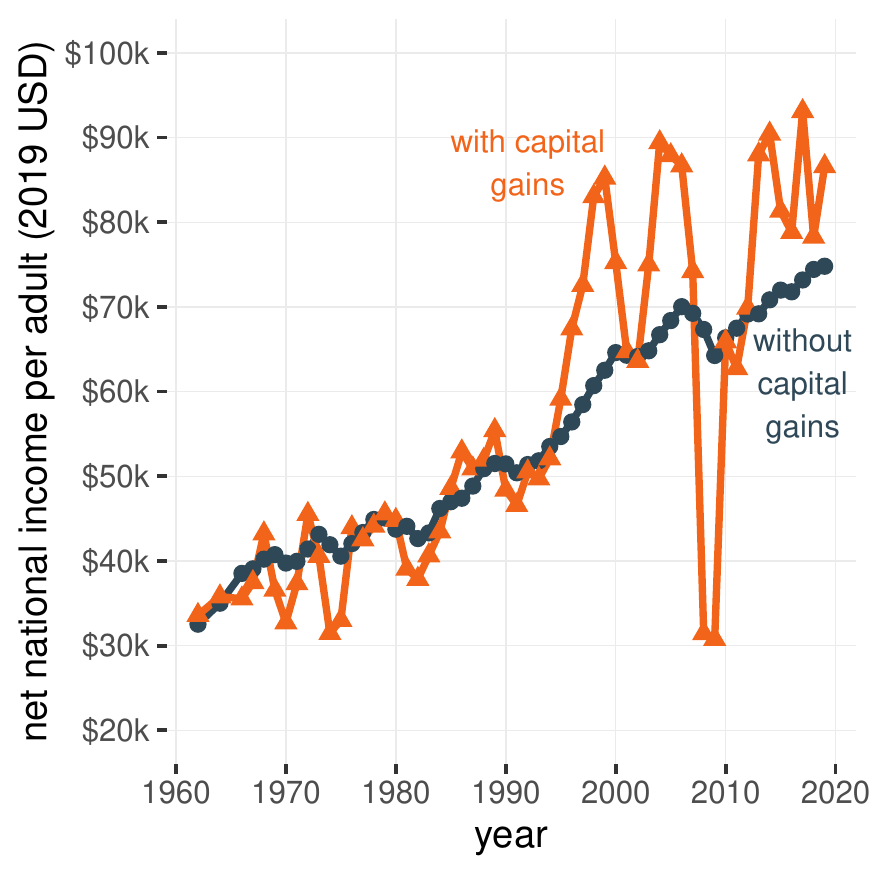}
        \caption{\centering Net National Income and Capital Gains}
        \label{fig:national-income-macro}
    \end{subfigure}%
    \begin{subfigure}[t]{0.499\textwidth}
        \includegraphics[width=\textwidth]{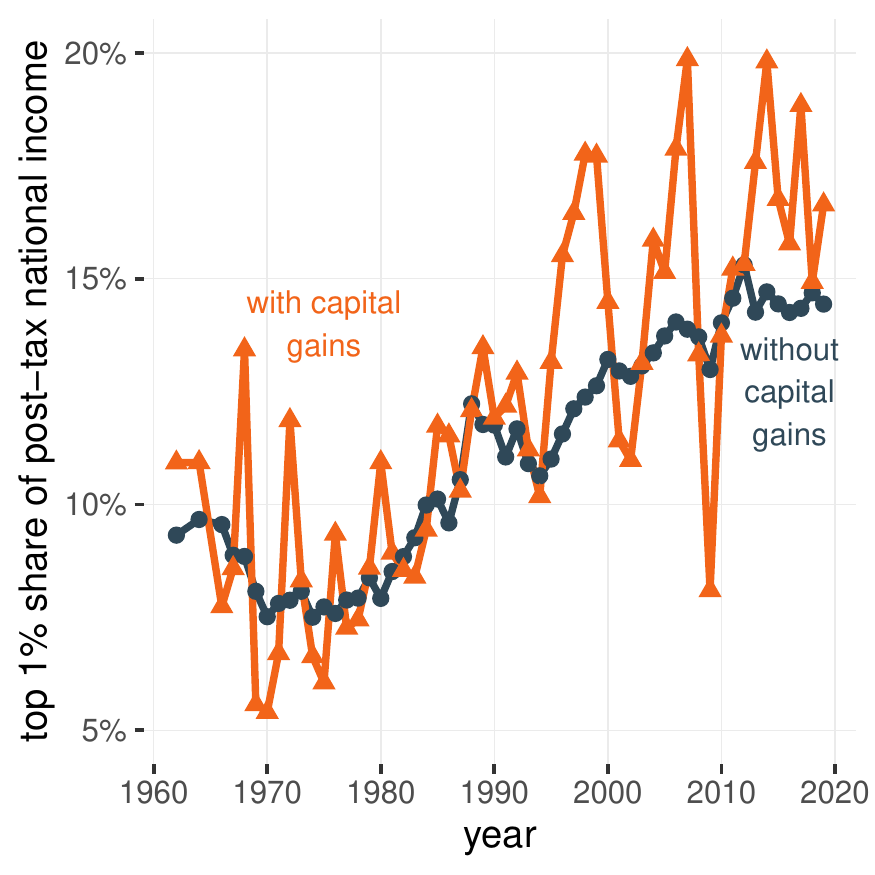}
        \caption{\centering Top 1\% Post-tax National Income Share, with and without Capital Gains}
        \label{fig:national-income-share}
    \end{subfigure}
    
    \vspace{1em}
    
    \begin{flushleft}
    \justify \footnotesize \textit{Source:} Author's computations using the \gls{dina} microdata and table TSD1 (online appendix) from \citet{piketty_distributional_2018}. \textit{Note:} The unit of analysis is the adult individual (20 or older). Income is split equally between members of couples. Capital gains are estimated assuming a constant rate of capital gains by asset type.
    \end{flushleft}
    \caption{The Impact of Capital Gains on National Income and Its Distribution}
    \label{fig:national-income}
\end{center}
\end{figure}

We can measure capital gains when they accrue to individuals or when they are realized. For our purpose, accrued capital gains are more useful than realized ones, because they are the ones that reconcile changes in the value of the balance sheet with national income and savings. On the other hand, whether a capital gain is realized now or later is the result of various tax and economic incentives that are not relevant here and do not correspond to any meaningful economic aggregate.

The \gls{dina} data only records taxable capital gains, which is essentially a measure of realized capital gains. These are a poor proxy for accrued capital gains \citep{alstadsaeter_accounting_2016}. Instead, I estimate them individually using the capitalization approach of \citet{robbins_capital_2018}. I retrieve the capital gains rate by year and asset type from the national accounts \citep[table TSD1 in appendix]{piketty_distributional_2018}. Then, I assume that for a given asset type, everyone gets the same rate of capital gains. By construction, these micro-level capital gains estimates are consistent with macro totals. Their distribution follows the logic of the \citet{saez_wealth_2016} capitalization method.\footnote{Although the income measure in the \gls{dina} data does not include capital gains, it does distribute income from the corporate sector to the owners of capital, which partly accounts for changes in asset prices. My capital gains measure is net of retained earnings, so there is no double counting.} \citet{robbins_capital_2018} provides a thorough discussion of why that measure is more appropriate to analyze the impact of asset price changes on inequality and the economy.

National income including capital gains can be quite volatile (figure~\ref{fig:national-income-macro}), but on average their inclusion matters on several fronts. \citet{robbins_capital_2018} shows that their inclusion overturns certain stylized facts about the United States economy (such as the long-run decline of saving rates) and strengthens others (such as the rising capital share and increase in income inequality). As shown in figure~\ref{fig:national-income-share}, capital gains were dampening the top 1\% share of post-tax national income during most of the 1970s, but since then, they have consistently increased it.

\paragraph{Wealth by Age}

The age information in the \gls{dina} data is very limited, so I cannot use it. Instead, I import it from the \gls{scf} and demographic estimates using constrained statistical matching. I calculate the rank in the wealth distribution in both the \gls{dina} and the \gls{scf} data, and the rank in the age distribution by sex and household type (single or couple) in the \gls{scf} data. Then, I match the \gls{dina} observations one by one to \gls{scf} observations based on their wealth rank to give them a rank in the age distribution.\footnote{Note that both datasets are weighted, so that observations end up being duplicated and partially matched to one another. When the samples contain $M$ and $N$ observations, respectively, the resulting dataset contains at most $M+N-1$ observations.} Finally, I use the population structure from the demographic data to attribute an age to every \gls{dina} observation. By construction, the method preserves the wealth distribution in \gls{dina}, the population by age and sex from demographic sources, and the copula between wealth and age from the \gls{scf}.

\subsection{Demography}\label{sec:demography}

I compute the entire demography of the United States from 1850 to 2100. Although the income and wealth data does not start until 1962, the model requires demographic data that starts much earlier. Indeed, I need to simulate how wealth gets transmitted from one generation to the next. Therefore, if a supercentenarian dies in the 1960s, I have to be able to simulate their entire life history to know how many live children they have and how old they are. For all years and all ages, I estimate data on the population structure by age and sex, mortality (i.e., life tables), fertility (for both sexes), and intergenerational ties (age and sex of children). Sometimes, data is only available by age group (e.g., of five years) or a subset of years (e.g., every ten years). Whenever necessary, I interpolate estimates with a monotonic cubic spline \citep{fritsch_monotone_1980} to get data for every single year and age.

\paragraph{Population by Age and Sex}

Before 1900, I directly estimate the population pyramid using the decennial census microdata from the IPUMS~USA database \citep{ruggles_steven_ipums_2022}. From 1900 to 1932, I use the National Intercensal Tables from the \citeauthor{united_states_census_bureau_national_2021}. From 1933 to 2016, I use population estimates from the Human Mortality Database\nocite{HumanMortalityDatabase}.\footnote{See \url{https://www.mortality.org/hmd/USA/DOCS/ref.pdf} for detailed primary sources.} After 2016, I use the projections from the World Population Prospects \citep{united_nations_world_2019}.

\paragraph{Life Tables}

Before 1900, I use the historical life tables from \citet{haines_estimated_1998}. From 1900 to 1932, I use the Human Life Table Database\nocite{HumanLifeTable}, and from 1933 to 2016, life tables from the Human Mortality Database\nocite{HumanMortalityDatabase}. After 2016, I rely on projections from the World Population Prospects \citep{united_nations_world_2019}. All tables are broken down by sex.

\paragraph{Age-Specific Fertility Rates by Birth Order}

I estimate age-specific fertility rates by birth order for both sexes. For women, they are directly available from 1933 to 2016 from the Human Fertility Database\nocite{HumanFertility}. From 1917 to 1932, I use data from the Human Fertility Collection\nocite{HumanFertilityCollection}. That same source provides fertility rates until going back to 1895--1899, but without the breakdown by birth order. Therefore, before 1917, I assume that the birth order composition remains constant. Before 1895, there is no age-specific data available, so I use the data on the total fertility rate and rescale the age profile from 1895--1899 to that value.\footnote{See Gapminder\nocite{Gapminder}: \url{https://www.gapminder.org/news/children-per-women-since-1800-in-gapminder-world/}}

Unlike female fertility rates, male fertility rates are not a standard demographic indicator, so they are not directly available from any source. To estimate them, I combine the age-specific female fertility rates with the joint distribution of the age of opposite-sex couples since 1850 calculated using the decennial census microdata from the IPUMS~USA database \citep{ruggles_steven_ipums_2022}.

\paragraph{Age and Sex of Children}

I simulate the distribution of the number, age, and sex of living children for each year after 1962 (when income and wealth data starts). I do this for every age and both sexes, which allows me to realistically model how wealth gets transmitted from one generation to the next. To that end, I combine all the data above. I make every person have children randomly over their past lifetime according to year, age, and sex-specific fertility rate. Because I have the breakdown by birth order, I can take into account how the decision to have another child depends on the number of children that one already has. Then, I make each child go through life and die randomly according to their year, age, and sex-specific mortality rate. As a result, I can tie every individual in the database to fictitious descendants that are, on average, representative of the true composition of descendants.

\subsection{Inheritance and Estate Tax}\label{sec:inheritance}

Part of the inheritance process is determined by the demographics and the distribution of wealth, while other parts have to be modeled separately. I assume that people die at random, conditional on their age and sex, so the distribution of inheritances corresponds to the distribution of wealth, weighted by mortality rates. I then assume that the wealth of decedents is either redistributed to their spouse (if any) or to their descendants (if they have no living spouse), after payment of the estate tax. The demography gives the age and sex of decedents (see section~\ref{sec:demography}). I assume that inheritance is split equally between children, as is the norm in the United States \citep{menchik_primogeniture_1980}.

While the demographic aspect of inheritance is endogenously determined by demography, I still need to model separately how wealth gets distributed for a given age and sex. This facet of the problem captures intergenerational wealth mobility, in the sense that wealthier people might also have wealthier parents and thus inherit more. There are two aspects to this question: the extensive margin (how likely are you to receive an inheritance in a given year?) and the intensive margin (how much inheritance do you receive?) To address this question, I use data from the \gls{scf}, which has been recording inheritance consistently since 1989. Because the probability of receiving an inheritance in a given year is very low overall (about half a percent, see figure~\ref{fig:inheritance-extensive-age}), I have to pool all the 1989--2016 waves to get sufficient sample sizes.

\begin{figure}[p]
\begin{center}
    \begin{subfigure}[t]{0.49\textwidth}
        \includegraphics[width=\textwidth]{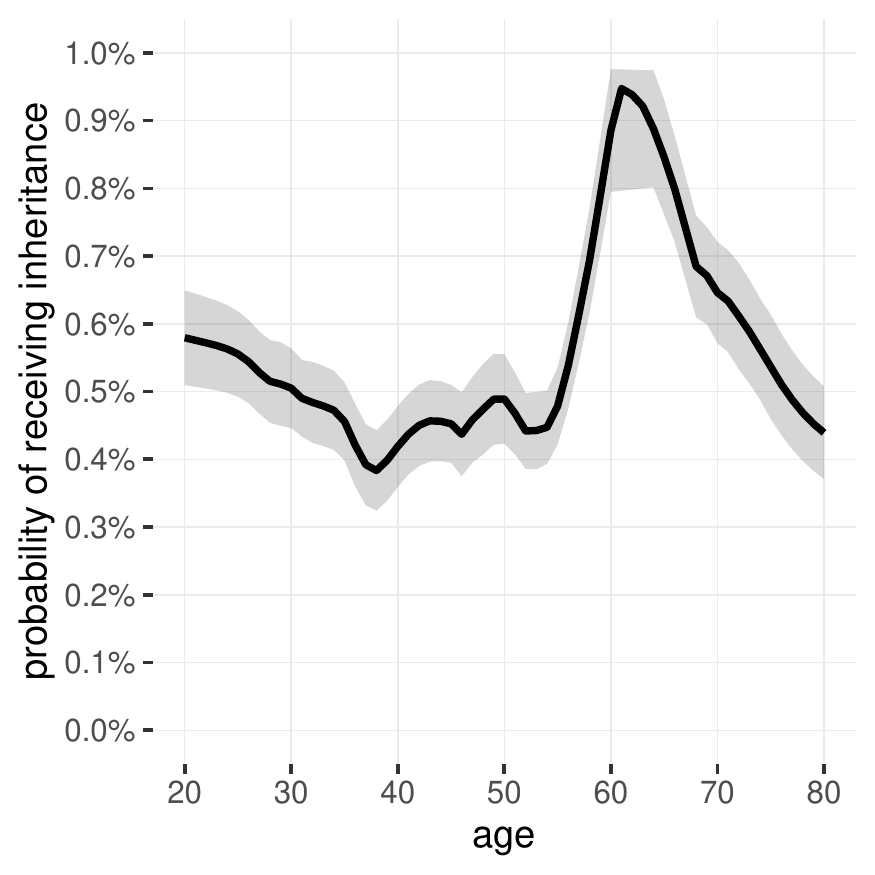}
        \caption{\centering Probability of Receiving Inheritance by Age}
        \label{fig:inheritance-extensive-age}
    \end{subfigure}
    \begin{subfigure}[t]{0.49\textwidth}
        \centering
        \includegraphics[width=\textwidth]{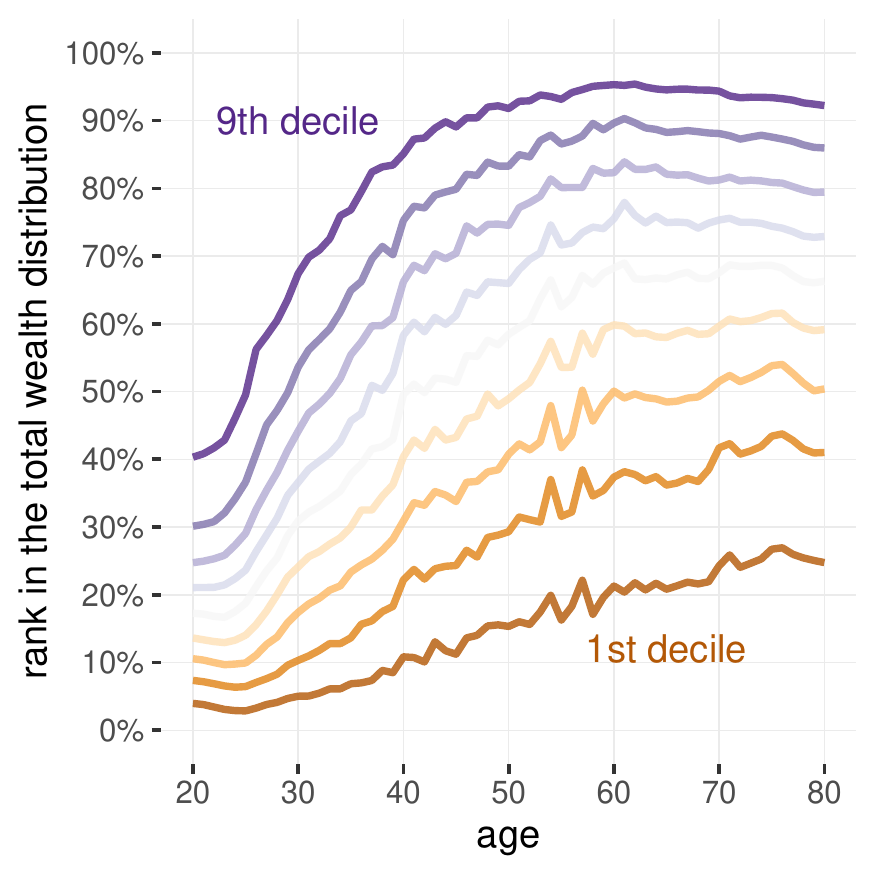}
        \caption{Rank in the Wealth Distribution by Age}
        \label{fig:inheritance-extensive-wealth-age}
    \end{subfigure}
    
    \vspace{1em}
    
    \begin{subfigure}[t]{0.49\textwidth}
        \includegraphics[width=\textwidth]{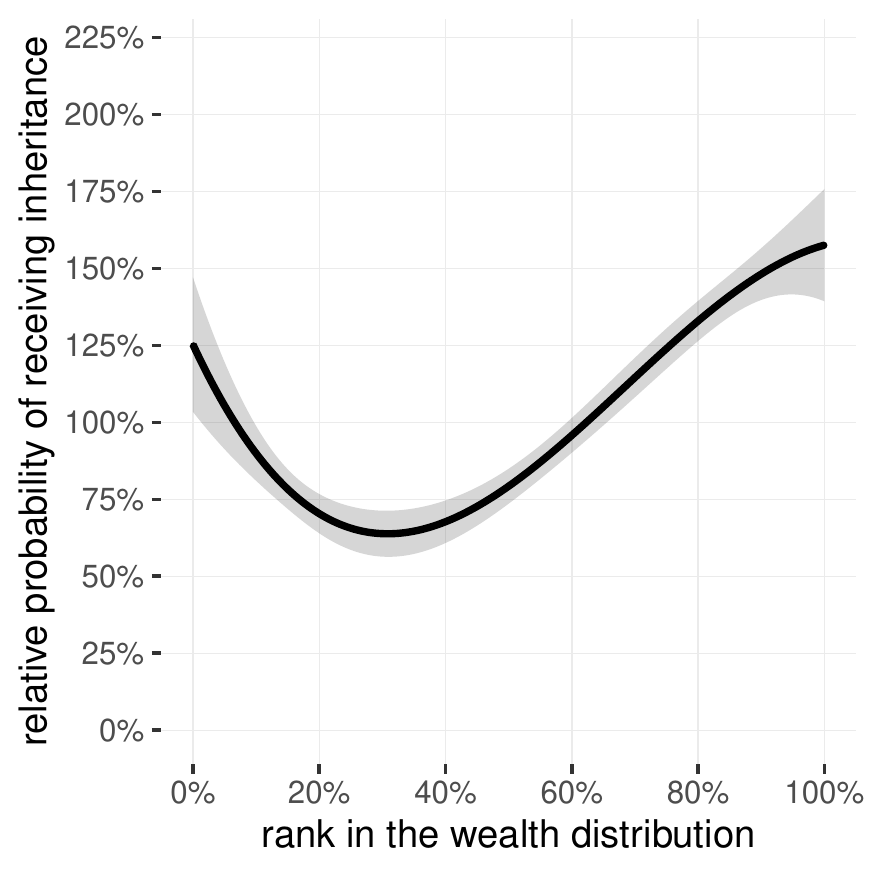}
        \caption{\centering Relative Probability of Inheritance \newline (Conditional on Age)}
        \label{fig:phi-inheritance-wealth}
    \end{subfigure}
    \begin{subfigure}[t]{0.49\textwidth}
        \includegraphics[width=\textwidth]{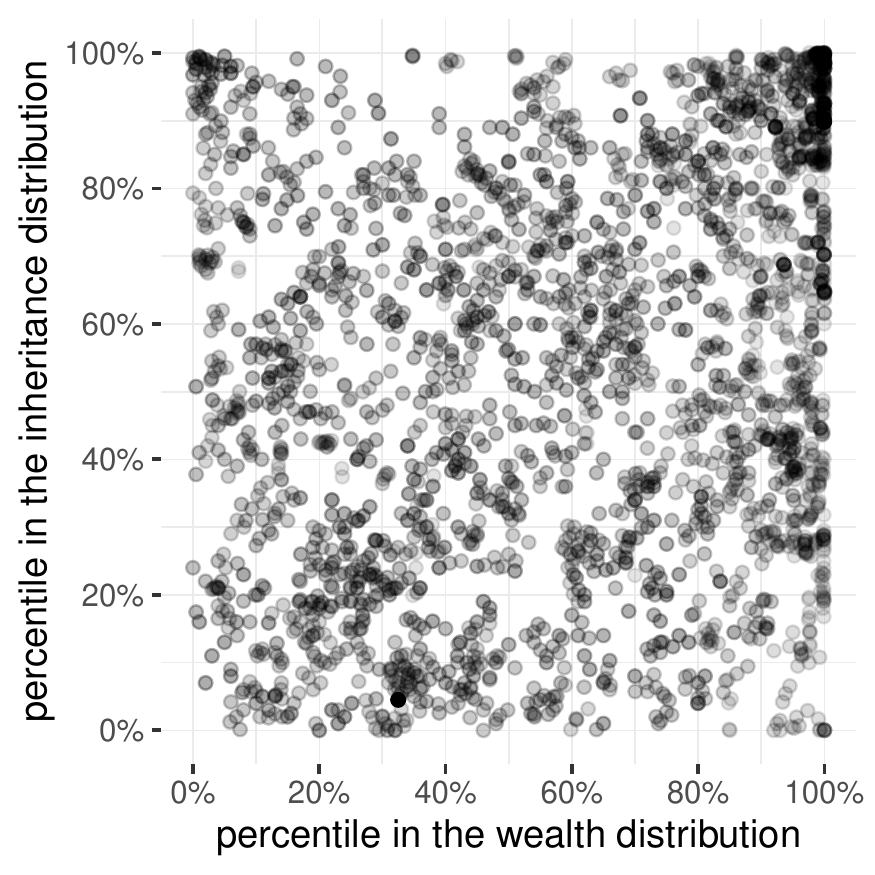}
        \caption{\centering Ranks in the Wealth and in the Inheritance Distribution (Conditional on Age)}
        \label{fig:inheritance-wealth-conditional}
    \end{subfigure}
    
    \vspace{1em}
    
    \begin{flushleft}
    \justify \footnotesize \textit{Source:} Own computations using the \glsfirst{scf} (1989--2019). Gray ribbons correspond to the 95\% confidence intervals. In figure~\ref{fig:inheritance-wealth-conditional}, opacity is proportional to the weight of observations.
    \end{flushleft}
    \caption{Modeling of Inheritance}
    \label{fig:model-inheritance}
\end{center}
\end{figure}

\paragraph{Extensive Margin}

Let $D_i=1$ if individual $i$ receives inheritance, and $D_i=0$ otherwise. Let $A_i$ be their age and $W_i$ their wealth. Assume that:
\begin{equation}\label{eqn:extensive-margin}
\prob\{D_i = 1|A_i = a, W_i = w\} = \prob\{D_i = 1|A_i = a\}\phi(F_{A_i = a}(w))
\end{equation}
where $F_{A_i = a}$ is the \gls{cdf} of wealth conditional on age, and $\int_0^1 \phi(r)\,\dif r = 1$. By construction, the expected value of the right-hand side of~(\ref{eqn:extensive-margin}) conditional on age is equal to $\prob\{D_i = 1|A_i = a\}$ so that the specification makes probabilistic sense.\footnote{It is the direct result of a change of variable $r = F_{A_i = a}(w)$ and using the fact that $\pd{}{w}F_{A_i = a}(w)=f_{A_i = a}(w)$, so that $\int_{-\infty}^{+\infty} \phi(F_{A_i = a}(w)) f_{A_i = a}(w)\,\dif w = \int_0^1 \phi(r)\,\dif r = 1$.} Note that $F_{A_i = a}(w)$ is the rank of $w$ in the wealth distribution (conditional on age), which is how we can make the formula~(\ref{eqn:extensive-margin}) consistent regardless of the shape of the wealth distribution.

The value of $\prob\{D_i = 1|A_i = a\}$ is determined by demography, so we only need to estimate $\phi$. I start by calculating a rank in the wealth distribution conditional on age by running a nonparametric quantile regression of wealth on age for every percentile (see figure~\ref{fig:inheritance-extensive-wealth-age}).\nocite{quantreg} I then regress the dummy $D_i$ for having received inheritance on that rank, multiplied by $\prob\{D_i = 1|A_i = a\}$. I use \gls{ols} and a cubic polynomial with coefficients constrained so that its integral over $[0,1]$ equals one (see figure~\ref{fig:phi-inheritance-wealth}). As we can see, even after partialling out the effect of age, wealthier people still experience a higher probability of receiving an inheritance. I use that polynomial as my estimate of $\phi$.

\paragraph{Intensive Margin}

I account for the intensive margin by modeling the joint distribution of the ranks in the wealth distribution and the inheritance distribution (i.e., the copula), conditional on age and on having received an inheritance. I take the subsample of inheritance receivers and calculate their rank in the wealth and inheritance distribution using a nonparametric quantile regression as I did for the extensive margin.

The dependence between the two ranks is weak but significant (see figure~\ref{fig:inheritance-wealth-conditional}): their Kendall's $\tau$ is equal to $8.3\%$. I represent this dependency parametrically using a bivariate copula. I select the most appropriate model out of a large family of 15 single-parameter copulas by finding the best fit according to the \gls{aic}, which is the Joe copula.\footnote{The list of copulas includes the Gaussian copula, Student's $t$ copula, the Clayton copula, the Gumbel copula, the Frank copula, the Joe copula, and rotated versions of these copulas.}\footnote{The Joe copula has the parametric form $C_{\theta}(u,v)={1-\left[(1-u)^\theta + (1-v)^\theta - (1-u)^\theta(1-v)^\theta \right]^{1/\theta}}$.} I estimate its parameter so as to match the empirical value for Kendall's $\tau$.\nocite{VineCopula}

\paragraph{Estate Tax}

\begin{figure}[ht]
\begin{center}
    \begin{subfigure}[t]{0.499\textwidth}
        \includegraphics[width=\textwidth]{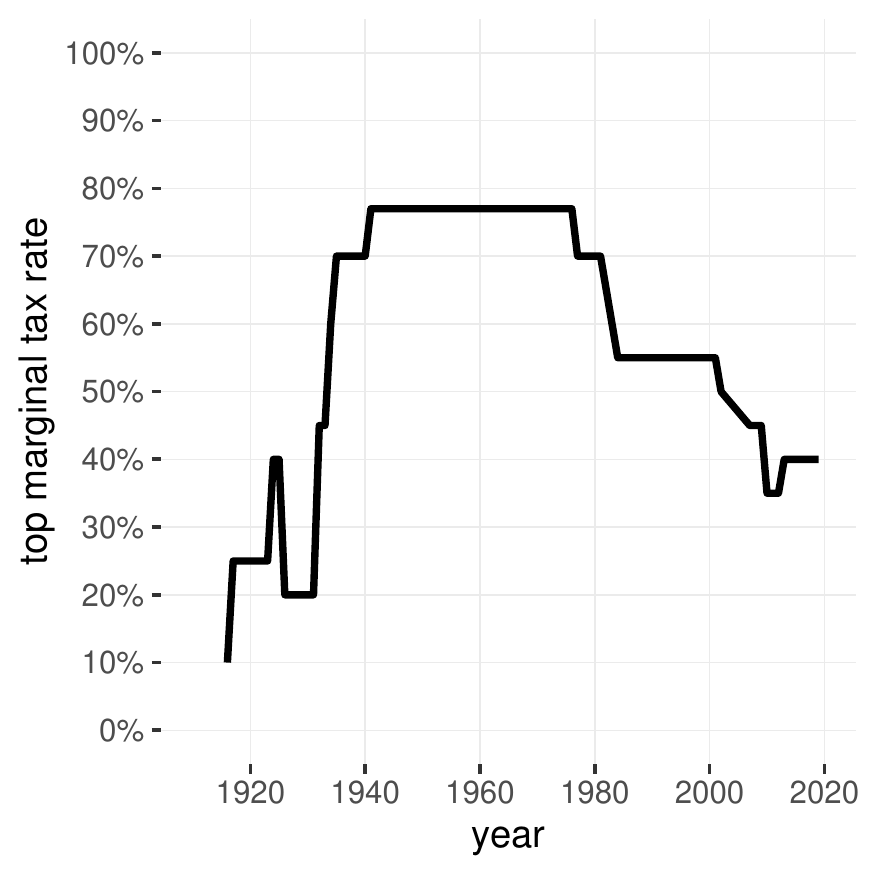}
        \caption{Top Marginal Estate Tax since 1916}
        \label{fig:marginal-estate-tax-rate}
    \end{subfigure}%
    \begin{subfigure}[t]{0.499\textwidth}
        \includegraphics[width=\textwidth]{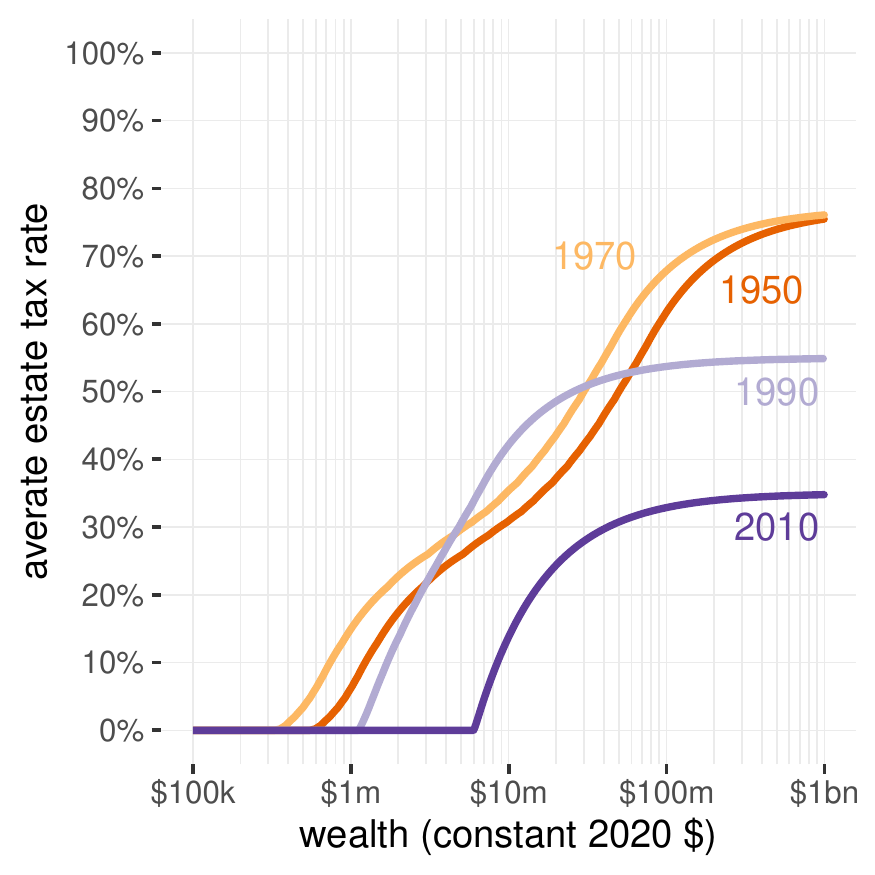}
        \caption{Average Estate Tax Rate, by Wealth}
        \label{fig:estate-tax-schedule}
    \end{subfigure}
    
    \vspace{0.5em}
    
    {\footnotesize \textit{Source:} Own computation using the statutory schedules of the federal estate tax.}
    \caption{Estate Tax}
    \label{fig:estate-tax}
\end{center}
\end{figure}

I account for the federal estate tax using the complete estate tax schedule and exemption amount for each year. The top marginal estate tax rate has followed a clear inverted U-shaped pattern over the 20th century (figure~\ref{fig:marginal-estate-tax-rate}), having been reduced by half since its mid-century peak. However, the changes to the overall progressivity of the estate tax are more ambiguous (figure~\ref{fig:estate-tax-schedule}). While the top marginal tax rate was very high in the 1950s, the top bracket did not kick in until extremely high levels of wealth. The 1980s reforms significantly reduced the top tax rate and increased the exemption amount so that by 1990, the very top and the upper middle of the wealth distribution were facing lower average tax rates. But \$10m estates were actually facing slightly higher average tax rates. By now, however, the estate has been lowered so much that its profile is unambiguously less progressive than in the 1950s.

\subsection{Marriages, Divorces and Assortative Mating}

\paragraph{Aggregate Marriages and Divorce Rates by Age}

I calculate population pyramids disaggregated by marital status using census microdata from the IPUMS~USA database \citep{ruggles_steven_ipums_2022}. I use this data to estimate the rate at which people get married and divorced by age and sex. This is possible assuming that mortality rates are the same regardless of marital status, and assuming that marriage rates are the same for people that are divorced are people that have never been married. Under these conditions, we have for each sex:
\begin{align*}
\text{(marriage rate)}_{a,t} &\textstyle = 1 - \frac{\text{(fraction never married)}_{a+1,t+1}}{\text{(fraction never married)}_{a,t}} \\[0.5em]
\text{(divorce rate)}_{a,t} &\textstyle= \frac{\text{(fraction divorced)}_{a+1,t+1} - [1 - \text{(marriage rate)}_{a,t}] \times \text{(fraction divorced)}_{a,t}}{\text{(fraction married)}_{a,t}}
\end{align*}
where $a$ is age, and $t$ is the year. This estimate being quite noisy, I winsorize the bottom 10\% and the top 10\% of values, and then apply a moving average with a ten-year age and year window. This gives the results of Figure~\ref{fig:marriage-rate}. For the simulation, we anchor these numbers to the aggregate data on the crude rate of marriage and divorce from the \gls{nvss} \citep{NVSS} (Figure~\ref{fig:marriage-divorce-rates}).

\paragraph{Assortative Mating}

\begin{figure}[p]
\begin{center}
    \begin{subfigure}[t]{\textwidth}
        \includegraphics[width=\textwidth]{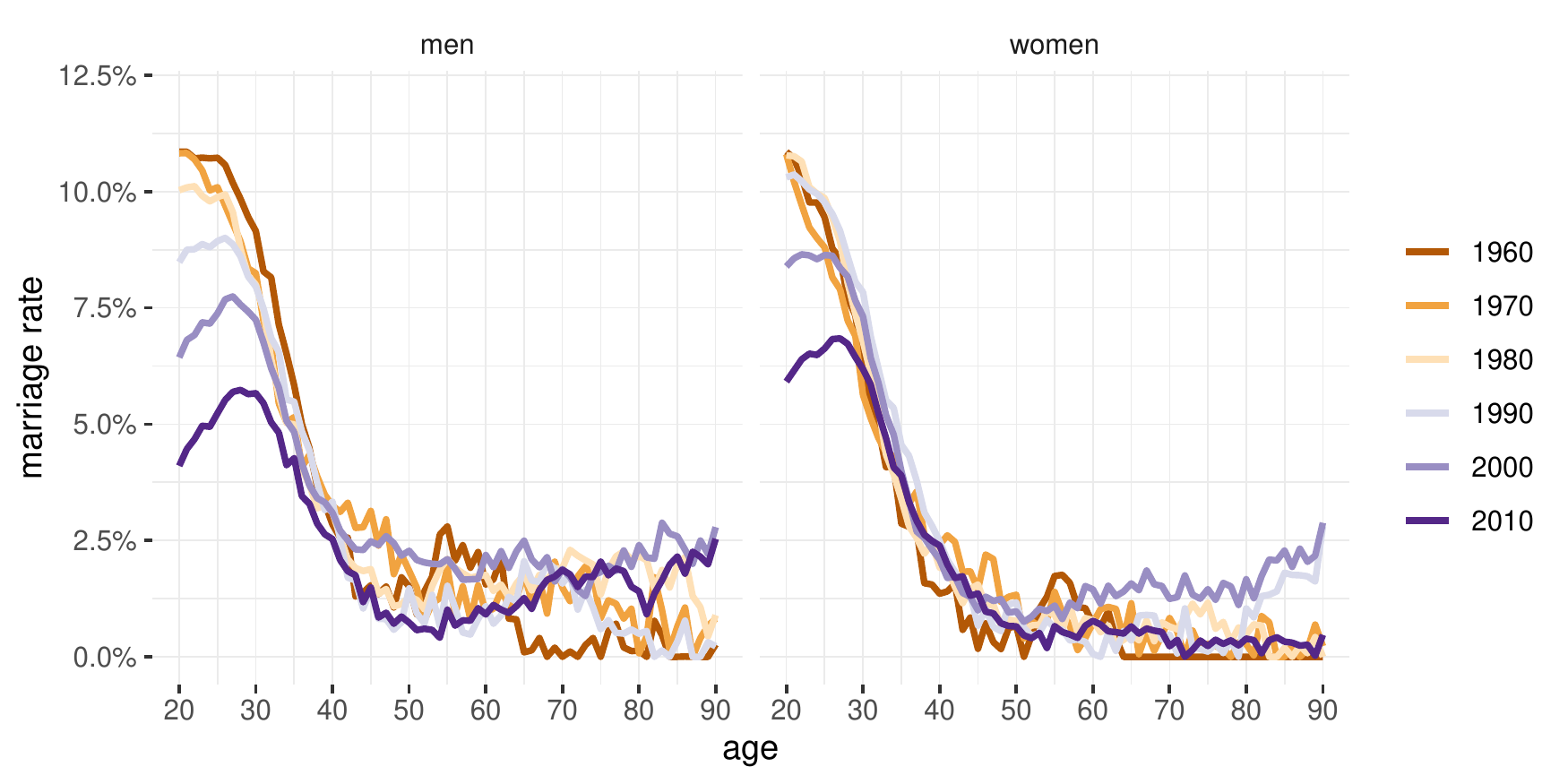}
        \caption{\centering Rates of Marriage by Age and Sex}
        \label{fig:marriage-rate}
    \end{subfigure}
    \vspace{1em}
    
    \begin{subfigure}[t]{0.49\textwidth}
        \includegraphics[width=\textwidth]{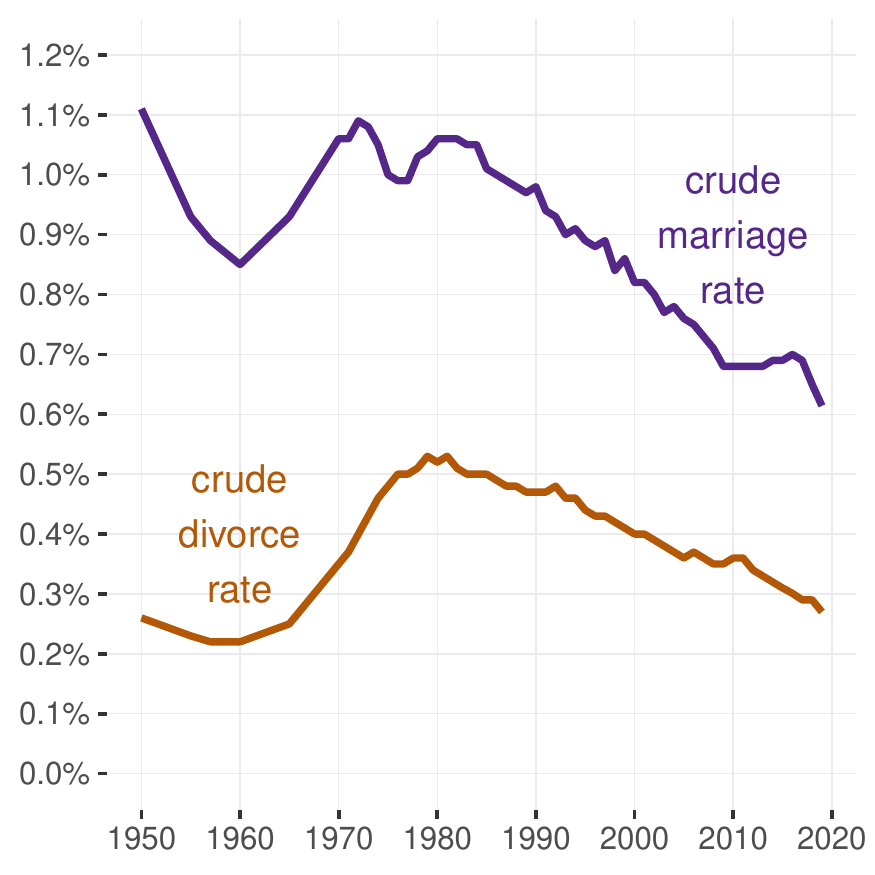}
        \caption{\centering Crude Aggregate Rates of Marriage and Divorce}
        \label{fig:marriage-divorce-rates}
    \end{subfigure}
    \begin{subfigure}[t]{0.49\textwidth}
        \includegraphics[width=\textwidth]{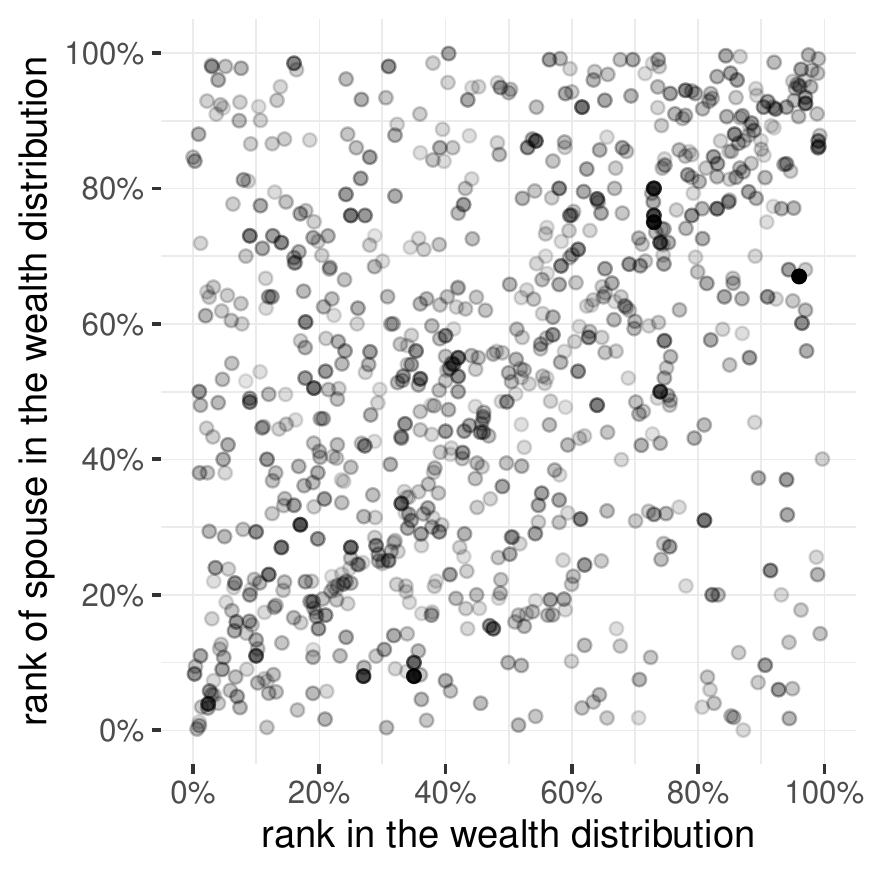}
        \caption{\centering Ranks in the Wealth Distribution of Both Spouses at Marriage}
        \label{fig:spouse-wealth-copula}
    \end{subfigure}
    \vspace{1em}
    
    \begin{flushleft}
    \justify \footnotesize \textit{Source:} Own computations using the census microdata from IPUMS \citep{ruggles_steven_ipums_2022}, \glsfirst{nvss} \citep{NVSS}, and \glsfirst{sipp} (2013--2016). In figure~\ref{fig:spouse-wealth-copula}, opacity is proportional to the weight of observations.
    \end{flushleft}
    \caption{Modeling of Marriage and Divorce}
    \label{fig:model-marriage-divorce}
\end{center}
\end{figure}

I calculate the extent of assortative mating by estimating the copula between the rank of each spouse in the wealth distribution, conditional on age, at the time they get married. To that end, I use data from the \gls{sipp}, which collected data on wealth between 2013 and 2016. I select all people who just got married (i.e., people single in year $t-1$ but married in year $t$) and fit a parametric copula to the joint rank of each spouse in the wealth distribution, conditional on age (Figure~\ref{fig:spouse-wealth-copula}). As in Section~\ref{sec:inheritance}, I select the best copula out of a large family of 15 single-parameter copulas according to the \gls{aic}, which is the Frank copula.\footnote{The formula of the Frank copula is $C(u,v)=-{\frac {1}{\theta }}\log \!\left[1+{\frac {(\exp(-\theta u)-1)(\exp(-\theta v)-1)}{\exp(-\theta )-1}}\right]$.} I estimate its parameter so as to match the empirical Kendall's $\tau$, which is equal to $28\%$.

I also estimate the effect of assortative mating at divorce by looking at how wealth is divided between both spouses right before they separate. I also use the \gls{sipp}, select people who are married in year $t$ but separated in year $t+1$, and extract the distribution of the share of the household's wealth that each spouse has.

%% file: appendix/3-estimations.tex
Overall, the estimation procedure in this paper only involves fitting a series of lines. In practice, there are several methods for fitting lines, and this paper uses the approach that makes the most sense given the nature of the problem. This section describes the procedure used and discusses the sensitivity of estimates to the various hyperparameters.

\subsection{Estimation Procedure}\label{sec:deming-estimation-procedure}

Recall that the estimation procedure in this paper requires fitting a linear relationship of the form given by equation~(\ref{eq:kf-full-int}), which I will simplify as follows to emphasize the line-fitting aspect:
\begin{equation}\label{eq:fit-eq}
\forall t \in \{1, \dots T\} \qquad y_{it}=a + bx_{it}
\end{equation}
for each wealth bin $i \in \{1, \dots, n\}$. The data are $y_{it}$ and $x_{it}$. The parameters are $a_i$ and $b_i$. While fitting this equation using \gls{ols} is an option, it is not a natural choice here for two reasons. First, that this equation does not hold perfectly in practice results from error terms that can be attributed simultaneously to $y_{it}$ and $x_{it}$. In that sense, the setting is closer to an error-in-variable model, for which \gls{ols} are known to be biased. Second, and relatedly, there is an invariance in equation~(\ref{eq:fit-eq}) that \gls{ols} breaks. Unlike usual econometric regressions, there is no obvious distinction between the outcome variable and the explanatory variable. The equation would make as much sense if the roles of $x_{it}$ and $y_{it}$ were reversed (and the meaning of the parameters $a_i$ and $b_i$ changed accordingly). But with \gls{ols}, regressing $y$ on $x$ yields a different result than regressing $x$ on $y$.

The solution is to use a \citet{deming_statistical_1943} regression, a generalization of orthogonal regression, which is a standard way of estimating error-in-variable models. The \citet{deming_statistical_1943} estimator minimizes:
\begin{equation*}
\min_{a,b,x^*_1,\dots,x^*_T} \sum_{t=1}^T (y_t - a - bx^*_t)^2 + \delta(x_t-x^*_t)^2
\end{equation*}
where $\delta$ is an hyperparameter, defined separately, that captures the relative variance of the measurement error on $y_{it}$ and $x_{it}$. When $\delta=0$ the estimator collapses to \gls{ols}, and when $\delta=+\infty$ it collapses to \gls{ols} as well, but with the left-hand side and the right-hand side reversed. This estimator has an analytical solution, which is:
\begin{align*}
b_i &= \frac{s_{yy}-\delta s_{xx} + \sqrt{(s_{yy}-\delta s_{xx})^2 + 4\delta s_{xy}^2}}{2s_{xy}} & a_i&=\overline{y} - b\overline{x}
\end{align*}
where:
\begin{align*}
\overline{x} &= \tfrac{1}{T}\sum x_{it} &
\overline{y} &= \tfrac{1}{T}\sum y_{it} \\
s_{xx} &= \tfrac{1}{T}\sum_t (x_{it}-\overline{x})^2 &
s_{xy} &= \tfrac{1}{T}\sum_t (x_{it}-\overline{x})(y_{it}-\overline{y}) &
s_{yy} &= \tfrac{1}{T}\sum_t (y_{it}-\overline{y})^2
\end{align*}
I estimate a variance of the measurement error on $x_{it}$ and $y_{it}$, by taking a five-year moving average of these values (in each wealth bin), and calculate the variance of the difference between the values and their moving average. This allows me to estimate $\delta$. This is a rough estimate, so I check that results are robust to it in Section~\ref{sec:robustness-checks}. The choice of $\delta$ turns out to have limited impact, because the line actually fits the data closely: error terms all smalls, and therefore assuming that they are due to $x_{it}$ or $y_{it}$ makes little difference. 

\subsection{Estimation of Standard Errors}\label{sec:std-err}

This paper's setting requires a nonstandard procedure for estimating standard errors for two reasons. First, because I use a \citet{deming_statistical_1943} regression, which imply different standard errors than \gls{ols}. Second, because the error terms are not independent: there are correlated both across time and across wealth bins. To deal with these issues, I use a parametric bootstrap procedure that efficiently simulates error terms that reproduce the key features of the data. 

First, I estimate an error term for each observation. Note that in geometrical terms, the \citet{deming_statistical_1943} performs an oblique projection of observations onto the fitted line, according to a direction set by the hyperparameter $\delta$. This projection is associated with the following scalar product in $\mathds{R}^2$:
\begin{equation*}
\langle (x_1, y_1), (x_2, y_2) \rangle_\delta \mapsto \delta x_1 x_2 + y_1y_2
\end{equation*}
Given this behavior, the error terms that the \citet{deming_statistical_1943} regression estimates for $x_{it}$ and $y_{it}$ are perfectly correlated, so it would not make sense to model them separately. Instead, I define a single error term $e_{it}$ for every observation as the scalar product between (i) the unit vector orthogonal to the fitted line, and (ii) the residual vector $(x_{it} - x^*_{it}, y_{it} - y^*_{it})$. The vector that is normal to the fitted line is $(b, -\delta)/\sqrt{\delta(b^2 + \delta)}$, where $b$ is the slope of the fitted line, hence:
\begin{equation*}
e_{it} = \frac{\langle (b, -\delta), (x_{it} - x^*_{it}, y_{it} - y^*_{it}) \rangle_\delta}{\sqrt{\delta(b^2 + \delta)}}  = \sqrt{\frac{\delta}{b^2 + \delta}}[b (x_{it} - x^*_{it}) - (y_{it} - y^*_{it})]
\end{equation*}
Then, for each wealth bin $i$, I estimate a coefficient $\rho_{i}$ of autocorrelation across time as:
\begin{equation*}
\rho_i = \left(\tfrac{1}{T}\sum_t e_{it}^2\right)^{-1}\left(\tfrac{1}{T-1}\sum_t e_{it}e_{i,t-1}\right)
\end{equation*}
And I also estimate a correlation coefficient $r$ of autocorrelation across wealth bins:
\begin{equation*}
r = \tfrac{1}{T}\sum_t \left(\tfrac{1}{n}\sum_i e^2_{it}\right)^{-1} \left(\tfrac{1}{n-1}\sum_i e_{it}e_{i-1,t}\right)
\end{equation*}
Unlike $\rho_i$, I do not vary the coefficient $r$ across time. Indeed, there is no empirical evidence, and also no good \textit{a priori} reasons, for it to vary as such. Instead, I average the value across all periods, which gives a more robust estimate. I assume that autocorrelations across time and wealth bins follow each follow an $AR(1)$ process. For each wealth bin $i$, the correlation matrix across time is $\Omega^w_i = \mathrm{Toeplitz}[1, \rho_i, \rho_i^2, \dots, \rho_i^T]$. The correlation matrix across wealth bins is the same for all periods and is equal to $\Omega^t = \mathrm{Toeplitz}[1, r, r^2, \dots, r^n]$. I assume that the correlation matrix $\Omega$ between all error terms, across all periods and all wealth bins, has a two-way $AR(1)$ structure, which can be constructed as:
\begin{equation*}
\Omega = \begin{bmatrix}
\Omega^w_1 &          &        &          \\
         & \Omega^w_2 &        &          \\
         &          & \ddots &          \\
         &          &        & \Omega^w_n
\end{bmatrix}^{1/2}
\underbrace{\begin{bmatrix}
\Omega^t &          &        &          \\
         & \Omega^t &        &          \\
         &          & \ddots &          \\
         &          &        & \Omega^t
\end{bmatrix}}_{\text{$T$ times}}
\begin{bmatrix}
\Omega^w_1 &          &        &          \\
         & \Omega^w_2 &        &          \\
         &          & \ddots &          \\
         &          &        & \Omega^w_n
\end{bmatrix}^{1/2}
\end{equation*}
where $(\cdot) \mapsto (\cdot)^{1/2}$ is the matrix square root for symmetric positive definite matrices. I estimate a standard error $\sigma_i$ in every wealth bin $i$ as:
\begin{equation*}
\sigma_i = \sqrt{\tfrac{1}{T} \sum_t e_{it}^2}
\end{equation*}
and define the vector $\sigma = (\sigma_1, \dots, \sigma_n)$. In line with the evidence, I assume that the error term is heteroscedastic over wealth bins but homoscedastic over time. The final, complete covariance matrix $\Sigma$ for error terms is, therefore, equal to $\Sigma = A'\Omega A$, where
\begin{equation*}
A = \textrm{Diagonal}(\underbrace{\sigma, \dots, \sigma}_{\text{$T$ times}})
\end{equation*}
I simulate the error terms $\tilde{e}_{it}$ according to a multivariate normal distribution with mean zero and covariance $\Sigma$. The 2-dimensional residual is the product of this value with the unit vector normal to the fitted line, i.e. $(b, -\delta)/\sqrt{\delta(b^2 + \delta)}$. So I obtain bootstrap replications of the sample by adding the simulated error terms to the fitted values as follows:
\begin{align*}
\tilde{x}_{it} &= x^*_{it} + \tilde{e}_{it} b/\sqrt{\delta(b^2 + \delta)} \\
\tilde{y}_{it} &= y^*_{it} - \tilde{e}_{it} \delta/\sqrt{\delta(b^2 + \delta)}
\end{align*}
I get bootstrap replications of the parameter values by running the \citet{deming_statistical_1943} again regression on the simulated data.

\subsection{Transformation of Parameters Due to the Inverse Hyperbolic Sine Transform}\label{sec:transform-param-asinh}

To facilitate the manipulation of the data and the estimations, I transform wealth using the inverse hyperbolic sine function, which we can define as:
\begin{equation*}
\ash(x) \mapsto \log(x + \sqrt{x^2 + 1})
\end{equation*}
This function is useful because it is bijective, behaves logarithmically at large scales, but still tolerates zero or negative values, for which it behaves linearly. A helpful feature of continuous-time stochastic processes is that the dynamics of wealth and the dynamics of $\ash(\text{wealth})$ are directly related to one another through Itô's lemma. Concretely, assume that wealth follows a \gls{sde} with the following drift and diffusion:
\begin{align*}
\mu_t(w) &= z_t(w) - c_t(w) & \sigma_t(w) &= \left(\psi^2_t(w) + \gamma^2_t(w)\right)^{1/2}
\end{align*}
where:
\begin{align*}
z_t(w) &\equiv y_t(w) + (r_t(w) - g_t)w  & \psi^2_t(w) &\equiv \upsilon^2_t(w) + \phi^2_t(w)w^2
\end{align*}
are, respectively, the drift and the diffusion induced by income. Then $\ash(\text{wealth})$ follows a \gls{sde} as well with the following drift and diffusion:
\begin{align*}
\tilde{\mu}_t(w) &= \frac{\mu_t(w)}{\sqrt{1+w^2}} - \frac{w}{2\sqrt{1+w^2}} \left(\frac{\sigma_t(w)}{\sqrt{1+w^2}}\right)^2 \\
\tilde{\sigma}_t(w) &= \frac{\sigma_t(w)}{\sqrt{1+w^2}}
\end{align*}
To perform the estimation, I must first transform the income parameters to conform to the $\ash$ scale:
\begin{align*}
\tilde{z}_t(w) &= \frac{z_t(w)}{\sqrt{1+w^2}} - \frac{w}{2\sqrt{1+w^2}}\left(\frac{\psi_t(w)}{\sqrt{1+w^2}}\right)^2 &
\tilde{\psi}_t(w) &= \frac{\psi_t(w)}{\sqrt{1+w^2}}
\end{align*}
Then I estimate $\tilde{c}(w)$ and $\tilde{\gamma}(w)$ using $\tilde{z}_t(w)$ as the income-induced drift, and $\tilde{\psi}_t(w)$ as the income-induced diffusion. I transform these parameters back into the linear scale:
\begin{align*}
c(w) &= \tilde{c}_t(w)\sqrt{1+w^2} - \frac{1}{2}\tilde{\gamma}(w)w & \gamma(w) = \tilde{\gamma}(w)\sqrt{1+w^2}
\end{align*}
which gives estimates that can directly be interpreted as the mean and the variance of consumption.

\subsection{Robustness Checks}\label{sec:robustness-checks}

The estimation involves several hyperparameters: they include $\delta$ in the \citet{deming_statistical_1943} regression (see Section~\ref{sec:deming-estimation-procedure}), as well as a series of bandwidths that control the degree of smoothing involved when estimating specific densities and derivatives. In practice, the amount of noise in the data guides the choice of these hyperparameters. This section lists their default values and, more importantly, provides checks that confirm that the results are robust to variations of these parameters.

\begin{table}[ht]
\centering
\resizebox{\linewidth}{!}{
\begin{tabular}{@{}p{5cm}p{9cm}p{5cm}p{5cm}@{}}
\toprule
Hyperparameter &
  Role &
  Default value &
  Estimation method and kernel (if applicable) \\ \midrule
$\delta$ &
  Determines the ratio of the variance of the error terms between both sides of the \citet{deming_statistical_1943} regression. &
  Depends on the wealth bin (see Section~\ref{sec:deming-estimation-procedure}, as well as line ``Bandwidth for estimating measurement error variance'' in this table) &
   \\ \midrule
Bandwidth for mean income over time &
  The average amount of income that is received by each wealth bin can sensibly vary over time with the business cycle. This value is smoothed over time to focus on the long-run dynamics instead. &
  $12.5$ years &
  locally constant regression with rectangular kernel \\ \midrule
Bandwidth for income variance &
  The estimate of the variance of income conditional on wealth is smoothed across wealth, both to get a more stable estimate, and to be able to estimate the derivative with respect to wealth. &
  1 time average national income (asinh scale) &
  locally constant regression with rectangular kernel \\ \midrule
Bandwidth for the derivative of the log density of wealth &
  This bandwidth parameter is used to calculate the derivative of the log density of wealth based on the histogram estimate of the density. &
  $1.5$ times average national income (asinh scale) &
  locally linear regression with rectangular kernel \\ \midrule
Bandwidth for the inverse of the survival function of wealth &
  The inverse of the survival function of wealth ($F(w)/f(w)$) is computed as an intermediary step to obtain the left-hand side of equation~(\ref{eq:estimation-complete}). This estimate is smoothed over time using this bandwidth parameter. &
  2 years &
  locally constant regression with rectangular kernel \\ \midrule
Bandwidth for auxiliary effects &
  The impact of auxiliary effects (demography, etc.) on the CDF of wealth are microsimulated for every year, and then are smoothed over time using this bandwidth. &
  10 years &
  locally constant regression with rectangular kernel \\ \midrule
Bandwidth for estimating measurement error variance &
  To estimate the hyperparameter $\delta$ in the \citet{deming_statistical_1943} regression, I estimate the variance of the measurement error for both variables in the equation, using the variance of the residual between observed values and smoothed version using this bandwidth parameter. &
  5 years &
  locally constant regression with rectangular kernel \\ \midrule
Bandwidth for derivative of the diffusion &
  To adjust the value of the drift, I need to estimate the derivative of the diffusion with respect to wealth, using this bandwidth parameter. &
  $0.5$ time average national income (asinh scale) &
  locally linear regression with rectangular kernel \\ \bottomrule
\end{tabular}
}
\caption{List and Default Values for Hyperparameters}\label{tab:list-hyperparameters}
\end{table}

Table~\ref{tab:list-hyperparameters} describes the hyperparameters used in the estimation. Besides $\delta$, all of them are bandwidth parameters used for smoothing out noise or short-term variations in the data. I apply smoothing using locally constant regressions with a rectangular kernel (i.e., moving averages). When I need to estimate derivatives, I use the coefficient from a locally linear regression with a rectangular kernel as well.

\begin{figure}[ht]
\captionsetup[subfigure]{justification=centering}
\begin{center}
    \begin{subfigure}[t]{0.499\textwidth}
        \centering
        \includegraphics[width=\textwidth]{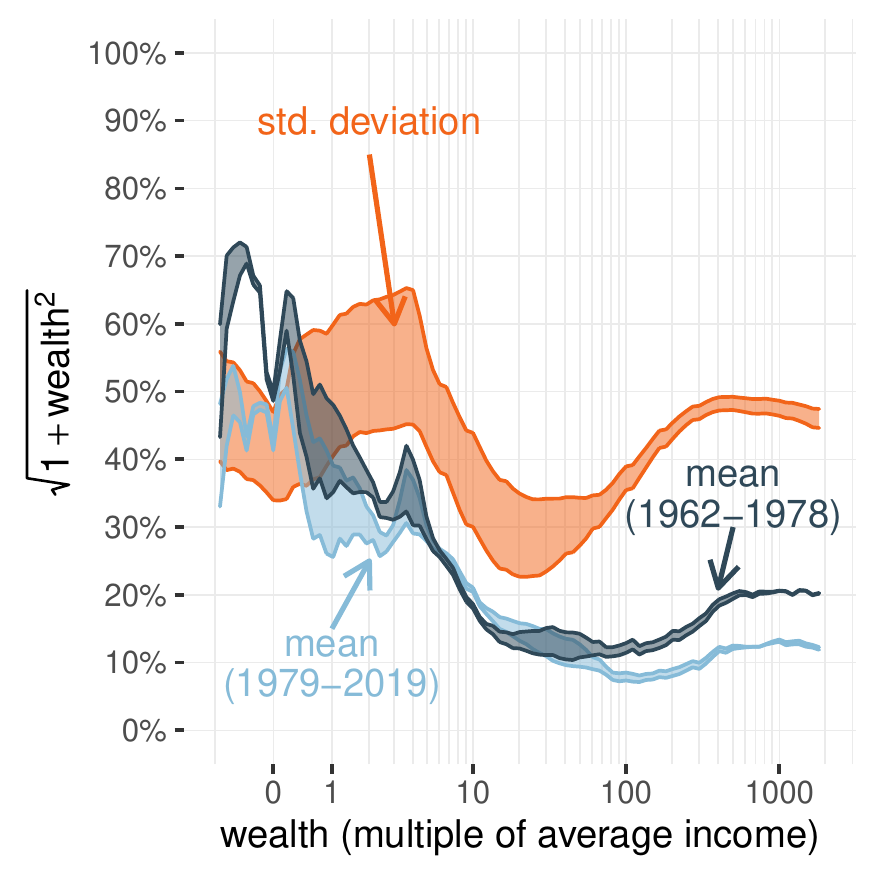}
        \caption{$\delta$ Hyperparameter in Deming Regression}
        \label{fig:robustness-delta}
        \begin{minipage}{0.9\linewidth}
            \footnotesize \textit{Note:} Full range of estimates obtained by varying the hyperparameter $\delta$ in the \citet{deming_statistical_1943} regression by factors $\alpha \in \{1/5, 1/2, 1, 2, 5\}$. See text for details.
        \end{minipage}
    \end{subfigure}%
    \begin{subfigure}[t]{0.499\textwidth}
        \centering
        \includegraphics[width=\textwidth]{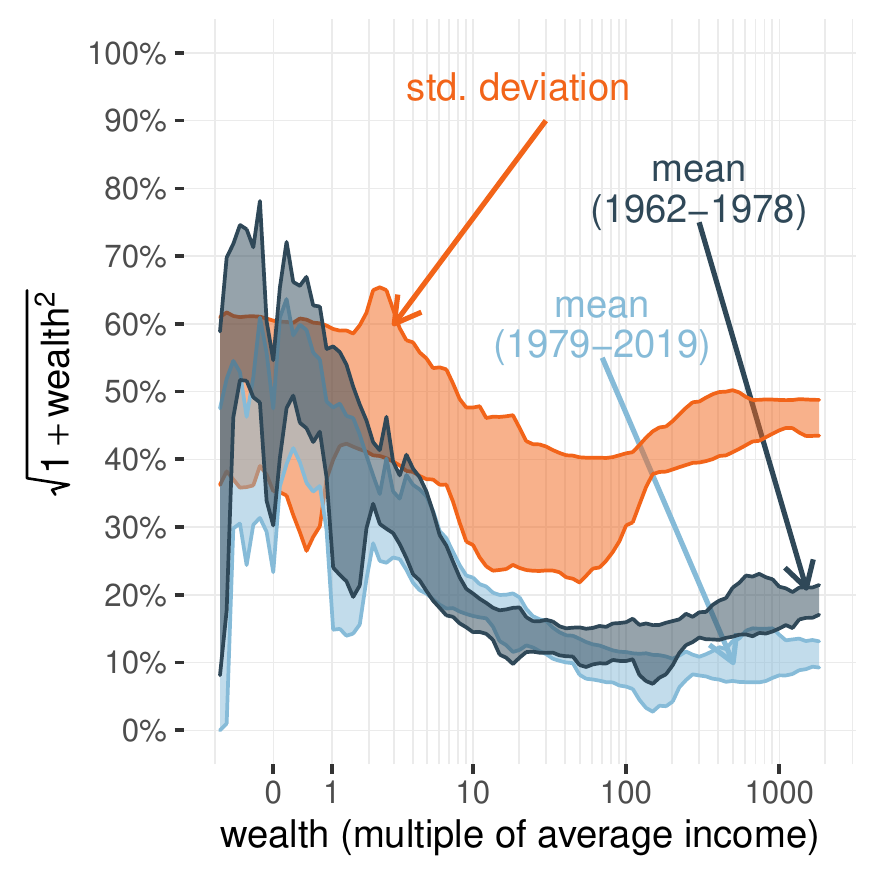}
        \caption{Bandwidth Hyperparameters}
        \label{fig:robustness-bw}
        \begin{minipage}{0.9\linewidth}
            \footnotesize \textit{Note:} 95\% interval range of the estimates obtained by varying the bandwidth parameters separately at random by factors $\alpha \in \{1/5, 1/4, 1/3, 1/2, 1, 2, 3, 4, 5\}$. See text for details.
        \end{minipage}
    \end{subfigure}
\caption{Robustness Checks}
\label{fig:robustness-propensity-consume}
\end{center}
\end{figure}

In Figure~\ref{fig:robustness-propensity-consume}, I show how the results change when I vary the hyperparameters: I find that the key findings are robust. Figure~\ref{fig:robustness-delta} focuses on $\delta$, the hyperparameter used in the \citet{deming_statistical_1943} regression. By default, this parameter is automatically determined using a rough estimate of the measurement error. But I can also modify this parameter directly and see how the results are affected. I adjust $\delta$ by a factor $\alpha \in \{1/5, 1/2, 1, 2, 5\}$, and then report the range of estimates we obtain using these different values. This range remains fairly narrow, especially for the drift, and at the top of the distribution. The limited impact of $\delta$ has a simple explanation. The parameter $\delta$ determines to what extent these residuals are attributed to the left-hand side or the right-hand side of equation~(\ref{eq:estimation-complete}). But the regression we run explains most of the data's variance, so the residuals remain small in all cases. As a result, where we attribute them has a limited impact.

Figure~\ref{fig:robustness-bw} does a similar exercise, this time focusing on the bandwidth parameters (with $\delta$ now being automatically calculated, using these bandwidth parameters). I perform 100 estimations of the parameters, where I separately vary each bandwidth parameter at random by a coefficient $\alpha \in \{1/5, 1/4, 1/3, 1/2, 1, 2, 3, 4, 5\}$.\footnote{I make an exception for the bandwidth parameter used to calculate the derivative of the diffusion. I only vary that parameter by a factor $\alpha \in \{1, 2, 3, 4, 5\}$, because undersmoothing in that part of the estimation leads to unrealistically noisy estimates.} Then, I show the range of final estimates where 95\% of estimates fall. These ranges are wider than the ones for the $\delta$ parameter alone, but the main results remain valid.

\subsection{Wealth Distribution with a Wealth Tax and a Lump-sum Rebate}\label{sec:lump-sum-rebate}

\begin{figure}[ht]
    \centering
    \begin{minipage}{0.7\linewidth}
    \includegraphics[width=\linewidth]{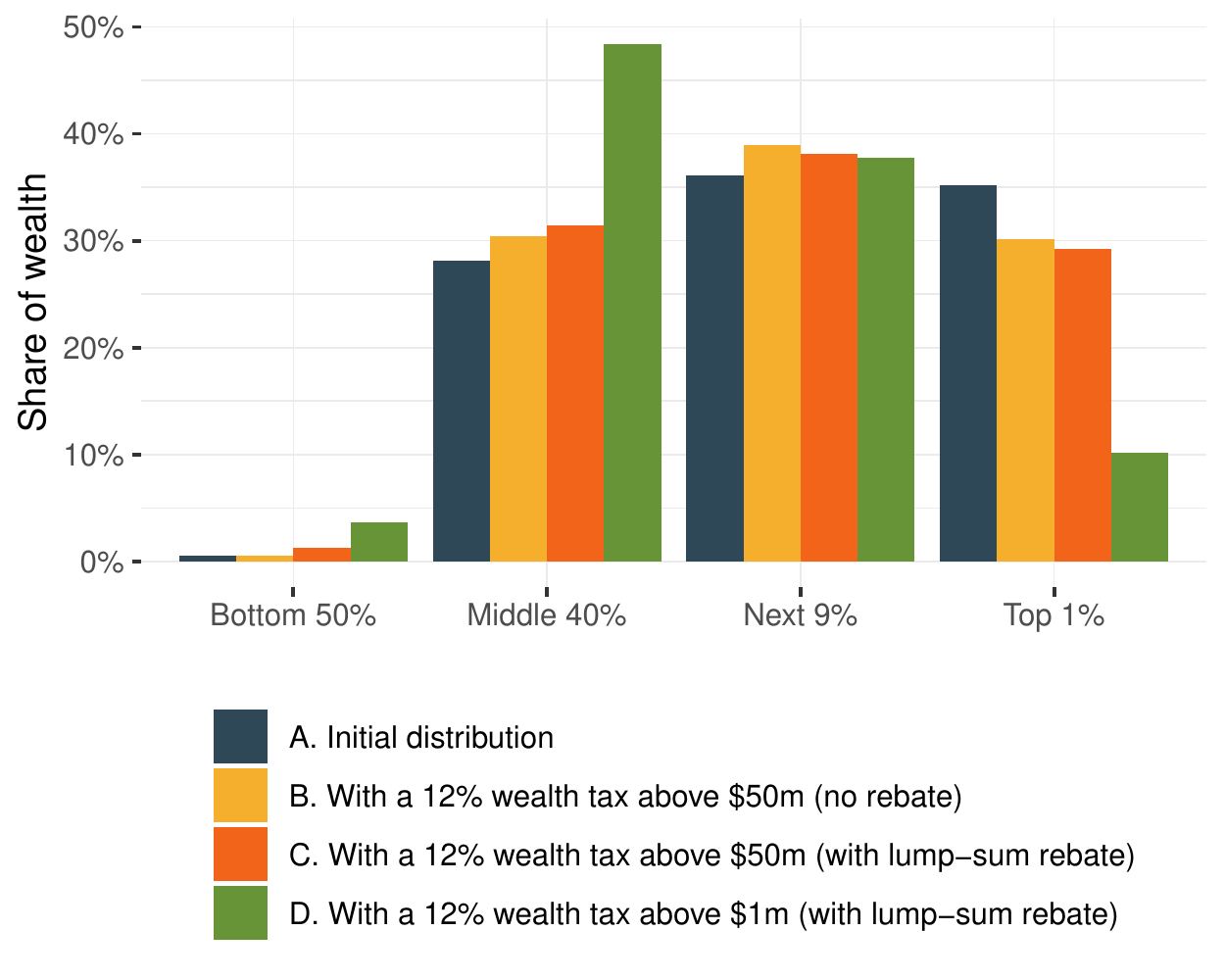}
    
    {\footnotesize \textit{Note:} This figure assumes benchmark parameters for behavioral responses (elasticity of tax avoidance $\epsilon = 1$ and elasticity of consumption $\eta = 1$.) For each bracket, the first bar shows the initial distribution of wealth, which corresponds to the year 2019 of the \gls{dina} data from \citet{saez_rise_2020}. The second bar corresponds to the long-run distribution, with a wealth tax of 12\% above \$50m (which maximizes revenue in the benchmark calibration.) The third bar shows the same distribution when the government rebates the tax revenue lump sum every year. The fourth show the same effect for a tax with a much larger base (all wealth above \$1m).}
    
    \caption{Long-run Changes of the Wealth Distribution with a Wealth Tax and and Lump-sum Rebate}
    \label{fig:distribution-wealth-tax-rebate}
    \end{minipage}
\end{figure}

This section provides some details on how a wealth tax would impact the distribution of wealth in the long run. For this exercise, I focus on a revenue-maximizing linear tax above \$50m, which in the benchmark calibration corresponds to a linear tax rate of 12\%. The tax government rebates the tax lump sum every year, so the dynamics of wealth are now:
\begin{equation*}
\dif w_{it} = (\mu(w_{it}) - \tau(w_{it}) + \bar{\tau})\dif t + \sigma_t(w_{it}) \dif B_{it}
\end{equation*}
where $\bar{\tau}$ is the average tax revenue. The steady-state wealth distribution is now characterized by:
\begin{equation*}
f^*(w) \propto \theta(w)[\lambda(w)]^{\bar{\tau}}f(w) \qquad \text{where} \qquad \lambda(w) = \exp\left\{2\int_{-\infty}^w \frac{1}{\sigma^2(s)}\,\dif s\right\}
\end{equation*}
and where $\theta(w)$ is defined similarly as in the no-rebate case (and, in the current example, accounts for behavioral responses.) The value $\bar{\tau}$ must satisfy the equation:
\begin{equation*}
\bar{\tau} - \frac{\int_{-\infty}^{+\infty}\tau(w)\theta(w)[\lambda(w)]^{\bar{\tau}}f(w)\,\dif w}{\int_{-\infty}^{+\infty}\theta(w)[\lambda(w)]^{\bar{\tau}}f(w)\,\dif w} = 0
\end{equation*}
which can be solved numerically. Figure~\ref{fig:distribution-wealth-tax-rebate} shows the results from this exercise. It reports the wealth shares of four brackets (the bottom 50\%, the middle 40\%, the next 9\%, and the top 1\%) and four scenarios (no wealth tax, wealth tax without rebate, wealth tax with rebate, broader wealth tax with rebate). The wealth distribution remains highly unequal in all cases: the bottom 50\% has virtually no wealth. Still, the wealth tax operates a significant amount of redistribution. The wealth tax alone primarily redistributes away from the top 1\% to the benefit of the middle 40\% and the next 9\%. But it does not affect the bottom 50\%. With the rebate, the bottom 50\% more than doubles its share, but from a very low baseline ($0.5\%$, to $1.3\%$). The tax with a broader base (all wealth above \$1m) increases it a bit more, to $3.6\%$.

%% file: references-extra.bib
@misc{HumanFertility,
    year = {2019},
    address = {Max Planck Institute for Demographic Research in Rostock, Germany and Vienna Institute of Demography in Vienna, Austria},
    title = {{Human Fertility Database}},
    url = {https://www.humanfertility.org/}
}

@misc{HumanFertilityCollection,
    year = {2019},
    address = {Max Planck Institute for Demographic Research in Rostock, Germany and Vienna Institute of Demography in Vienna, Austria},
    title = {{Human Fertility Collection}},
    url = {https://www.fertilitydata.org/}
}

@misc{HumanLifeTable,
    year = {2019},
    address = {Max Planck Institute for Demographic Research in Rostock, Germany, Department of Demography at University of California at Berkeley, USA and Institut d'{\'{e}}tudes d{\'{e}}mographiques (INED) in Paris, France},
    title = {{Human Life Table Database}},
    url = {https://www.lifetable.de/}
}

@misc{HumanMortalityDatabase,
    year = {2019},
    address = {Max Planck Institute for Demographic Research in Rostock, Germany, Department of Demography at University of California at Berkeley, USA},
    title = {{Human Mortality Database}},
    url = {https://www.mortality.org}
}

@misc{Gapminder,
    year = {2019},
    author = {Gapminder},
    title = {{Children per women since 1800}},
    url = {https://www.gapminder.org/news/children-per-women-since-1800-in-gapminder-world/}
}

@misc{NVSS,
    year = {2022},
    author = {{Center for Disease Control}},
    title = {{National Vital Statistics System}},
    url = {https://www.cdc.gov/nchs/nvss/index.htm}
}

@Manual{VineCopula,
  title = {VineCopula: Statistical Inference of Vine Copulas},
  author = {Ulf Schepsmeier and Jakob Stoeber and Eike Christian Brechmann and Benedikt Graeler and Thomas Nagler and Tobias Erhardt},
  year = {2018},
  note = {R package version 2.1.8},
  url = {https://CRAN.R-project.org/package=VineCopula},
}

@Manual{quantreg,
  title = {quantreg: Quantile Regression},
  author = {Roger Koenker},
  year = {2019},
  note = {R package version 5.40},
  url = {https://CRAN.R-project.org/package=quantreg},
}


%% file: references.bib
@article{grey_regular_1994,
	title = {Regular {Variation} in the {Tail} {Behaviour} of {Solutions} of {Random} {Difference} {Equations}},
	volume = {4},
	issn = {1050-5164},
	url = {https://projecteuclid.org/journals/annals-of-applied-probability/volume-4/issue-1/Regular-Variation-in-the-Tail-Behaviour-of-Solutions-of-Random/10.1214/aoap/1177005205.full},
	doi = {10.1214/aoap/1177005205},
	number = {1},
	urldate = {2022-09-14},
	journal = {The Annals of Applied Probability},
	author = {Grey, D. R.},
	year = {1994},
}

@article{grincevicius_one_1975,
	title = {One limit distribution for a random walk on the line},
	volume = {15},
	issn = {0363-1672, 1573-8825},
	url = {http://link.springer.com/10.1007/BF00969789},
	doi = {10.1007/BF00969789},
	language = {en},
	number = {4},
	urldate = {2022-09-14},
	journal = {Lithuanian Mathematical Journal},
	author = {Grincevićius, A. K.},
	year = {1975},
	pages = {580--589},
}

@article{straub_positive_2020,
	title = {Positive {Long}-{Run} {Capital} {Taxation}: {Chamley}-{Judd} {Revisited}},
	volume = {110},
	issn = {0002-8282},
	shorttitle = {Positive {Long}-{Run} {Capital} {Taxation}},
	url = {https://pubs.aeaweb.org/doi/10.1257/aer.20150210},
	doi = {10.1257/aer.20150210},
	language = {en},
	number = {1},
	urldate = {2022-09-14},
	journal = {American Economic Review},
	author = {Straub, Ludwig and Werning, Iván},
	year = {2020},
	pages = {86--119},
}

@article{bassetto_redistribution_2006,
	title = {Redistribution, taxes, and the median voter},
	volume = {9},
	issn = {10942025},
	url = {https://linkinghub.elsevier.com/retrieve/pii/S1094202506000056},
	doi = {10.1016/j.red.2006.02.001},
	language = {en},
	number = {2},
	urldate = {2022-09-14},
	journal = {Review of Economic Dynamics},
	author = {Bassetto, Marco and Benhabib, Jess},
	year = {2006},
	pages = {211--223},
}

@article{de_saporta_tail_2004,
	title = {Tail of a linear diffusion with {Markov} switching},
	volume = {339},
	issn = {1631073X},
	url = {https://linkinghub.elsevier.com/retrieve/pii/S1631073X0400439X},
	doi = {10.1016/j.crma.2004.09.022},
	language = {en},
	number = {9},
	urldate = {2022-09-14},
	journal = {Comptes Rendus Mathematique},
	author = {de Saporta, Benoîte and Yao, Jian-Feng},
	year = {2004},
	pages = {643--646},
}

@article{saez_optimal_2013,
	title = {Optimal progressive capital income taxes in the infinite horizon model},
	volume = {97},
	issn = {00472727},
	url = {https://linkinghub.elsevier.com/retrieve/pii/S0047272712001016},
	doi = {10.1016/j.jpubeco.2012.09.002},
	language = {en},
	urldate = {2022-09-14},
	journal = {Journal of Public Economics},
	author = {Saez, Emmanuel},
	year = {2013},
	pages = {61--74},
}

@article{fremeaux_inequalities_2020,
	title = {Inequalities and the individualization of wealth},
	volume = {184},
	issn = {00472727},
	url = {https://linkinghub.elsevier.com/retrieve/pii/S0047272720300098},
	doi = {10.1016/j.jpubeco.2020.104145},
	language = {en},
	urldate = {2022-09-06},
	journal = {Journal of Public Economics},
	author = {Frémeaux, Nicolas and Leturcq, Marion},
	year = {2020},
	pages = {104145},
}

@techreport{blanchet_real-time_2022,
	address = {Cambridge, MA},
	title = {Real-{Time} {Inequality}},
	url = {http://www.nber.org/papers/w30229.pdf},
	language = {en},
	number = {w30229},
	urldate = {2022-09-06},
	institution = {National Bureau of Economic Research},
	author = {Blanchet, Thomas and Saez, Emmanuel and Zucman, Gabriel},
	year = {2022},
	doi = {10.3386/w30229},
	pages = {w30229},
}

@techreport{derenoncourt_wealth_2022,
	address = {Cambridge, MA},
	title = {Wealth of {Two} {Nations}: {The} {U}.{S}. {Racial} {Wealth} {Gap}, 1860-2020},
	shorttitle = {Wealth of {Two} {Nations}},
	url = {http://www.nber.org/papers/w30101.pdf},
	language = {en},
	number = {w30101},
	urldate = {2022-09-06},
	institution = {National Bureau of Economic Research},
	author = {Derenoncourt, Ellora and Kim, Chi Hyun and Kuhn, Moritz and Schularick, Moritz},
	year = {2022},
	doi = {10.3386/w30101},
	pages = {w30101},
}

@unpublished{gaillard_wealth_2021,
	title = {Wealth, {Returns}, and {Taxation}: {A} {Tale} of {Two} {Dependencies}},
	shorttitle = {Wealth, {Returns}, and {Taxation}},
	url = {https://www.ssrn.com/abstract=3966130},
	language = {en},
	urldate = {2022-09-06},
	author = {Gaillard, Alexandre and Wangner, Philipp},
	year = {2021},
}

@unpublished{zoutman_elasticity_2018,
	title = {The {Elasticity} of {Taxable} {Wealth}: {Evidence} from the {Netherlands}},
	author = {Zoutman, Floris},
	year = {2018},
}

@article{smith_top_2022,
	title = {Top {Wealth} in {America}: {New} {Estimates} under {Heterogeneous} {Returns}},
	issn = {0033-5533, 1531-4650},
	shorttitle = {Top {Wealth} in {America}},
	url = {https://academic.oup.com/qje/advance-article/doi/10.1093/qje/qjac033/6678447},
	doi = {10.1093/qje/qjac033},
	language = {en},
	urldate = {2022-08-29},
	journal = {The Quarterly Journal of Economics},
	author = {Smith, Matthew and Zidar, Owen and Zwick, Eric},
	year = {2022},
	pages = {qjac033},
}

@article{atkinson_design_1976,
	title = {The design of tax structure: {Direct} versus indirect taxation},
	volume = {6},
	issn = {00472727},
	shorttitle = {The design of tax structure},
	url = {https://linkinghub.elsevier.com/retrieve/pii/0047272776900414},
	doi = {10.1016/0047-2727(76)90041-4},
	language = {en},
	number = {1-2},
	urldate = {2022-08-25},
	journal = {Journal of Public Economics},
	author = {Atkinson, A.B. and Stiglitz, J.E.},
	year = {1976},
	pages = {55--75},
}

@techreport{auerbach_taxation_2001,
	address = {Cambridge, MA},
	title = {Taxation and {Economic} {Efficiency}},
	url = {http://www.nber.org/papers/w8181.pdf},
	language = {en},
	number = {w8181},
	urldate = {2022-08-17},
	institution = {National Bureau of Economic Research},
	author = {Auerbach, Alan and Hines, James},
	year = {2001},
	doi = {10.3386/w8181},
	pages = {w8181},
}

@unpublished{ring_wealth_2020,
	title = {Wealth {Taxation} and {Household} {Saving}: {Evidence} from {Assessment} {Discontinuities} in {Norway}},
	shorttitle = {Wealth {Taxation} and {Household} {Saving}},
	url = {https://www.ssrn.com/abstract=3716257},
	language = {en},
	urldate = {2022-08-17},
	author = {Ring, Marius Alexander Kalleberg},
	year = {2020},
}

@article{vervaat_stochastic_1979,
	title = {On a stochastic difference equation and a representation of non–negative infinitely divisible random variables},
	volume = {11},
	issn = {0001-8678, 1475-6064},
	url = {https://www.cambridge.org/core/product/identifier/S0001867800033024/type/journal_article},
	doi = {10.2307/1426858},
	language = {en},
	number = {4},
	urldate = {2022-08-16},
	journal = {Advances in Applied Probability},
	author = {Vervaat, Wim},
	year = {1979},
	pages = {750--783},
}

@article{kolmogorov_uber_1931,
	title = {Über die analytischen {Methoden} in der {Wahrscheinlichkeitsrechnung}},
	volume = {104},
	issn = {0025-5831, 1432-1807},
	url = {http://link.springer.com/10.1007/BF01457949},
	doi = {10.1007/BF01457949},
	language = {de},
	number = {1},
	urldate = {2022-08-05},
	journal = {Mathematische Annalen},
	author = {Kolmogorov, A.},
	year = {1931},
	pages = {415--458},
}

@article{goldie_implicit_1991,
	title = {Implicit {Renewal} {Theory} and {Tails} of {Solutions} of {Random} {Equations}},
	volume = {1},
	issn = {1050-5164},
	url = {https://projecteuclid.org/journals/annals-of-applied-probability/volume-1/issue-1/Implicit-Renewal-Theory-and-Tails-of-Solutions-of-Random-Equations/10.1214/aoap/1177005985.full},
	doi = {10.1214/aoap/1177005985},
	number = {1},
	urldate = {2022-08-16},
	journal = {The Annals of Applied Probability},
	author = {Goldie, Charles M.},
	year = {1991},
}

@article{bach_rich_2020,
	title = {Rich {Pickings}? {Risk}, {Return}, and {Skill} in {Household} {Wealth}},
	volume = {110},
	issn = {0002-8282},
	shorttitle = {Rich {Pickings}?},
	url = {https://pubs.aeaweb.org/doi/10.1257/aer.20170666},
	doi = {10.1257/aer.20170666},
	language = {en},
	number = {9},
	urldate = {2022-08-12},
	journal = {American Economic Review},
	author = {Bach, Laurent and Calvet, Laurent E. and Sodini, Paolo},
	year = {2020},
	pages = {2703--2747},
}

@article{benhabib_wealth_2019,
	title = {Wealth {Distribution} and {Social} {Mobility} in the {US}: {A} {Quantitative} {Approach}},
	volume = {109},
	issn = {0002-8282},
	shorttitle = {Wealth {Distribution} and {Social} {Mobility} in the {US}},
	url = {https://pubs.aeaweb.org/doi/10.1257/aer.20151684},
	doi = {10.1257/aer.20151684},
	language = {en},
	number = {5},
	urldate = {2022-08-12},
	journal = {American Economic Review},
	author = {Benhabib, Jess and Bisin, Alberto and Luo, Mi},
	year = {2019},
	pages = {1623--1647},
}

@article{fagereng_heterogeneity_2020,
	title = {Heterogeneity and {Persistence} in {Returns} to {Wealth}},
	volume = {88},
	issn = {0012-9682},
	url = {https://www.econometricsociety.org/doi/10.3982/ECTA14835},
	doi = {10.3982/ECTA14835},
	language = {en},
	number = {1},
	urldate = {2022-08-12},
	journal = {Econometrica},
	author = {Fagereng, Andreas and Guiso, Luigi and Malacrino, Davide and Pistaferri, Luigi},
	year = {2020},
	pages = {115--170},
}

@article{cagetti_entrepreneurship_2006,
	title = {Entrepreneurship, {Frictions}, and {Wealth}},
	volume = {114},
	issn = {0022-3808, 1537-534X},
	url = {https://www.journals.uchicago.edu/doi/10.1086/508032},
	doi = {10.1086/508032},
	language = {en},
	number = {5},
	urldate = {2022-08-12},
	journal = {Journal of Political Economy},
	author = {Cagetti, Marco and De Nardi, Mariacristina},
	year = {2006},
	pages = {835--870},
}

@article{quadrini_entrepreneurship_2000,
	title = {Entrepreneurship, {Saving}, and {Social} {Mobility}},
	volume = {3},
	issn = {10942025},
	url = {https://linkinghub.elsevier.com/retrieve/pii/S1094202599900777},
	doi = {10.1006/redy.1999.0077},
	language = {en},
	number = {1},
	urldate = {2022-08-12},
	journal = {Review of Economic Dynamics},
	author = {Quadrini, Vincenzo},
	year = {2000},
	pages = {1--40},
}

@article{de_nardi_wealth_2004,
	title = {Wealth {Inequality} and {Intergenerational} {Links}},
	volume = {71},
	issn = {1467-937X, 0034-6527},
	url = {https://academic.oup.com/restud/article-lookup/doi/10.1111/j.1467-937X.2004.00302.x},
	doi = {10.1111/j.1467-937X.2004.00302.x},
	language = {en},
	number = {3},
	urldate = {2022-08-12},
	journal = {The Review of Economic Studies},
	author = {De Nardi, Mariacristina},
	year = {2004},
	pages = {743--768},
}

@techreport{carroll_why_1998,
	address = {Cambridge, MA},
	title = {Why {Do} the {Rich} {Save} {So} {Much}?},
	url = {http://www.nber.org/papers/w6549.pdf},
	language = {en},
	number = {w6549},
	urldate = {2022-08-12},
	institution = {National Bureau of Economic Research},
	author = {Carroll, Christopher},
	year = {1998},
	doi = {10.3386/w6549},
	pages = {w6549},
}

@article{krusell_income_1998,
	title = {Income and {Wealth} {Heterogeneity} in the {Macroeconomy}},
	volume = {106},
	issn = {0022-3808, 1537-534X},
	url = {https://www.journals.uchicago.edu/doi/10.1086/250034},
	doi = {10.1086/250034},
	language = {en},
	number = {5},
	urldate = {2022-08-12},
	journal = {Journal of Political Economy},
	author = {Krusell, Per and Smith, Jr., Anthony A.},
	year = {1998},
	pages = {867--896},
}

@article{bewley_permanent_1977,
	title = {The permanent income hypothesis: {A} theoretical formulation},
	volume = {16},
	issn = {00220531},
	shorttitle = {The permanent income hypothesis},
	url = {https://linkinghub.elsevier.com/retrieve/pii/0022053177900096},
	doi = {10.1016/0022-0531(77)90009-6},
	language = {en},
	number = {2},
	urldate = {2022-08-11},
	journal = {Journal of Economic Theory},
	author = {Bewley, Truman},
	year = {1977},
	pages = {252--292},
}

@article{huggett_risk-free_1993,
	title = {The risk-free rate in heterogeneous-agent incomplete-insurance economies},
	volume = {17},
	issn = {01651889},
	url = {https://linkinghub.elsevier.com/retrieve/pii/016518899390024M},
	doi = {10.1016/0165-1889(93)90024-M},
	language = {en},
	number = {5-6},
	urldate = {2022-08-11},
	journal = {Journal of Economic Dynamics and Control},
	author = {Huggett, Mark},
	year = {1993},
	pages = {953--969},
}

@unpublished{xavier_wealth_2021,
	title = {Wealth {Inequality} in the {US}: the {Role} of {Heterogeneous} {Returns}},
	shorttitle = {Wealth {Inequality} in the {US}},
	url = {https://www.ssrn.com/abstract=3915439},
	language = {en},
	urldate = {2022-08-11},
	author = {Xavier, Inês},
	year = {2021},
}

@unpublished{cioffi_heterogeneous_2021,
	title = {Heterogeneous {Risk} {Exposure} and the {Dynamics} of {Wealth} {Inequality}},
	url = {https://rcioffi.com/files/jmp/cioffi_jmp2021_princeton.pdf},
	author = {Cioffi, Riccardo A},
	year = {2021},
}

@article{saez_using_2001,
	title = {Using {Elasticities} to {Derive} {Optimal} {Income} {Tax} {Rates}},
	volume = {68},
	issn = {0034-6527, 1467-937X},
	url = {https://academic.oup.com/restud/article-lookup/doi/10.1111/1467-937X.00166},
	doi = {10.1111/1467-937X.00166},
	language = {en},
	number = {1},
	urldate = {2022-04-24},
	journal = {The Review of Economic Studies},
	author = {Saez, Emmanuel},
	year = {2001},
	pages = {205--229},
}

@article{kesten_random_1973,
	title = {Random difference equations and {Renewal} theory for products of random matrices},
	volume = {131},
	issn = {0001-5962},
	url = {http://projecteuclid.org/euclid.acta/1485889791},
	doi = {10.1007/BF02392040},
	language = {en},
	number = {0},
	urldate = {2022-05-31},
	journal = {Acta Mathematica},
	author = {Kesten, Harry},
	year = {1973},
	pages = {207--248},
}

@article{benhabib_distribution_2011,
	title = {The {Distribution} of {Wealth} and {Fiscal} {Policy} in {Economies} {With} {Finitely} {Lived} {Agents}},
	volume = {79},
	issn = {0012-9682},
	url = {http://doi.wiley.com/10.3982/ECTA8416},
	doi = {10.3982/ECTA8416},
	language = {en},
	number = {1},
	urldate = {2022-05-30},
	journal = {Econometrica},
	author = {Benhabib, Jess and Bisin, Alberto and Zhu, Shenghao},
	year = {2011},
	pages = {123--157},
}

@unpublished{hubmer_comprehensive_2018,
	title = {A {Comprehensive} {Quantitative} {Theory} of the {U}.{S}. {Wealth} {Distribution}},
	url = {http://www.econ.yale.edu/smith/hks_v18.pdf},
	author = {Hubmer, Joachim and Krusell, Per and Smith, Anthony A},
	year = {2018},
}

@book{piketty_capital_2014,
	address = {Cambridge, Massachusetts London},
	title = {Capital in the twenty-first century},
	isbn = {978-0-674-97985-7},
	language = {eng},
	publisher = {The Belknap Press of Harvard University Press},
	author = {Piketty, Thomas},
	translator = {Goldhammer, Arthur},
	year = {2014},
}

@article{saez_progressive_2019,
	title = {Progressive {Wealth} {Taxation}},
	number = {Fall},
	journal = {Brookings Papers on Economic Activity},
	author = {Saez, Emmanuel and Zucman, Gabriel},
	year = {2019},
	pages = {437--533},
}

@techreport{brulhart_taxing_2016,
	address = {Cambridge, MA},
	title = {Taxing {Wealth}: {Evidence} from {Switzerland}},
	shorttitle = {Taxing {Wealth}},
	url = {http://www.nber.org/papers/w22376.pdf},
	language = {en},
	number = {w22376},
	urldate = {2022-05-30},
	institution = {National Bureau of Economic Research},
	author = {Brülhart, Marius and Gruber, Jonathan and Krapf, Matthias and Schmidheiny, Kurt},
	year = {2016},
	doi = {10.3386/w22376},
	pages = {w22376},
}

@article{jakobsen_wealth_2020,
	title = {Wealth {Taxation} and {Wealth} {Accumulation}: {Theory} and {Evidence} {From} {Denmark}*},
	volume = {135},
	issn = {0033-5533, 1531-4650},
	shorttitle = {Wealth {Taxation} and {Wealth} {Accumulation}},
	url = {https://academic.oup.com/qje/article/135/1/329/5584349},
	doi = {10.1093/qje/qjz032},
	language = {en},
	number = {1},
	urldate = {2022-05-30},
	journal = {The Quarterly Journal of Economics},
	author = {Jakobsen, Katrine and Jakobsen, Kristian and Kleven, Henrik and Zucman, Gabriel},
	year = {2020},
	pages = {329--388},
}

@article{seim_behavioral_2017,
	title = {Behavioral {Responses} to {Wealth} {Taxes}: {Evidence} from {Sweden}},
	volume = {9},
	issn = {1945-7731, 1945-774X},
	shorttitle = {Behavioral {Responses} to {Wealth} {Taxes}},
	url = {https://pubs.aeaweb.org/doi/10.1257/pol.20150290},
	doi = {10.1257/pol.20150290},
	language = {en},
	number = {4},
	urldate = {2022-05-30},
	journal = {American Economic Journal: Economic Policy},
	author = {Seim, David},
	year = {2017},
	pages = {395--421},
}

@article{londono-velez_enforcing_2021,
	title = {Enforcing {Wealth} {Taxes} in the {Developing} {World}: {Quasi}-{Experimental} {Evidence} from {Colombia}},
	volume = {3},
	issn = {2640-205X, 2640-2068},
	shorttitle = {Enforcing {Wealth} {Taxes} in the {Developing} {World}},
	url = {https://pubs.aeaweb.org/doi/10.1257/aeri.20200319},
	doi = {10.1257/aeri.20200319},
	language = {en},
	number = {2},
	urldate = {2022-05-30},
	journal = {American Economic Review: Insights},
	author = {Londoño-Vélez, Juliana and Ávila-Mahecha, Javier},
	year = {2021},
	pages = {131--148},
}

@article{aiyagari_uninsured_1994,
	title = {Uninsured {Idiosyncratic} {Risk} and {Aggregate} {Saving}},
	volume = {109},
	issn = {0033-5533, 1531-4650},
	url = {https://academic.oup.com/qje/article-lookup/doi/10.2307/2118417},
	doi = {10.2307/2118417},
	language = {en},
	number = {3},
	urldate = {2022-05-30},
	journal = {The Quarterly Journal of Economics},
	author = {Aiyagari, S. R.},
	year = {1994},
	pages = {659--684},
}

@article{farhi_estate_2013,
	title = {Estate {Taxation} with {Altruism} {Heterogeneity}},
	volume = {103},
	issn = {0002-8282},
	url = {https://pubs.aeaweb.org/doi/10.1257/aer.103.3.489},
	doi = {10.1257/aer.103.3.489},
	language = {en},
	number = {3},
	urldate = {2022-05-30},
	journal = {American Economic Review},
	author = {Farhi, Emmanuel and Werning, Iván},
	year = {2013},
	pages = {489--495},
}

@article{farhi_capital_2010,
	title = {Capital {Taxation} and {Ownership} {When} {Markets} {Are} {Incomplete}},
	volume = {118},
	issn = {0022-3808, 1537-534X},
	url = {https://www.journals.uchicago.edu/doi/10.1086/657996},
	doi = {10.1086/657996},
	language = {en},
	number = {5},
	urldate = {2022-05-30},
	journal = {Journal of Political Economy},
	author = {Farhi, Emmanuel},
	year = {2010},
	pages = {908--948},
}

@article{judd_redistributive_1985,
	title = {Redistributive taxation in a simple perfect foresight model},
	volume = {28},
	issn = {00472727},
	url = {https://linkinghub.elsevier.com/retrieve/pii/0047272785900209},
	doi = {10.1016/0047-2727(85)90020-9},
	language = {en},
	number = {1},
	urldate = {2022-05-30},
	journal = {Journal of Public Economics},
	author = {Judd, Kenneth L.},
	year = {1985},
	pages = {59--83},
}

@article{saez_rise_2020,
	title = {The {Rise} of {Income} and {Wealth} {Inequality} in {America}: {Evidence} from {Distributional} {Macroeconomic} {Accounts}},
	volume = {34},
	issn = {0895-3309},
	shorttitle = {The {Rise} of {Income} and {Wealth} {Inequality} in {America}},
	url = {https://pubs.aeaweb.org/doi/10.1257/jep.34.4.3},
	doi = {10.1257/jep.34.4.3},
	language = {en},
	number = {4},
	urldate = {2022-05-18},
	journal = {Journal of Economic Perspectives},
	author = {Saez, Emmanuel and Zucman, Gabriel},
	year = {2020},
	pages = {3--26},
}

@article{chamley_optimal_1986,
	title = {Optimal {Taxation} of {Capital} {Income} in {General} {Equilibrium} with {Infinite} {Lives}},
	volume = {54},
	issn = {00129682},
	url = {https://www.jstor.org/stable/1911310?origin=crossref},
	doi = {10.2307/1911310},
	number = {3},
	urldate = {2022-05-30},
	journal = {Econometrica},
	author = {Chamley, Christophe},
	year = {1986},
	pages = {607},
}

@article{piketty_theory_2013,
	title = {A {Theory} of {Optimal} {Inheritance} {Taxation}},
	volume = {81},
	issn = {0012-9682},
	url = {http://doi.wiley.com/10.3982/ECTA10712},
	doi = {10.3982/ECTA10712},
	language = {en},
	number = {5},
	urldate = {2022-05-30},
	journal = {Econometrica},
	author = {Piketty, Thomas and Saez, Emmanuel},
	year = {2013},
	pages = {1851--1886},
}

@incollection{piketty_wealth_2015,
	title = {Wealth and {Inheritance} in the {Long} {Run}},
	volume = {2},
	isbn = {978-0-444-59430-3},
	url = {https://linkinghub.elsevier.com/retrieve/pii/B9780444594297000169},
	language = {en},
	urldate = {2022-05-29},
	booktitle = {Handbook of {Income} {Distribution}},
	publisher = {Elsevier},
	author = {Piketty, Thomas and Zucman, Gabriel},
	year = {2015},
	doi = {10.1016/B978-0-444-59429-7.00016-9},
	pages = {1303--1368},
}

@article{gomez_decomposing_2022,
	title = {Decomposing the {Growth} of {Top} {Wealth} {Shares}},
	url = {https://www.matthieugomez.com/files/topshares.pdf},
	language = {en},
	journal = {Econometrica},
	author = {Gomez, Matthieu},
	year = {2022},
	pages = {61},
}

@book{deming_statistical_1943,
	address = {New York},
	title = {Statistical adjustment of data},
	isbn = {978-0-486-64685-5},
	language = {eng},
	publisher = {Wiley},
	author = {Deming, W. Edwards},
	year = {1943},
}

@unpublished{robbins_capital_2018,
	title = {Capital {Gains} and the {Distribution} of {Income} on the {United} {States}},
	url = {https://users.nber.org/~robbinsj/jr_inequ_jmp.pdf},
	author = {Robbins, Jacob A},
	year = {2018},
}

@techreport{alstadsaeter_accounting_2016,
	address = {Cambridge, MA},
	title = {Accounting for {Business} {Income} in {Measuring} {Top} {Income} {Shares}: {Integrated} {Accrual} {Approach} {Using} {Individual} and {Firm} {Data} from {Norway}},
	shorttitle = {Accounting for {Business} {Income} in {Measuring} {Top} {Income} {Shares}},
	url = {http://www.nber.org/papers/w22888.pdf},
	language = {en},
	number = {w22888},
	urldate = {2022-05-18},
	institution = {National Bureau of Economic Research},
	author = {Alstadsæter, Annette and Jacob, Martin and Kopczuk, Wojciech and Telle, Kjetil},
	year = {2016},
	doi = {10.3386/w22888},
	pages = {w22888},
}

@article{menchik_primogeniture_1980,
	title = {Primogeniture, {Equal} {Sharing}, and the {U}.{S}. {Distribution} of {Wealth}},
	volume = {94},
	issn = {00335533},
	url = {https://academic.oup.com/qje/article-lookup/doi/10.2307/1884542},
	doi = {10.2307/1884542},
	number = {2},
	urldate = {2022-05-17},
	journal = {The Quarterly Journal of Economics},
	author = {Menchik, Paul L.},
	year = {1980},
	pages = {299},
}

@article{haines_estimated_1998,
	title = {Estimated {Life} {Tables} for the {United} {States}, 1850–1910},
	volume = {31},
	issn = {0161-5440, 1940-1906},
	url = {http://www.tandfonline.com/doi/abs/10.1080/01615449809601197},
	doi = {10.1080/01615449809601197},
	language = {en},
	number = {4},
	urldate = {2022-05-17},
	journal = {Historical Methods: A Journal of Quantitative and Interdisciplinary History},
	author = {Haines, Michael R.},
	year = {1998},
	pages = {149--169},
}

@misc{united_states_census_bureau_national_2021,
	title = {National {Intercensal} {Tables}: 1900-1990},
	url = {https://www.census.gov/content/census/en/data/tables/time-series/demo/popest/pre-1980-national.html},
	author = {{United States Census Bureau}},
	year = {2021},
}

@book{united_nations_world_2019,
	title = {World {Population} {Prospects}, 2019 {Revision}},
	isbn = {978-92-1-148316-1},
	url = {https://population.un.org/wpp/Publications/},
	language = {English},
	author = {{United Nations}},
	year = {2019},
	note = {OCLC: 1142478963},
}

@misc{ruggles_steven_ipums_2022,
	title = {{IPUMS} {USA}: {Version} 12.0},
	shorttitle = {{IPUMS} {USA}},
	url = {https://usa.ipums.org},
	language = {en},
	urldate = {2022-05-17},
	publisher = {Minneapolis, MN: IPUMS},
	author = {Ruggles, Steven and Flood, Sarah and Foster, Sophia and Goeken, Ronald and Schouweiler, Megan and Sobek, Matthew},
	collaborator = {United States Census Bureau},
	year = {2022},
	doi = {10.18128/D010.V12.0},
	note = {Version Number: 12.0
Type: dataset},
}

@article{fritsch_monotone_1980,
	title = {Monotone {Piecewise} {Cubic} {Interpolation}},
	volume = {17},
	issn = {0036-1429, 1095-7170},
	url = {http://epubs.siam.org/doi/10.1137/0717021},
	doi = {10.1137/0717021},
	language = {en},
	number = {2},
	urldate = {2022-05-17},
	journal = {SIAM Journal on Numerical Analysis},
	author = {Fritsch, F. N. and Carlson, R. E.},
	year = {1980},
	pages = {238--246},
}

@article{saez_generalized_2016,
	title = {Generalized {Social} {Marginal} {Welfare} {Weights} for {Optimal} {Tax} {Theory}},
	volume = {106},
	issn = {0002-8282},
	url = {https://pubs.aeaweb.org/doi/10.1257/aer.20141362},
	doi = {10.1257/aer.20141362},
	language = {en},
	number = {1},
	urldate = {2022-04-27},
	journal = {American Economic Review},
	author = {Saez, Emmanuel and Stantcheva, Stefanie},
	year = {2016},
	pages = {24--45},
}

@article{carroll_buffer-stock_1997,
	title = {Buffer-{Stock} {Saving} and the {Life} {Cycle}/{Permanent} {Income} {Hypothesis}},
	volume = {112},
	issn = {0033-5533, 1531-4650},
	url = {https://academic.oup.com/qje/article-lookup/doi/10.1162/003355397555109},
	doi = {10.1162/003355397555109},
	language = {en},
	number = {1},
	urldate = {2022-04-28},
	journal = {The Quarterly Journal of Economics},
	author = {Carroll, Christopher D.},
	year = {1997},
	pages = {1--55},
}

@book{friedman_theory_1957,
	address = {Princeton (N. J.)},
	series = {National bureau of economic research},
	title = {A {Theory} of the {Consumption} {Function}},
	isbn = {978-0-691-04182-7},
	language = {eng},
	number = {63},
	publisher = {Princeton university press},
	author = {Friedman, Milton},
	year = {1957},
}

@article{gabaix_dynamics_2016,
	title = {The {Dynamics} of {Inequality}},
	volume = {84},
	issn = {0012-9682},
	url = {https://www.econometricsociety.org/doi/10.3982/ECTA13569},
	doi = {10.3982/ECTA13569},
	language = {en},
	number = {6},
	urldate = {2022-04-28},
	journal = {Econometrica},
	author = {Gabaix, Xavier and Lasry, Jean-Michel and Lions, Pierre-Louis and Moll, Benjamin},
	year = {2016},
	pages = {2071--2111},
}

@incollection{piketty_optimal_2013,
	title = {Optimal {Labor} {Income} {Taxation}},
	volume = {5},
	isbn = {978-0-444-53759-1},
	url = {https://linkinghub.elsevier.com/retrieve/pii/B9780444537591000078},
	language = {en},
	urldate = {2022-04-24},
	booktitle = {Handbook of {Public} {Economics}},
	publisher = {Elsevier},
	author = {Piketty, Thomas and Saez, Emmanuel},
	year = {2013},
	doi = {10.1016/B978-0-444-53759-1.00007-8},
	pages = {391--474},
}

@article{gabaix_power_2009,
	title = {Power {Laws} in {Economics} and {Finance}},
	volume = {1},
	issn = {1941-1383, 1941-1391},
	url = {https://www.annualreviews.org/doi/10.1146/annurev.economics.050708.142940},
	doi = {10.1146/annurev.economics.050708.142940},
	language = {en},
	number = {1},
	journal = {Annual Review of Economics},
	author = {Gabaix, Xavier},
	year = {2009},
	pages = {255--294},
}

@article{garbinti_accounting_2021,
	title = {Accounting for {Wealth}-{Inequality} {Dynamics}: {Methods}, {Estimates}, and {Simulations} for {France}},
	volume = {19},
	issn = {1542-4766, 1542-4774},
	shorttitle = {Accounting for {Wealth}-{Inequality} {Dynamics}},
	url = {https://academic.oup.com/jeea/article/19/1/620/5846044},
	doi = {10.1093/jeea/jvaa025},
	language = {en},
	number = {1},
	journal = {Journal of the European Economic Association},
	author = {Garbinti, Bertrand and Goupille-Lebret, Jonathan and Piketty, Thomas},
	year = {2021},
	pages = {620--663},
}

@article{saez_simpler_2018,
	title = {A simpler theory of optimal capital taxation},
	volume = {162},
	issn = {00472727},
	url = {https://linkinghub.elsevier.com/retrieve/pii/S0047272717301688},
	doi = {10.1016/j.jpubeco.2017.10.004},
	language = {en},
	urldate = {2021-10-09},
	journal = {Journal of Public Economics},
	author = {Saez, Emmanuel and Stantcheva, Stefanie},
	year = {2018},
	pages = {120--142},
}

@article{gyongy_mimicking_1986,
	title = {Mimicking the one-dimensional marginal distributions of processes having an ito differential},
	volume = {71},
	issn = {0178-8051, 1432-2064},
	url = {http://link.springer.com/10.1007/BF00699039},
	doi = {10.1007/BF00699039},
	language = {en},
	number = {4},
	urldate = {2021-10-08},
	journal = {Probability Theory and Related Fields},
	author = {Gyöngy, I.},
	year = {1986},
	pages = {501--516},
}

@techreport{saez_trends_2020,
	address = {Cambridge, MA},
	title = {Trends in {US} {Income} and {Wealth} {Inequality}: {Revising} {After} the {Revisionists}},
	shorttitle = {Trends in {US} {Income} and {Wealth} {Inequality}},
	url = {http://www.nber.org/papers/w27921.pdf},
	language = {en},
	number = {27921},
	urldate = {2021-10-07},
	institution = {National Bureau of Economic Research},
	author = {Saez, Emmanuel and Zucman, Gabriel},
	year = {2020},
	doi = {10.3386/w27921},
}

@article{kuhn_income_2020,
	title = {Income and {Wealth} {Inequality} in {America}, 1949–2016},
	volume = {128},
	issn = {0022-3808, 1537-534X},
	url = {https://www.journals.uchicago.edu/doi/10.1086/708815},
	doi = {10.1086/708815},
	language = {en},
	number = {9},
	urldate = {2021-10-07},
	journal = {Journal of Political Economy},
	author = {Kuhn, Moritz and Schularick, Moritz and Steins, Ulrike I.},
	year = {2020},
	pages = {3469--3519},
}

@article{piketty_distributional_2018,
	title = {Distributional {National} {Accounts}: {Methods} and {Estimates} for the {United} {States}},
	volume = {133},
	issn = {0033-5533},
	url = {https://doi.org/10.1093/qje/qjx043},
	doi = {10.1093/qje/qjx043},
	number = {2},
	journal = {The Quarterly Journal of Economics},
	author = {Piketty, Thomas and Saez, Emmanuel and Zucman, Gabriel},
	year = {2018},
	pages = {553--609},
}

@article{saez_wealth_2016,
	title = {Wealth {Inequality} in the {United} {States} since 1913: {Evidence} from {Capitalized} {Income} {Tax} {Data}},
	volume = {131},
	issn = {0033-5533},
	url = {https://doi.org/10.1093/qje/qjw004},
	doi = {10.1093/qje/qjw004},
	number = {2},
	journal = {The Quarterly Journal of Economics},
	author = {Saez, Emmanuel and Zucman, Gabriel},
	year = {2016},
	pages = {519--578},
}
